\documentclass[runningheads]{llncs}
\usepackage[T1]{fontenc}

\usepackage{graphicx}
\usepackage{amsmath}
\usepackage{proof} 
\usepackage{booktabs} 
\usepackage{amssymb} 
\usepackage{mathtools} 
\usepackage{tikz} 
\usetikzlibrary{automata}
\usetikzlibrary{shapes}
\usetikzlibrary{arrows,calc,fit}
\usetikzlibrary{snakes}
\usepackage{color,soul} 
\setul{0.5ex}{0.15ex}
\setulcolor{red}
\usepackage{comment} 
	\usepackage{hyperref} 
	\usepackage{enumerate} 

	\newcommand{\rNum}[1]{\expandafter{\romannumeral #1\relax}}
	\newcommand{\rNUM}[1]{\uppercase\expandafter{\romannumeral #1\relax}}

	\newcommand{\bff}[1]{\mathbf{#1}}
	\newcommand{\mbb}[1]{\mathbb{#1}}
	\newcommand{\mcl}[1]{\mathcal{#1}}
	\newcommand{\spc}[1]{\begin{spacing}{#1}}
		\newcommand{\spce}{\end{spacing}}
	
	\newcommand{\la}[0]{\langle}
	\newcommand{\ra}[0]{\rangle}
	\newcommand{\mfr}[1]{\mathfrak{#1}}

	\newcommand{\Romann}[1]{\MakeUppercase{\romannumeral #1}}
	
	\newcommand{\ConfV}[0]{\mathit{CVar}}
	\newcommand{\Conf}[0]{\bff{Conf}}
	\newcommand{\Prop}[0]{\bff{Prop}}
	\newcommand{\Prog}[0]{\bff{Prog}}

	\newcommand{\LDL}[0]{\mathit{DL}_p}
	\newcommand{\abort}[0]{\uparrow}
	\newcommand{\ter}[0]{\downarrow}
	\newcommand{\true}[0]{\mathit{true}}
	\newcommand{\false}[0]{\mathit{false}}
	\newcommand{\Fmla}[0]{\bff{Form}}
	\newcommand{\Pred}[0]{\bff{Prop}}
	\newcommand{\Sem}[1]{[\![ #1 ]\!]}
	\newcommand{\cnt}[0]{\mathit{C}}

	\newcommand{\trans}[0]{\longrightarrow}

	\newcommand{\Clo}[0]{\bff{Cl}}

	\newcommand{\GDL}[0]{\textit{GDL}}
	\newcommand{\Assn}[0]{\bff{Assn}}

	\newcommand{\suf}[0]{\preceq_s}
	\newcommand{\psuf}[0]{\prec_s}
	
	\newcommand{\psufr}[0]{\succ_s}
	\newcommand{\sufeq}[0]{\approx_s}
	\newcommand{\mult}[0]{\preceq_m}
	\newcommand{\pmult}[0]{\prec_m}
	\newcommand{\multr}[0]{\succeq_m}
	\newcommand{\pmultr}[0]{\succ_m}
	\newcommand{\multeq}[0]{{\approx_m}}

	\newcommand{\TA}[0]{\bff{TA}}
	\newcommand{\app}[0]{\mfr{I}}
	\newcommand{\free}[0]{\textit{free}}
	
	\newcommand{\DLF}[0]{\bff{DL_p}}
	
	\newcommand{\boolsem}[0]{{\mfr{T}}}
	\newcommand{\cfeq}[0]{=_{\textit{cf}}}
	\newcommand{\Eval}[0]{\bff{Eval}}

	\newcommand{\ClProg}[0]{{\Clo(\Prog)}}
	\newcommand{\ClConf}[0]{{\Clo(\Conf)}}
	\newcommand{\ClFmla}[0]{{\Clo(\Fmla)}}
	\newcommand{\pfDLp}[0]{{P_\textit{DLP}}}
	\newcommand{\Beha}[0]{\bff{Tran}}
	
	\newcommand{\Oper}[0]{\Lambda}
	\newcommand{\AFmla}[0]{\bff{F}}
	
	\newcommand{\FV}[0]{\textit{FV}}
	
	\newcommand{\CT}[0]{\textit{CT}}
	\newcommand{\Termi}[0]{\bff{Ter}}
	\newcommand{\termi}[0]{\Downarrow}
	\newcommand{\Terminate}[0]{\Omega}
	\newcommand{\Domain}[1]{{(\TA_{#1}, \Eval_{#1}, \app_{#1})}}
	\newcommand{\Extra}[0]{E}
	\newcommand{\SExtra}[0]{E_s}
	\DeclareMathOperator{\seq}{;}
	\DeclareMathOperator{\cho}{\cup}
	\newcommand{\WP}[0]{{\textit{WP}}}
	\newcommand{\E}[0]{E}
	\newcommand{\fodl}[0]{\textit{FO}}
	\newcommand{\se}[0]{\leadsto}
	\newcommand{\sepc}[0]{{*}}
	\newcommand{\sepi}[0]{{-\!*}}
	\newcommand{\dlFmla}[0]{\Fmla_{\textit{dl}}}
	\newcommand{\partto}[0]{\rightarrowtail}
	\newcommand{\Alloc}[0]{\bff{cons}}
	\newcommand{\disAlloc}[0]{\bff{disp}}
	\newcommand{\dom}[0]{\textit{dom}}
	\DeclareMathOperator{\disj}{\bot}
	\newcommand{\allocto}[0]{\dashrightarrow}
	\newcommand{\Sep}[0]{{\textit{SP}}}
	\newcommand{\std}[0]{\textit{std}}
	\newcommand{\BD}[0]{\textit{BD}}
	\newcommand{\Exp}[0]{\bff{Exp}}
	\newcommand{\Term}[0]{\TA}
	
	\newcommand{\g}[0]{g}
	\newcommand{\Sub}[0]{{\textit{Sub}}}
	
	\newcommand{\dddef}[0]{=_{df}}

	\newcommand{\Var}[0]{\mathit{Var}}
	\newcommand{\Etrap}[0]{\textit{trap}}
	\newcommand{\Eloop}[0]{\textit{loop}}
	\newcommand{\Eend}[0]{\textit{end}}
	\newcommand{\Eif}[0]{\textit{if}}
	\newcommand{\Ethen}[0]{\textit{then}}
	
	\newcommand{\Eemit}[0]{\textit{emit}}
	\newcommand{\Eexit}[0]{\textit{exit}}
	
	\newcommand{\Epause}[0]{\textit{pause}}
	\DeclareMathOperator{\Split}{|}
	\DeclareMathOperator{\suc}{\$}
	\newcommand{\Wwhile}[0]{\textit{while}}
	\newcommand{\Wdo}[0]{\textit{do}}
	\newcommand{\Wif}[0]{\textit{if}}
	\newcommand{\Wthen}[0]{\textit{then}}
	\newcommand{\Welse}[0]{\textit{else}}
	\newcommand{\Wend}[0]{\textit{end}}
	
\begin{document}
		\title{Parameterized Dynamic Logic --- Towards A Cyclic Logical Framework for General Program Specification and Verification}
		%
		%
		
		\author{Yuanrui Zhang\inst{1}\orcidID{0000-0002-0685-6905}}
	
	%
	
	\authorrunning{Y. Zhang}
	
	%
	
	\institute{Collage of Software, Nanjing University of Aeronautics and Astronautics, China\\
		\email{yuanruizhang@nuaa.edu.cn, zhangyrmath@126.com} \\
	}

	\maketitle              
	\begin{abstract}
		We present a theory of parameterized dynamic logic, namely $\LDL$, for specifying and reasoning about a rich set of program models based on their transitional behaviours. 
		Different from most dynamic logics that deal with regular expressions or a particular type of formalisms, $\LDL$ introduces a type of labels called ``program configurations'' as explicit program status for symbolic executions, allowing programs and formulas to be of arbitrary forms according to interested domains. 
		This characteristic empowers dynamic logical formulas with a direct support of symbolic-execution-based reasoning, while still maintaining reasoning based on syntactic structures in traditional dynamic logics through a rule-lifting process.   
		We propose a proof system and build a cyclic preproof structure special for $\LDL$, which guarantees the soundness of infinite proof trees induced by symbolically executing programs with explicit/implicit loop structures. 
		The soundness of $\LDL$ is formally analyzed and proved. 
		$\LDL$ provides a flexible verification framework based on the theories of dynamic logics.
		It helps reduce the burden of developing different dynamic-logic theories for different programs, and save the additional transformations in the derivations of non-compositional programs. 
		We give some examples of instantiations of $\LDL$ in particular domains, showing the potential and advantages of using $\LDL$ in practical usage.

	\end{abstract}

	\section{Introduction}
	\label{section:Summery}
	
	\subsection{Background and Motivations}
	
	Program verification based on theorem proving has become a hot topic recently, especially with the rapid development of AI-assisted technologies that can greatly enhance the automation of interactive proof processes~\cite{yang2024formalmathematicalreasoningnew}. 
	
	Dynamic logic~\cite{Harel00} is an extension of modal logic for reasoning about programs.
	Here, the term `program' can mean an abstract formal model, not necessarily an explicit computer program. 
	In dynamic logic, a program $\alpha$ is embedded into a modality $[\cdot]$ in a form of $[\alpha]\phi$, meaning that after all executions of $\alpha$, formula $\phi$ holds. 
	Formula $\phi\to [\alpha]\psi$ exactly captures partial correctness of programs expressed by triple $\{\phi\}\alpha\{\psi\}$ in Hoare logic~\cite{Hoare63}. 
	Essentially, 
	verification in dynamic logic is a process of continuously altering the form $[\alpha]\phi$ based on how $\alpha$ evolves according to its semantics. 
	This characteristic brings an advantage of clearly observing programs' behavioral changes throughout the whole verification processes. 
	With a dynamic form $[\alpha]\phi$ that mixes programs and formulas, 
	dynamic logic is able to express many complex program properties such as 
	$[\alpha]\la\beta\ra\phi$, expressing that after all executions of $\alpha$, there is an execution of $\beta$ after which $\phi$ holds (where $\la\cdot\ra$ is the dual operator of $[\cdot]$ with $\la\beta\ra\phi$ defined as $\neg[\beta]\neg\phi$). 
	As one of popular logical formalisms, dynamic logic 
	has been used for many types of programs, such as process algebras~\cite{Benevides10}, programming languages~\cite{Beckert2016}, synchronous systems~\cite{Zhang21,Zhang22}, hybrid systems~\cite{Platzer07b,Platzer18} and probabilistic systems~\cite{Pardo22,kozen85,Feldman84}. 
	Because of modality $\la\cdot\ra$, it is a natural candidate for unifying correctness and `incorrectness' reasoning recently proposed and developed in work e.g.~\cite{Hearn19,Zilberstain23}. 
	
	The programs of dynamic logics are usually (tailored) regular expressions with tests (e.g.~\cite{Fischer79,Platzer07b}) 
	or actual programming languages like Java~\cite{Beckert2016}. 
	The denotational semantics of these models is usually \emph{compositional}, in the sense that 
	the deductions of the logics mean to dissolve the syntactic structures of these models. 
	For example, in propositional dynamic logic (PDL)~\cite{Fischer79}, to prove a formula $[\alpha \cup \beta]\phi$, we prove both formulas $[\alpha]\phi$ and $[\beta]\phi$, where $\alpha$ and $\beta$ are sub-regular-expressions of $\alpha\cup \beta$.  
	This so-called `divide-and-conquer' approach to program verification has brought numerous benefits, the most significant being its ability to scale verifications by allowing parts of programs to be skipped (see an example in Section~\ref{section:Lifting Process in While Programs}). 
	
	However, this typical \emph{structure-based} reasoning heavily relies on programs' denotational semantics, thus brings two major drawbacks: 
	
	Firstly, dynamic logics are some special theories aiming for particular domains. 
	Thus for a new type of programs, one needs to either perform a 
	transformation from target programs into the program models of a dynamic logic so that existing inference rules can be utilized, or
	carefully design a set of particular rules to adapt the new programs' semantics. 
	The former usually means losses of partial program structural information during the transformations, while the latter often demands a large amount of work. 
	For example, Verifiable C~\cite{Appel_Dockins_Hobor_Beringer_Dodds_Stewart_Blazy_Leroy_2014}, a famous program verifier for C programming language based on Hoare logic, used nearly 40,000 lines of Coq code to define the logic theory. 
	
	Secondly but importantly as well, some programs are naturally non-compositional or do not temporally provide a compositional semantics at some level of system design. Typical examples are neural networks~\cite{Goodfellow16}, automata-based system models (e.g. a transition model of AADL~\cite{ZBYang14}) and some programming languages (as we will see soon). 
	For these models, one can only
	re-formalize their structures by additional transformation procedures so that structure-based reasoning is possible. 
	One example is synchronous imperative programming languages such as Esterel~\cite{Berry92} and Quartz~\cite{Gesell12}. 
	In a synchronous parallel program $\alpha \parallel \beta$, 
	the executions of $\alpha$ and $\beta$ cannot be reasoned independently because the concurrent executions at each instance are generally not true interleaving. 
	To maintain the consistency of synchrony, events in one instance have to be arranged in a special order, which makes symbolic executions of  $\alpha \parallel \beta$ as a whole a necessary step in program derivations. 
	We will see an example in Appendix~\ref{section:Example Two: A Synchronous Loop Program}. 

	In this paper, to compensate for the above two shortcomings in existing dynamic-logic theories, we present a theory of parameterized dynamic logic ($\LDL$). 
	$\LDL$ supports a general approach for specification and verification of computer programs and system models. 
	On one hand, $\LDL$ offers an abstract-leveled setting for describing a broad set of programs and formulas, and 
	is empowered to
	support a \emph{symbolic-execution-based} reasoning based on a set of program transitions, called ``program behaviours''. This reduces the burden of developing different dynamic-logic theories for different programs, and saves the additional transformations in the derivations of non-compositional programs. 
	On the other hand, $\LDL$'s forms are still compatible with  the existing structural rules of dynamic logics through a lifting process. 
	This makes $\LDL$ subsume (at least partial) existing dynamic-logic theories into a single logical framework.

	\subsection{Illustration of Main Idea}
	Informally, in $\LDL$ we treat dynamic formula $[\alpha]\phi$ as a `parameterized' one in which program $\alpha$ and formula $\phi$ can be of arbitrary forms, provided with a structure  $\sigma$ called ``program configuration'' to record current program status for programs' symbolic executions. 
	So, a $\LDL$ dynamic formula is of the form: $\sigma : [\alpha]\phi$, following a convention of labeled formulas. 
	It expresses the same meaning as $[\alpha]\phi$ in traditional dynamic logics, except that program status is shown explicitly alongside $[\alpha]\phi$. 
	
	\ifx
	Informally, $\LDL$ is based on the typical forms: $[\alpha]\phi$ with modality $[\cdot]$, but 
	`parameterizes' the programs $\alpha$ and formulas $\phi$, and 
	adds a parameter: \emph{program configurations} $\sigma$. 
	Therefore, a $\LDL$ dynamic formula is of the \emph{labelled} form: $\sigma : [\alpha]\phi$, which expresses the same meaning as $[\alpha]\phi$ in dynamic logics, except that it 
	assumes $\alpha, \phi$ and $\sigma$ to be of arbitrary forms. 
	As we will see, the extra label $\sigma$ plays the critical role 
	of recording the context of a program execution. 
	\fi

	To see how $\LDL$ formulas are powerful for supporting both symbolic-execution-based and structure-based reasoning, consider a simple example. 
	We prove a formula $\phi_1 \dddef (x\ge 0\to [x := x + 1]x > 0)$ in first-order dynamic logic~\cite{Harel79} (FODL), 
	where $x$ is a variable ranging over the set of integer numbers.  
	Intuitively, $\phi_1$ means that if $x\ge 0$ holds, then $x>0$ holds after the execution of the assignment $x := x + 1$. 
	In FODL, to derive formula $\phi_1$, 
	we apply the structural rule: 
	$$
	\begin{aligned}
		\label{equ:assignRule}
		\infer[^{(x := e)}]
		{[x := e]\phi}
		{\phi[x / e]}
	\end{aligned}
	$$
	for assignment
	on formula $[x:=x+1]x > 0$ by substituting $x$ of $x > 0$ with $x + 1$, and obtain $x + 1 > 0$. 
	Formula  $\phi_1$ thus becomes $\phi'_1\dddef (x\ge 0 \to x + 1 > 0)$, which is true for any integer number $x\in \mbb{Z}$. 
	
	While in $\LDL$, formula $\phi_1$ can be expressed as a form:  $\psi_1\dddef (t\ge 0\to \{x\mapsto t\} : [x := x + 1]x > 0)$, where formula $[x := x + 1]x > 0$ is labeled by
	configuration $\{x\mapsto t\}$, capturing the current program status with $t$ a free variable, meaning 
	``variable $x$ has value $t$''. 
	This form may seem tedious at first sight. 
	But one soon can find out that with a configuration explicitly showing up, to derive formula $\psi_1$, we no longer need rule $(x:=e)$, but instead can directly perform a program transition of $x := x + 1$ as: $(x := x + 1, \{x\mapsto t\})\trans (\ter, \{x\mapsto t + 1\})$~\footnote{note that to distinguish `$\trans$' and `$\to$'}, which assign $x$'s value with its current value added by $1$. 
	Here $\ter$ indicates a program termination. 
	Formula $\psi_1$ thus becomes $\psi'_1\dddef (t\ge 0\to \{x\mapsto t + 1\} : [\ter] x > 0)$, which is 
	actually $\psi''_1\dddef (t\ge 0\to \{x\mapsto t + 1\} : x > 0)$ by eliminating the dynamic part `$[\ter]$' since $\ter$ behaves nothing. 
	By applying configuration $\{x\mapsto t + 1\}$ on formula $x > 0$ (which means replacing variable $x$ of $x > 0$ with its value $t+1$ in $\{x\mapsto t + 1\}$), 
	we obtain the same valid formula $\phi'_1$: $t\ge 0\to t + 1 > 0$ (modulo the free variables $x, t$). 
	
	The above derivations of $\LDL$ formulas benefit from that for many programs,  transitions like $(x := x + 1, \{x\mapsto t\})\trans (\ter, \{x\mapsto t + 1\})$ are natural from their operational semantics and can be easily formalized (as will be seen in Appendix~\ref{section:Formal Definitions of While Programs}), while 
	structural rules like $(x := e)$ would consume more efforts to design. And there might exists no structural rules in certain cases, as will be illustrated in the example of Esterel programs in Appendix~\ref{section:Example Two: A Synchronous Loop Program}. 
	
	With the labels $\sigma$ structurally independent from their dynamic formulas $[\alpha]\phi$, $\LDL$ formulas $\sigma : [\alpha]\phi$ also support programs' structural rules through a so-called ``lifting process'' (Section~\ref{section:Lifting Process From Program Domains}). 
	This capability provides $\LDL$ with a flexible framework in which different 
	inference rules can be applied to make a trade-off between  symbolic-execution-based and structure-based reasoning in practical deduction processes. 
	For example, here, to derive the formula $\psi$ above, one can also apply a lifted rule from rule $(x := e)$: 
	$$\begin{aligned} \infer[]
		{\{x\mapsto t\} : [x := e]\phi}
		{\{x\mapsto t\} : \phi[x / e]}\end{aligned}$$ 
	by fixing $\{x\mapsto t\}$, 
	and then $\psi_1$ becomes $\psi'''_1\dddef (t\ge 0)\to (\{x\mapsto t\} : x + 1 > 0)$, which is also formula $\phi'_1$: $t\ge 0\to t + 1 > 0$ after applying configuration $\{x\mapsto t\}$ on formula $x + 1 > 0$. 
	
	\ifx
	$\LDL$ formulas benefit from that no structural rules special for a target program are needed, but only a program operational semantics and a definition of how a configuration applies to a non-dynamical formula. 
	Normally, for existing programs and system models, formalizing a program's operational semantics (like $(x := x + 1, \{x\mapsto t\})\trans (\ter, \{x\mapsto t + 1\})$) is easier and more direct than designing an inference rule (like rule $(x := e)$).
	As will be illustrated in Section~\ref{section:Example Two: A Synchronous Loop Program}, in some case, for a target program with a certain structure, there might exist no structural rules, or, designing such a rule can be complex. 
	\fi
	
	\ifx
	On the other hand, 
	with the labelled configurations structurally independent from their dynamic formulas, 
	$\LDL$  formulas are still compatible with the syntactic-structure-based rules in specific domains. 
	One can still apply these rules just as before by keeping the labelled parts unchanged. 
	For this example, we can apply a ``lifting version'' of rule (\ref{equ:assignRule}): $\begin{aligned} \infer[]
		{\sigma : [x := e]\psi}
		{\sigma : \psi[x / e]}\end{aligned}$, and $\psi$ then becomes $\psi''\dddef (\{x\mapsto t\} : x\ge 0)\to (\{x\mapsto t\} : x + 1 > 0)$, which also can be transformed into formula $\phi'$: $t\ge 0\to t + 1 > 0$ after the applications of $\{x\mapsto t\}$. 
	This ``lifting ability'' allows $\LDL$ to subsume special dynamic logic theories and 
	provides a flexible framework in which 
	different inference rules can be applied to make a trade-off between 
	compositional reasoning and symbolic executions in practical deduction processes. 
	\fi
	
	\subsection{Main Contributions}
	
	In this paper, different from the traditional approach that relies on Kripke structures (cf.~\cite{Harel00}), we give the semantics of $\LDL$ directly based on the program behaviours of target programs. 
	
	After building the logic, we propose a proof system for $\LDL$, providing a set of rules for reasoning about general programs based on program behaviours. 
	To provide compatibilities of $\LDL$ with existing dynamic-logic theories, 
	we propose a lifting process for introducing additional structural rules from particular program domains. 
	
	Unlike the derivation processes based on dissolving program structures, deriving a $\LDL$ formula may lead to an infinite proof tree. 
	We develop a cyclic proof (cf.~\cite{Brotherston07}) approach for $\LDL$.
	Generally, cyclic proof is a technique to ensure that a certain infinite proof tree, called `preproof', can lead to a valid conclusion if it is `cyclic', that is, containing a certain type of derivation traces (Definition~\ref{def:Progressive Step/Progressive Derivation Trace}). 
	Cyclic proof system provides a solution for reasoning about programs with implicit loop structures and well supports the symbolic-execution-based reasoning of $\LDL$ formulas. 
	We propose a ``cyclic preproof'' structure for $\LDL$ (Definition~\ref{def:Cyclic Preproof})
	and prove that it is sound under certain conditions.
	
	\ifx
	In this paper, 
	to give the semantics of $\LDL$, we follow the way of defining a traditional dynamic logic by building a special Kripke structure based on the operational semantics of target programs. 
	Our way follows the way of defining a labelled logical formula, as in~\cite{Docherty19}. 
	After building  $\LDL$, we propose a proof system for $\LDL$, providing a set of rules for reasoning arbitrary programs. 
	Unlike the derivation processes based on dissolving program structures, the derivation process of a $\LDL$ formula may be infinite. 
	To solve this problem, 
	we adopt the cyclic proof approach (cf.~\cite{Brotherston07}), a technique to ensure that a certain proof tree, called \emph{preproof}, can lead to a valid conclusion if it is `cyclic', that is, containing a certain type of derivation traces (Definitions~\ref{def:Progressive Step/Progressive Derivation Trace}).    
	Cyclic proof system provides a support for reasoning about programs with implicit loop structures and is well-suited for symbolic-execution-based reasoning of $\LDL$ formulas. 
	We propose a ``cyclic preproof structure'' for $\LDL$ (Definition~\ref{def:Cyclic Preproof})
	and prove that it is sound under certain conditions. 
	\fi

	To summarize, our contributions are mainly three folds:
	\begin{itemize}
		\item We give the syntax and semantics of $\LDL$ formulas.
		\item We build a proof system for $\LDL$.
		\item We construct a cyclic preproof structure for $\LDL$ and prove its soundness. 
	\end{itemize}
	
	The rest of the paper is organized as follows. 
	Section~\ref{section:Dynamic Logic LDL} defines the syntax and semantics of $\LDL$ formulas. 
	In Section~\ref{section:A Cyclic Proof System for LDL}, we propose a cyclic proof system for $\LDL$. 
	In Section~\ref{section:Case Studies}, we show how $\LDL$ can be useful as a verification framework by giving some examples.  
	Section~\ref{section:Proof of Theorem - theo:Soundness of A Cyclic Preproof}
	analyzes and proves the soundness of $\LDL$. 
	Section~\ref{section:Related Work} introduces some previous work that is closely related to ours, while Section~\ref{section:Discussions and Future Work} 
	makes some discussions and talks about future work. 
	
	\section{Dynamic Logic $\LDL$}
	\label{section:Dynamic Logic LDL}
	
	
	In this section, we mainly define the syntax and semantics of $\LDL$. 
	What mostly distinguishes $\LDL$ from other dynamic logics is that its programs and formulas are not in particular forms, but instead can be of any forms, only provided that some restrictions (see Definition~\ref{def:Program Properties}) are respected. 
	Program configurations are introduced in $\LDL$ as explicit program status for describing program behaviours, enabling reasoning directly through symbolic executions. 
	In Section~\ref{section:Programs and Configurations}, we introduce the basic settings for the ingredients of $\LDL$. They are prerequisites for formally defining $\LDL$. 
	In Section~\ref{section:Examples of Term Structures} we introduce two examples of instantiations of program structures of $\LDL$ which will be used throughout this paper to illustrate our ideas. 
	Section~\ref{section:Program Behaviours} introduces the notion of program behaviours and its restriction under which this current work is concerned. 
	In Section~\ref{section:Syntax and Semantics of LDL}, we formally define the syntax and semantics of $\LDL$.

	
	
	
	
	
	\ifx
	The theory of $\LDL$ is based on a pre-defined set of terms over a countable set $\Sigma$ of signatures and a countable set of variables $\Var$. 
	In the following  subsections we will firstly define terms, based on which 
	we will introduce the concepts of programs, configurations, formulas and program behaviours. 
	They form the basic ingredients of the theory of $\LDL$. 
	\fi

	\subsection{Programs, Configurations and Formulas}
	\label{section:Programs and Configurations}
	The construction of $\LDL$ is based on a set $\TA$ of \emph{terms} defined over a set $\Var$ of variables and a set $\Sigma$ of signatures. 
	In $\TA$, we distinguish three subsets 
	$\Prog, \Conf$ and $\Fmla$, representing the sets of programs, configurations and formulas respectively. 
	$\TA \supseteq \Prog\cup \Conf\cup \Fmla$. 
	We use $\equiv$ to represent the identical equivalence between two terms in $\TA$. 
	A function $f : \TA \to \TA$ is called \emph{structural} (w.r.t. $\Sigma$) if it satisfies that for any $n-$ary signature $s_n\in \Sigma$ ($n\ge 0$, $0-$ary signature is a constant) and term $s_n(t_1,...,t_n)$ with $t_1,...,t_n\in \TA$, $f(s_n(t_1,...,t_n)) \equiv s_n(f_1(t_1),...,f_n(t_n))$ for some structural functions $f_1,...,f_n : \TA\to \TA$. 
	Naturally, we assume that a structural function always maps a term to a term with the same type: $\Prog, \Conf$ and $\Fmla$. 

	\ifx
	In $\Fmla$, 
	we distinguish a type of formulas as ``program transitions'', denoted by $\Beha$, and a type of formulas as ``program terminations'', denoted by $\Termi$. 
	Thus $\Beha\subseteq \Fmla, \Termi\subseteq \Fmla$. 
	They will be introduced later in Section~\ref{section:Program Behaviours}. 
	\fi
	
	\ifx
	For $A\in \{\Prog, \Conf, \Fmla\}$, we use $\equiv$ to represent the identical equivalence between two terms of $A$. 
	We distinguish a subset $\Clo(A)$ of \emph{closed terms} of $A$. 
	As we will see, closed terms correspond to determined boolean semantics in our discussions. 
	A closed term is usually, though not necessarily in our setting, a term without \emph{free variables} --- a notion whose definition depends on explicit structures of $A$. 
	Non-closed terms in $A\setminus \Clo(A)$ are also called \emph{open}. 
	An \emph{evaluation} $\rho : A\to \Clo(A)$ maps each term of $A$ to a closed term of $\Clo(A)$. 
	We write as $\Eval$ the set of all evaluations. 
	\fi

	In $\Prog$, there is one distinguished program $\ter$, called \emph{termination}. 
	It indicates a completion of executions of a program. 
	
	Program configurations have impacts on formulas. 
	We assume that $\Conf$ is associated with a function
	$\app : (\Conf\times \Fmla) \to \Fmla$, called \emph{configuration interpretation}, that returns a formula by applying a configuration on a formula. 
	\ifx
	$\app$ is \emph{closed} under closed terms, in the sense that 
	for any $\sigma\in \Conf$ and $\phi\in \ClFmla$, $\app(\sigma, \phi) \in\ClFmla$. 
	\fi

	An \emph{evaluation} $\rho : \TA\to \TA$ is a structural function that maps each formula $\phi\in \Fmla$ to a \emph{proposition} --- a special formula that has a boolean semantics of either truth (represented as $1$) or falsehood (represented as $0$). 
	We denote the set of all propositions as $\Prop$. $\Prop\subseteq \Fmla$. 
	The boolean semantics of formulas is expressed by function $\boolsem : \Prop\to \{0, 1\}$. 
	$\Prop$ forms the semantical domain as the basis for defining $\LDL$. 
	An evaluation $\rho\in \Eval$ \emph{satisfies} a formula $\phi\in \Fmla$, denoted by 
	$\rho\models_\boolsem \phi$, if $\boolsem(\rho(\phi)) = 1$. 
	A formula $\phi\in \Fmla$ is \emph{valid}, denoted by 
	$\models_\boolsem \phi$, if $\rho\models_\boolsem \phi$ for all $\rho\in \Eval$.

	\ifx
	The subset $\Clo(\Fmla)$ of $\Fmla$ is called \emph{propositions}, which has a boolean semantics $\boolsem: \Clo(\Fmla)\to \{0,1\}$ that interprets each term to either true ($1$) or false ($0$).  
	In $\ClFmla$, we distinguish two propositions $\true$ and $\false$, corresponding to the boolean true and false respectively. 
	An evaluation $\rho\in \Eval$ \emph{satisfies} a formula $\phi\in \Fmla$, denoted by 
	$\rho\models_\boolsem \phi$, is defined if $\boolsem(\rho(\phi)) = 1$. 
	A formula $\phi\in \Fmla$ is \emph{valid}, denoted by 
	$\models_\boolsem \phi$, if $\rho\models_\boolsem \phi$ for any $\rho\in \Eval$. 
	\fi
	
	In the rest of this paper, we call \emph{program domain} a special instantiation of structure $\Domain{}$ for terms $\TA$, evaluations $\Eval$ and functions $\app$. 
	
	\ifx
	In the rest of this paper, let $\TA = \{\Prog, \Conf, \Fmla, \rho, \app\}$ called a \emph{program domain}. 
	Let $\TA = \Prog\cup \Conf\cup \Fmla$ be the set of all terms. 
	We use $A_N$ to denote an instantiated structure or function of $A$ with name $N$, for example, $\TA_N, \rho_N, \Conf_N$, etc. 
	\fi
	
	
	\subsection{Examples of Program Domains}
	\label{section:Examples of Term Structures}
	We consider two instantiated program domains, namely $\Domain{\WP}$ and $\Domain{\E}$. 
	Table~\ref{table:An Example of Program Structures} depicts two program models in these domains respectively: a \emph{While} program $\textit{WP}\in \Prog_{\textit{WP}}$ and an Esterel program $\textit{E}\in \Prog_E$. 
	Below we only give an informal explanation of relevant structures, but at the same time providing comprehensive formal definitions in Appendix~\ref{section:Formal Definitions of While Programs} for \emph{While} programs, in case readers want to know more details.  
	
	In these two program domains, the sets of variables are denoted by $\Var_{\textit{WP}}$ and $\Var_E$ respectively. 
	In program $\textit{WP}$, variables are $s, n, x$, whose domain is the set of integer numbers $\mbb{Z}$. 
	In program $\textit{E}$, variables are $x$ and $S$. 
	$x$ is also called a ``local variable'', with $\mbb{Z}$ as its domain. 
	Signal $S$ is a special variable distinguished from local variables, whose domain is $\mbb{Z}\cup \{\bot\}$, with $\bot$ to express `not-happened' of a signal (also called `absence').   
	Intuitively, assignment $x := e$ means assigning the value of an expression $e$ to a variable $x$. 
	Signal emission $\Eemit\ S(e)$ means assigning the value of an expression $e$ to a signal variable $S$ and broadcasts the value. 
	A mapping $x\mapsto e$ of a configuration means that variable $x$ has the value of expression $e$. 
	Formulas are first-order logical formulas in the domain $\mbb{Z}$ or $\mbb{Z}\cup\{\bot\}$ with arithmetical expressions. 
	Examples are $n > 0$, $\forall x. x + y = 0$, $x > 0\wedge y > 0$, etc.

	A configuration of a \emph{While} program is a variable storage in which each variable has a unique value of an expression. 
	For example, a configuration $\{x\mapsto 5, y\mapsto 1\}$ has variable $x$ storing value $5$ and variable $y$ storing value $1$.
	On the other hand, 
	a configuration of an Esterel program is a variable stack, allowing several (local) variables with the same name. For example, $\{x\mapsto 5 \Split y\mapsto 1 \Split y \mapsto 2\}$ 
	represents a configuration in which there are 3 variables: $x$, and two $y$s storing different values.
	The separator `$\Split$' instead of `$,$' is used to distinguish from a configuration of a \emph{While} program, with the right-most side as the top of the stack. 
	\ifx
	For simplicity, we can understand that configurations are expressed by $n-$ary signatures in $\Sigma$. 
	In reality, however, a configuration is usually implemented by binary signatures in an inductive way. 
	For example, a configuration $\{x_1\mapsto e_1,...,x_n\mapsto e_n\}$ 
	can be expressed as a list structure defined as: $\textit{Cons}\ (\textit{Cons}\ (...(\textit{Cons}\ \textit{nil}\ x_1\mapsto e_1)...)\ x_{n-1}\mapsto e_{n-1})\ x_n\mapsto e_n$, with $\textit{nil}\in \Sigma$ representing an empty list and $(\textit{Cons}\ \_\ \_)\in \Sigma$ combining a list with a list element. 
	\fi
	
	``Free variables'' are those occurrences of variables 
	that are not in the effect scope of any binders, such as the assignment $x:= (\cdot)$, or quantifiers $\forall x. (\cdot), \exists x. (\cdot)$. 
	Let $\FV_N: \TA_N\to \Var_{N}$ ($N\in \{\textit{WP}, E\}$) return the set of all free variables of a term. 
	For example, in program $\textit{WP}$, variable $n$ of expression $n - 1$ in $n := n - 1$ and variable $s$ of expression $s + n$ in $s := s + n$ are bound by the assignments: $n := (\cdot)$ and $s := (\cdot)$ respectively, while variable $n$ is free in formula $n > 0$ and in expression $s + n$ of assignment $s := s + n$.  
	So $\FV_{\textit{WP}}(\textit{WP}) = \{n\}$. 
	In a configuration, while mapping $x\mapsto \cdot$ is a binder for labeled formulas (as defined in Appendix~\ref{section:Formal Definitions of While Programs}), 
	a variable in an expression $e$ of a mapping $x\mapsto e$ is always free, since interpretation $\app_\WP$ introduced below is simply a type of substitutions.  
	\ifx
	In a configuration, mapping $x\mapsto \cdot$ is a binder. 
	A variable in an expression $e$ of a mapping $x\mapsto e$ is free w.r.t. a configuration $\sigma$ if it does not appear on the left side of any mapping in $\sigma$. 
	For example, variable $y$ in $y + 1$ of mapping $x\mapsto y + 1$ of a configuration $\{x\mapsto y + 1, y\mapsto v\}$ is bound by mapping $y\mapsto v$, while $v$ is free w.r.t. $\{x\mapsto y + 1, y\mapsto v\}$. 
	In \emph{While} and Esterel programs, for a configuration $\sigma$, we require that all variables in an expression $e$ of a mapping $x\mapsto e$ must be free w.r.t. $\sigma$. 
	This avoids confusing mapping such as $x\mapsto x + 1$. 
	For example, the configuration $\{x\mapsto y + 1, y\mapsto v\}$ above is illegal, instead one can write it as $\{x\mapsto v + 1, y\mapsto v\}$. 
	\fi
	
	\ifx
	In the context of this example, 
	closed terms are terms without ``free variables''.
	Free variables are those occurrences of variables 
	that are not in the effect scope of any binders, such as the assignment $x:= (\cdot)$, or quantifiers $\forall x. (\cdot), \exists x. (\cdot)$. 
	Let $\FV_N: \TA_N\to \Var_{N}$ ($N\in \{\textit{WP}, E\}$) return the set of all free variables of a term. 
	For example, in program $\textit{WP}$, variable $n$ of expression $n - 1$ in $n := n - 1$ and variable $s$ of expression $s + n$ in $s := s + n$ are bound by the assignments: $n := (\cdot)$ and $s := (\cdot)$ respectively, while variable $n$ is free in formula $n > 0$ and in expression $s + n$ of assignment $s := s + n$.  
	So $\textit{WP}$ is not a closed term, with $\FV_{\textit{WP}}(\textit{WP}) = \{n\}$. 
	Similarly, in a configuration $\{x\mapsto v + x + 1, y\mapsto x + 1\}$, 
	variable $v$ is free, while variable $x$ of expressions $v + x + 1$ and $x + 1$ are bound by the mapping $x \mapsto (\cdot)$. 
	\fi
	
	An evaluation $\rho(t)$ ($\rho\in \Eval_\WP$ or $\rho\in \Eval_\E$) returns a ``closed term'' (terms without free variables) by replacing each free variable of term $t$ with a value in $\mbb{Z}$ or $\mbb{Z}\cup \{\bot\}$. 
	For example, if $\rho(n) = 5$, then $\rho(n > 0) = (5 > 0)$. 
	$5 > 0$ is a proposition with a boolean semantics: $\boolsem_N(5 > 0) = 1$ in the theory of integer numbers. 
	A substitution $t[e/x]$ replaces each free variable $x$ of $t$ with term $e$. 
	
	
	Given a configuration $\sigma$ and a formula $\phi$, 
	in a \emph{While} program, 
	the interpretation $\app_{\textit{WP}}(\sigma, \phi)$ is defined as a type of substitutions (see Appendix~\ref{section:Formal Definitions of While Programs}).
	It replaces each free variable $x$ of $\phi$ with its value stored in $\sigma$. 
	For example, $\app_{\textit{WP}}(\{n\mapsto 1, s\mapsto 0\}, n > 0) = (1 > 0)$.  
	In an Esterel program, on the other hand, 
	the interpretation $\app_{E}(\sigma, \phi)$ replaces each free variable $x$ of $\phi$ with the value of the top-most variable $x$ in $\sigma$. 
	For instance, $\app_{E}(\{n \mapsto 1\ |\ n\mapsto 2\ |\ s \mapsto 0\}, n > 0) = (2 > 0)$ (by taking the value $2$ of the right $n$). 
	

	\begin{table}[tb]
		\begin{center}
			\noindent\makebox[\textwidth]{%
				\scalebox{0.9}{
					\begin{tabular}{c}
						\toprule
						\multicolumn{1}{l}{A \textit{While} program:}
						\\
						\begin{tabular}{r l}
							$\textit{WP}$
							&
							$
							\begin{aligned}
								\dddef&
								\{
								\textit{while}\
								(n > 0)\
								\textit{do}\
								s := s + n\ ;\
								n := n - 1\
								\textit{end}\
								\},
							\end{aligned}
							$
							\\
							\mbox{Configuration:}
							&
							$ 
							\{x_1\mapsto e_1, x_2\mapsto e_2,...,x_n\mapsto e_n\}, \ \ (n\ge 0)$
						\end{tabular}
						\\
						\midrule
						\midrule
						\multicolumn{1}{l}{An Esterel program:}
						\\
						\begin{tabular}{r l}
							$E$
							&
							$\dddef \{\Etrap\ A \parallel B\ \Eend\}, $
							where\\
							\multicolumn{2}{c}{
								$
								\begin{aligned}
									A \dddef&\ 
									\{\Eloop\ (\Eemit\ S(0)\ ;\ x := x - S\ ;\ \Eif\ x = 0\ \Ethen\ \Eexit\ \Eend\ ;\ \Epause)\ \Eend\},\\
									B \dddef&\ 
									\{\Eloop\ (\Eemit\ S(1)\ ;\ \Epause)\ \Eend\}
								\end{aligned}
								$
							}
							\\
							Configuration:
							&
							$
							\{x_1\mapsto e_1\ |\ x_2\mapsto e_2\ |\ ...\ |\ x_n\mapsto e_n\},\ \ (n\ge 0)
							$
						\end{tabular}
						\\
						\bottomrule
					\end{tabular}
				}
			}
		\end{center}
		\caption{An Example of Program Domains}
		\label{table:An Example of Program Structures}
	\end{table}
	
	\ifx
	\subsection{Examples of Term Structures}
	\label{section:Examples of Term Structures}
	
	We use an example of program domains throughout this paper to illustrate our ideas. 
	Note that below we only explain how these terms are expressed using signatures in our theory. 
	We do not show how to construct these terms formally in general, 
	because this actually depends on the system models or programming languages we choose. 
	But most importantly, it does not really matter because the theory of $\LDL$ is independent from explicit structures. 
	For some notions in the example, we only give informal explanations, as to rigorously define them requires explicit formal definitions of relevant structures. 
	
	
	Table~\ref{table:An Example of Program Structures} depicts the structures of two program models: a \emph{While} program $\textit{WP}$ and an Esterel program $\textit{E}$. 
	All terms listed are variables, signatures or compositions of signatures from pre-defined sets $\Sigma$ and $\Var$. 
	Set $\textit{Nat}$ of natural numbers, as usual, can be defined based on a constant $0\in \Sigma$ using a signature $\suc$ called `successor', in an inductive way: 
	(1) $0\in \textit{Nat}$; (2) $\suc t\in \textit{Nat}$ if $t\in \textit{Nat}$. 
	Signal $S$ is a type special variables distinguished from variable $x$.  
	$(N := \_)$, $(N\mapsto \_)$ and $\Eemit\ N(\_)$ are $1-$ary signatures, with $N$ a special constant not a variable. 
	Intuitively, assignment $x := e$ means assigning an expression $e$ to a variable $x$. 
	Signal emission $\Eemit\ S(e)$ means assigning an expression $e$ to a signal variable $S$ and broadcasts the value of $e$. 
	A mapping $x\mapsto e$ of a configuration means that variable $x$ has the value of expression $e$. 
	$n > 0$ is a formula. 
	
	A configuration of a \emph{While} program is each variable storage in which a variable has a unique value of an expression, for example a configuration $\{x\mapsto 5, y\mapsto 1\}$, which maps variable $x$ to $5$ and $y$ to $1$.
	On the other hand, 
	a configuration of an Esterel program is a variable stack, allowing several (local) variables with the same name. For example, $\{x\mapsto 5 \Split y\mapsto 1 \Split y \mapsto 2\}$ represents a configuration in which there are 3 variables: $x$, and two $y$s with different values.
	``$\Split$'' is used instead of ``$,$'' to remind that it is a stack, with the right-most side as the top of the stack. 
	For simplicity, we can understand that configurations are expressed by $n-$ary signatures in $\Sigma$. 
	In reality, however, a configuration is usually implemented by binary signatures in an inductive way. 
	For example, a configuration $\{x_1\mapsto e_1,...,x_n\mapsto e_n\}$ 
	can be expressed as a list structure defined as: $\textit{Cons}\ (\textit{Cons}\ (...(\textit{Cons}\ \textit{nil}\ x_1\mapsto e_1)...)\ x_{n-1}\mapsto e_{n-1})\ x_n\mapsto e_n$, with $\textit{nil}\in \Sigma$ representing an empty list and $(\textit{Cons}\ \_\ \_)\in \Sigma$ combining a list with a list element. 

	In the context of this example, informally, ``free variables'' are those 
	whose values do not rely on any program statements or configuration mappings and are not bound by any logical quantifiers like $\forall, \exists$. 
	For example, in program $\textit{WP}$, variable $n_{(1)}$ is free in formula $n_{(\textit{WP}, 1)} > 0$, while variable $n_{(1)}$ in assignment $n := n_{(\textit{WP}, 3)} - 1$ bound by the assignment $(n := \_)$ itself. 
	So $\textit{WP}$ is not closed, with $\FV(\textit{WP}) = \{n_{(\textit{WP}, 1)}, n_{(\textit{WP}, 2)}\}$. 
	Similarly, in a configuration $\{x\mapsto v_{(1)} + x_{(1)} + 1, y\mapsto x_{(2)} + 1\}$, 
	variable $v_{(1)}$ is free, while variables $x_{(1)}$ and $x_{(2)}$ are bound by the mapping $(x \mapsto \_)$.  
	
	Given a configuration $\sigma$ and a formula $\phi$, 
	in a \emph{While} program, 
	informally, 
	the interpretation $\app(\sigma, \phi)$ replaces each free variable $x$ of $\phi$ with its value stored in $\sigma$. 
	For example, $\app(\{n\mapsto 1, s\mapsto 0\}, n > 0) = (1 > 0)$, 
	$\app(\{n\mapsto s, s\mapsto 0\}, n > 0) = (0 > 0)$ (with bound variable $s$ firstly evaluated as $0$ in mapping $n\mapsto s$). 
	In an Esterel program, on the other hand, 
	the interpretation $\app(\sigma, \phi)$ replaces each free variable $x$ of $\phi$ with the value of the top-most variable $x$ in $\sigma$. 
	For instance, $\app(\{n \mapsto 1\ |\ n\mapsto 2\ |\ s \mapsto 1\}, n > 0) = (2 > 0)$ (by taking the value $2$ of the right $n$). 
	The above $(0 > 0)$ and $(2 > 0)$ are propositions (as closed formulas) with a boolean semantics: $\boolsem(0 > 0) = 0$ and $\boolsem(2 > 0) = 1$ in the theory of natural numbers.  
	
	From the example of Table~\ref{table:An Example of Program Structures}, 
	we can see that the notions of free variables and configuration interpretations can vary in different structures. 

	\ifx
	Table~\ref{table:An Example of Program Structures} depicts a \emph{while} program as an example of particular structures of $\Prog, \Conf$ and $\Oper$. 
	Given an initial value of variables $n$, program $\textit{WP}$ computes the sum from $n$ to $1$ stored in variable $s$. 
	The configurations of $\textit{WP}$ are defined as variable storage that maps a variable to a value. For example, $\{n \mapsto 1, s\mapsto 0\}$ denotes a configuration that maps $n$ to $1$ and $s$ to $0$. 
	The interpretation $\app(\sigma)$ of a configuration $\sigma$ is in the usual sense. 
	For example, $\app(\{n\mapsto 1, s\mapsto 0\})$ is the mapping that maps $n, s$ to $1, 0$ respectively, and maps other variables to themselves. 
	We have that, for instance, $\app(\{n\mapsto 1, s\mapsto 0\})(n\ge 0) = 1\ge 0$ by  assigning $n$ to $1$ in $n\ge 0$. 
	

	The lower part of the table shows several rule schemata for \emph{while} programs, which are related to the derivations of $\textit{WP}$ in Section~\ref{section:Example One: A While Program}. 
	Note that these rule schemata are independent from the theory of $\LDL$ and we only assume the existence of the set $\Oper$ of program transitions in our later discusses. 
	
	In Section~\ref{section:Example Two: A Synchronous Loop Program} we will see another example in Esterel language~\cite{Berry92} where 
	a program configuration is not a storage, but a stack. 
	\fi

	\ifx
	\textit{While} programs and Esterel programs have different syntax and configurations. 
	The operational semantics of \textit{while} programs is given by a finite set of rule schemata. 
	While we do not give the operational semantics of Esterel programs here due to its complexity. 
	One can refer to~\cite{Butucaru10} for more details. 
	\fi
	
	\ifx
	By its operational semantics, \textit{while} programs are deterministic programs. 
	And it is not hard to see that it satisfies the program properties in Definition~\ref{def:Program Properties}. 
	Esterel programs are also deterministic (cf.~\cite{Berry92}). 
	\fi


	\fi
	
	\subsection{Program Behaviours}
	\label{section:Program Behaviours}
	
	\textbf{Program Behaviours}.
	A \emph{program state} is a pair $(\alpha, \sigma)\in \Prog\times \Conf$ of programs and configurations. 
	In $\Fmla$, 
	we assume
	a binary relation $(\alpha, \sigma)\trans (\alpha', \sigma')$ called a \emph{program transition} between two program states $(\alpha, \sigma)$ and $(\alpha', \sigma')$.
	For any evaluation $\rho\in \Eval$, 
	proposition $\rho((\alpha, \sigma)\trans (\alpha', \sigma'))$ is assumed to be defined as: $\rho((\alpha, \sigma)\trans (\alpha', \sigma')) \dddef \rho(\alpha, \sigma)\trans \rho(\alpha', \sigma')$. 
	\emph{Program behaviours}, denoted by $\Oper$, is the set of all true program transitions: 
	$$
	\Oper\dddef\{(\alpha, \sigma)\trans (\alpha', \sigma')\ |\ \boolsem((\alpha, \sigma)\trans (\alpha', \sigma')) = 1\},  
	$$
	which is assumed to be \emph{well-defined} in the sense that 
	there is no true transitions of the form: $(\ter, \sigma)\trans ...$ from a terminal program.

	A \emph{path} $tr$ over $\Oper$ is a finite or infinite sequence:
	$s_1s_2...s_n...$ ($n\ge 1$), where for each pair $(s_i, s_{i+1})$ ($i\ge 1$), $(s_i\trans s_{i+1})\in \Oper$ is a true program transition.  
	A path is \emph{terminal}, if it ends with program $\ter$ in the form of $(\alpha_1,\sigma_1)...(\ter, \sigma_n)$ ($n\ge 1$). 
	We call a path \emph{minimum}, in the sense that in it there is no two identical program states. 
	A state $s$ is called \emph{terminal}, if from $s$ there is a terminal path over $\Oper$. 
	
	In order to reason about termination of a program, in $\Fmla$ we assume a unary operator $(\alpha, \sigma)\termi$ called the \emph{termination} of a program state $(\alpha, \sigma)$. 
	For an evaluation $\rho\in \Eval$, proposition $\rho((\alpha, \sigma)\termi)$ is defined as: $\rho((\alpha, \sigma)\termi)\dddef (\rho(\alpha, \sigma)) \termi$. 
	$\boolsem((\alpha, \sigma)\termi) \dddef 1$ if 
	there exists a terminal path over $\Oper$ starting from $(\alpha, \sigma)$. 
	We use $\Terminate$ to denote the set of all true program terminations: 
	$\Terminate\dddef \{(\alpha, \sigma)\termi\ |\ \boolsem((\alpha, \sigma)\termi) = 1\}$. 
	
	\ifx
	Let $\Termi\dddef\{(\alpha, \sigma)\termi\ |\ \alpha\in \Prog, \sigma\in\Conf\}$ be the set of all program terminations, 
	with $\Clo(\Termi)\dddef \{(\alpha, \sigma)\termi\ |\ \alpha\in \ClProg, \sigma\in\ClConf\}$ defined as the set of closed program terminations of $\Termi$. 
	For any program state $(\alpha, \sigma)\in \ClProg\times \ClConf$, 
	$\boolsem((\alpha, \sigma)\termi) = 1$ iff 
	there exists a terminal path over $\Oper$ starting from $(\alpha, \sigma)$. 
	We use $\Terminate$ to denote the set of all true program terminations: 
	$\Terminate\dddef \{(\alpha, \sigma)\termi\ |\ \boolsem((\alpha, \sigma)\termi) = 1\}\cap \Clo(\Termi)$. 
	\fi
	
	For programming languages, usually, program behaviours $\Oper$ (also program terminations $\Terminate$) are defined as structural operational semantics~\cite{Plotkin81} of Plotkin's style expressed as a set of rules, in a manner that a transition $\boolsem((\alpha, \sigma)\trans (\alpha', \sigma')) = 1$ only when it can be inferred according to these rules in a finite number of steps. 
	Appendix~\ref{section:Formal Definitions of While Programs} gives an example of the operational semantics of \emph{While} programs (Table~\ref{table:Operational Semantics of While Programs}). 
	
	\ifx
	As will be seen in Section~\ref{section:A Proof System for LDL}, program behaviours $\Oper$ and terminations $\Terminate$ can be expressed by a set of inference rules of program transitions, in a manner that a transition $(\alpha, \sigma)\trans (\alpha', \sigma')\in \Oper$, or a program termination $(\alpha, \sigma)\termi$, if 
	it can be inferred according to these rules in a finite number of steps. 
	\fi
	
	\begin{example}[Examples of Program Behaviours]
		In the example of Table~\ref{table:An Example of Program Structures}, 
		we have a program transition 
		$(\textit{WP}, \{n\mapsto 5, s\mapsto 0\})\trans (\textit{WP}', \{n\mapsto 5, s\mapsto 5\})$, where $\textit{WP}'\dddef n := n - 1\ ;\ \textit{WP}$, by executing assignment $s := s + n$. 
		This transition is also a program behaviour, since it matches the 
		operational semantics of $s := s + n$ (see Table~\ref{table:Operational Semantics of While Programs} of Appendix~\ref{section:Formal Definitions of While Programs} for more details).  
		However, 
		transition $(\textit{WP}, \{n\mapsto v, s\mapsto 0\})\trans (\textit{WP}', \{n\mapsto v, s\mapsto 5\})$ is not a proposition, because its truth value depends on whether $v > 0$ is true.  
		
		\ifx
		Let $\rho_{\WP}(n) = 5$, 
		then $(\rho_\WP(\WP), \sigma)\trans (\rho_{\WP}(\textit{WP}'), \{s\mapsto 0, n\mapsto 5\})$ is a program behaviour, where 
		$\rho_\WP(\textit{WP})\dddef s := 0\ ;\ \rho_\WP(\textit{WP}')$ and 
		$\rho_\WP(\textit{WP}')\dddef \textit{while}\ (5 > 0)\ \textit{do}\ s := s + 5\ ;\ n := n - 1\ \textit{end}$ are closed terms. 
		\fi
		
	\end{example}
	
	\textbf{Restrictions on Program Behaviours}.
	In this paper, our current research on the proof system of $\LDL$ (Section~\ref{section:A Cyclic Proof System for LDL}) confines us to only consider a certain type of programs whose behaviours satisfy the following properties. We simply call them \emph{program properties}. 
	
	
	\begin{definition}[Properties of Program Behaviours]
		\label{def:Program Properties}
		In program behaviours $\Oper$, it satisfies that 
		for each program state $(\alpha, \sigma)$, the following properties hold:
		\begin{enumerate}
			
			\ifx
			\item\label{item:Well-definedness-2}Well-Definedness. 
			For any $\sigma_1$ with $\sigma \cfeq \sigma_1$, 
			there is a bijection $f : A\to B$ such that for any $\sigma'\in A$, $\sigma' \cfeq f(\sigma')$, 
			where 
			$A\dddef \{\sigma'\ |\ (\alpha, \sigma)\trans (\alpha', \sigma')\mbox{ for some $\alpha'$}\}$, 
			$B\dddef \{\sigma'\ |\ (\alpha, \sigma_1)\trans (\alpha', \sigma')\mbox{ for some $\alpha'$}\}$. 
			\fi

			\item\label{item:Branching Finiteness} Branching Finiteness. From $(\alpha, \sigma)$, there exists only a finite number of transitions.

			\item\label{item:Termination Finiteness} Termination Finiteness. 
			From $(\alpha, \sigma)$, there is only a finite number of minimum terminal paths.  
		\end{enumerate}
	\end{definition}
	
	The currently discussed programs restricted by Definition~\ref{def:Program Properties} are actually a rich set, 
	including, for example, all deterministic programs (i.e., there is only one transition from any closed program state) and programs with finite state spaces (i.e., from each closed program state, only a finite number of closed program states can be reached). 
	\textit{While} programs and Esterel programs introduced in Table~\ref{table:An Example of Program Structures} above are both deterministic programs. They satisfy the program properties in Definition~\ref{def:Program Properties}. 
	
	However, there exist types of programs that do not satisfy (some of) these program properties. 
	For example, hybrid programs (cf.~\cite{Platzer07b}) with continuous behaviours do not satisfy branching finiteness. Some non-deterministic programs, like probabilistic programs~\cite{Barthe_Katoen_Silva_2020}, do not satisfy termination finiteness. 
	More future work will focus on relaxing these conditions to include a wider range of programs.

	\ifx
	We introduce the notion of equivalence that two configurations carry the same information about programs, which plays a critical role for building the proof system of $\LDL$ in Section~\ref{section:A Proof System for LDL}. 
	
	\begin{definition}[Configuration Equivalence]
		\label{def:Configuration Equivalence}
		Given two configurations $\sigma_1, \sigma_2\in \Conf$, 
		the equivalent relation $\sigma_1\cfeq \sigma_2$ is defined such that
		\begin{enumerate}
			\item \label{item:same effect} for any formula $\phi\in \Fmla$, $\app(\sigma_1, \phi)\equiv\app(\sigma_2, \phi)$; 
			\item \label{item:same program behaviour}for any state $(\alpha, \sigma_1)$ and $(\alpha, \sigma_2)$, 
			there is a bijection $\theta : \Oper(\alpha, \sigma_1)\to \Oper(\alpha, \sigma_2)$ such that for any $(\alpha', \sigma'_1)\in \Oper(\alpha, \sigma_1)$, $\theta((\alpha', \sigma'_1))\dddef (\alpha', \sigma'_2)$ for some $\alpha', \sigma'_1, \sigma'_2$ with $\sigma'_1\cfeq \sigma'_2$. 
		\end{enumerate}
	\end{definition}
	
	Intuitively, 
	Definition~\ref{def:Configuration Equivalence}-\ref{item:same effect} says that two configurations have the same effect on terms. 
	This does not necessarily mean two configurations are identical. 
	Definition~\ref{def:Configuration Equivalence}-\ref{item:same program behaviour} means that 
	if two configurations carry the same program information, then the program behaviours induced by them should be the same. 
	\fi
	
	\subsection{Syntax and Semantics of $\LDL$}
	\label{section:Syntax and Semantics of LDL}
	
	\ifx
	Notice that here we do not care about under which model a predicate $p$ is true. This depends on how we choose the domain $\Fmla$. We only stipulate that each closed term in $\Fmla$ must be either true or false. 
	\fi
	
	Based on the assumed sets $\Prog, \Conf$ and $\Fmla$, we give the syntax of $\LDL$ as follows. 
	
	\ifx
	\begin{definition}[$\LDL$ Formulas]
		\label{def:LDL Formulas}
		
		In $\Fmla$, we distinguish a modality operator: $[\cdot]$ and two logical operators: negation $\neg$ and conjunction $\wedge$. 
		A ``dynamic formula'' is inductively defined as follows:
		\begin{enumerate}
			\item $[\alpha]\phi$ is a dynamic formula if $\alpha\in \Prog$ and $\phi\in \Fmla$;
			\item $\neg\phi$ is a dynamic formula if $\phi$ is a dynamic formula in $\Fmla$;
			\item $\phi_1\wedge \phi_2$ is a dynamic formula, if $\phi_1, \phi_2\in \Fmla$ and one of $\phi_1,\phi_2$ is a dynamic formula. 
		\end{enumerate}
		
		A parameterized dynamic logical ($\LDL$) formula is a formula of the form:
		$\sigma : \phi$, where $\phi$ is a dynamic formula in $\Fmla$. 
		
	\end{definition}
	
	We call $\sigma : \phi$ a \emph{labeled formula} if $\phi\in \Fmla$ and call $\sigma$ the \emph{label} of $\phi$. 
	We also call `$[\alpha]$' a \emph{dynamic part} of a dynamic formula that contains it. 
	Denote the set of dynamic formulas in $\Fmla$ as $\dlFmla$, 
	and denote the set of $\LDL$ formulas as $\DLF$. 
	\fi
	
	\begin{definition}[$\LDL$ Formulas]
		\label{def:LDL Formulas}
		A parameterized dynamic logical ($\LDL$) formula
		$\phi$ is defined as follows in BNF form:
		$$
		\begin{aligned}
			\psi \dddef&\  F\ |\ \neg \psi\ |\ \psi\wedge \psi\ |\ [\alpha]\psi,\\
			\phi\dddef &\ F\ |\ \sigma : \psi\ |\ \neg\phi\ |\ \phi\wedge \phi,
		\end{aligned}
		$$
		where $F\in \Fmla$, $\alpha\in \Prog$ and $\sigma\in \Conf$. 
		
		We denote the set of $\LDL$ formulas as $\DLF$. 
	\end{definition}
	
	$\sigma$ and $\phi$ are called the \emph{label} and the \emph{unlabeled part} of formula $\sigma : \phi$ respectively. 
	We call $[\alpha]$ the \emph{dynamic part} of a formula, and 
	call a $\LDL$ formula having dynamic parts a \emph{dynamic formula}. 
	Intuitively, formula $[\alpha]\phi$ means that after all executions of program $\alpha$, formula $\phi$ holds. 
	$\la\cdot\ra$ is the dual operator of $[\cdot]$. 
	Formula $\la\alpha \ra\phi$ can be written as $\neg [\alpha]\neg\phi$. 
	Other formulas with logical connectives such as $\vee$ and $\to$ can be expressed by formulas with $\neg$ and $\wedge$ accordingly.  
	Formula $\sigma : \phi$ means that $\phi$ holds under configuration $\sigma$, as $\sigma$ has an impact on the semantics of formulas as indicated by interpretation $\app$.  
	
	\ifx
	A formula of $\DLF$ is \emph{closed} if in it all terms from $A$ ($A\in \TA$) are closed. We use $\Clo(\DLF)$ to denote the set of all closed terms in $\DLF$. 
	The identical equivalence relation $\equiv$ between $\LDL$ formulas can be trivially extended from $\equiv_A$ ($A\in \TA$). 
	We assume an evaluation $\rho : \DLF\to \Clo(\DLF)$ extended from 
	$\rho_A$ ($A\in \TA$), and require that 
	$\rho$ is \emph{structural} for $\LDL$ formulas w.r.t. $\Eval$ in the following sense: for any $\sigma$, $F\in \Fmla$ and unlabeled formulas $\phi, \phi_1, \phi_2$,  
	\begin{enumerate}[(1)]
		\item $\rho(\sigma : F)\equiv \rho_1(\sigma) : \rho_2(F)$ for some $\phi_1, \phi_2\in \Eval$;
		\item $\rho(\sigma : \neg\phi)\equiv \rho_1(\sigma) : \neg \rho_2(\phi)$ for some $\phi_1, \phi_2\in \Eval$;
		\item $\rho(\sigma : (\phi_1\wedge \phi_2))\equiv \rho_1(\sigma) : (\rho_2(\phi_1)\wedge \rho_3(\phi_2))$ for some $\rho_1, \rho_2, \rho_3\in \Eval$;
		\item $\rho(\sigma : [\alpha]\phi)\equiv \rho_1(\sigma) : [\rho_2(\alpha)]\rho_3(\phi)$ for some $\rho_1, \rho_2, \rho_3\in \Eval$. 
	\end{enumerate}
	With this property, it can be proved that $\rho$ is actually consistent with closed terms. 
	\fi

	The semantics of a traditional dynamic logic is based on a Kripke structure~\cite{Harel00}, in which programs are interpreted as a set of world pairs and logical formulas are interpreted as a set of worlds. 
	In $\LDL$, we define the semantics by extending $\rho$ and $\boolsem$ from Section~\ref{section:Programs and Configurations} to formulas $\DLF$ and directly base on program behaviours $\Oper$. 
	

	\ifx
	Traditionally, the semantics of a dynamic logic is denotational and is given as a Kripke structure~\cite{Harel00}. 
	Following a similar manner, 
	to give the semantics of $\LDL$ formulas, 
	we firstly build a special Kripke structure consisting of closed configurations as worlds, and program transitions as relations between worlds. 
	This structure allows us to define a satisfaction relation of $\GDL$ formulas by evaluations and configurations, 
	based on which, the semantics of $\LDL$ is given. 
	\fi
	
	\ifx
	As shown in Section~\ref{???}, traditionally, the semantics of a dynamic logic is denotational and is given as a Kripke structure~\cite{???}. 
	Since $\LDL$ is an abstract logic where the explicit forms of its formulas rely on the discussed domains $\Prog, \Conf$ and $\Fmla$, the traditional notion of Kripke structures does not work as it requires explicit structures of programs. 
	
	Instead, we build a special type of Kripke structure directly based on closed configurations $\Clo(\Conf)$. 
	The operational semantics of programs is reflected as relations between closed configurations in the Kripke structure. 
	\fi

	\ifx
	\begin{definition}[Kripke Structure of $\LDL$]
		Given $\Prog, \Conf, \Fmla(\Pred)$ and $\Oper$, 
		the Kripke structure of $\LDL$ is a triple $M \dddef (W, \to, \mcl{I})$, where 
		$W = \Eval\times\Prog\times\Conf$ is the set of `worlds', $\to\subseteq W\times W$ is a set of relations between worlds, $\mcl{I}: \Fmla\to \mcl{P}(W)$
		is an interpretation of formulas in $\Fmla$ on the set of worlds, 
		satisfying the following conditions:
		\begin{enumerate}
			\item For each world $(\rho, \alpha, \sigma)$, $(\rho, \alpha, \sigma) \rightarrow (\rho, \alpha', \sigma')$ iff $((\alpha, \sigma)\trans (\alpha', \sigma'))\in \Oper$.
			\item For each formula $\phi\in \Fmla$, $(\rho, \alpha, \sigma)\in \mcl{I}(\phi)$ iff $\rho\models_\boolsem \app(\sigma, \phi)$. 
		\end{enumerate}
		
	\end{definition}
	
	Different from the typical Kripke structure (cf.~\cite{Harel00}), 
	in the Kripke structure of $\LDL$, the transitions between worlds are defined according to the program behaviours $\Oper$, rather than the syntactic structures of programs.

	Below we do not distinguish relation $\sigma \xrightarrow{\alpha/\alpha'} \sigma'$ between worlds $\sigma, \sigma'$ from 
	a program transition $(\alpha, \sigma)\trans (\alpha', \sigma')$. 
	\fi
	
	\ifx
	A \emph{path} $tr$ over $\Oper$ is a finite or infinite sequence:
	$s_1s_2...s_n...$ ($n\ge 1$), where for each pair $(s_i, s_{i+1})$ ($i\ge 1$), $(s_i\trans s_{i+1})\in \Oper$ is a program transition.  
	A path is \emph{terminal}, if it ends with program $\ter$ in the form of $(\alpha_1,\sigma_1)...(\ter, \sigma_n)$ ($n\ge 1$). 
	We call a path \emph{minimum}, in the sense that in it there is no two identical program states.
	\fi
	
	\ifx
	\begin{definition}[Evaluation on $\LDL$ Formulas]
		In $\LDL$ we assume a structural evaluation extended from $\rho : \TA\to \TA$, still denoted by $\rho$, such that for any $\phi\in \DLF$, $\rho(\phi)$ is a proposition that has a boolean semantics being either true of false.  
		
		We denote the set of all propositions of $\DLF$ by $\rho$ as $\Prop_\DLF$. 
	\end{definition}
	\fi
	
	\begin{definition}[Semantics of $\LDL$ Formulas]
		\label{def:Semantics of Labelled LDL Formulas}
		An  evaluation $\rho\in \Eval$ structurally maps a $\LDL$ formula $\phi$ into a proposition, whose truth value is defined inductively as follows by extending $\boolsem$:
		\begin{enumerate}
			\item $\boolsem(\rho(F))$ is already defined, if $F\in \Fmla$;
			\item $\boolsem(\rho(\sigma : F))\dddef 1$, if $F\in \Fmla$ and $\rho\models_\boolsem\app(\sigma, F)$;
			\item $\boolsem(\rho(\sigma : \neg \phi))\dddef \boolsem(\rho(\neg (\sigma : \phi)))$;
			\item $\boolsem(\rho(\sigma : \phi_1\wedge \phi_2))\dddef \boolsem(\rho((\sigma : \phi_1)\wedge (\sigma : \phi_2)))$;
			\item $\boolsem(\rho(\sigma : [\alpha]\phi))\dddef 1$, if for all terminal paths
			$(\alpha_1, \sigma_1)...(\ter, \sigma_n)$ ($n\ge 1$) over $\Oper$ with $(\alpha_1, \sigma_1)\equiv \rho(\alpha, \sigma)$, 
			$\boolsem(\rho(\sigma_n : \phi))\dddef 1$; 
			\item $\boolsem(\rho(\neg \phi))\dddef 1$, if $\boolsem(\rho(\phi)) = 0$;
			\item $\boolsem(\rho(\phi_1\wedge \phi_2))\dddef 1$, if both $\boolsem(\rho(\phi_1)) = 1$ and $\boolsem(\rho(\phi_2)) = 1$;
			\item $\boolsem(\rho(\phi))\dddef 0$ for any $\phi\in \DLF$, otherwise.    
		\end{enumerate}
		
		For any $\phi\in \DLF$, write $\rho\models_\boolsem \phi$ if $\boolsem(\rho(\phi)) = 1$, or simply $\rho\models \phi$. 
		
		\ifx
		A $\LDL$ formula $\sigma : \phi$ is satisfied by an evaluation $\rho\in \Eval$, denoted by $\rho\models \sigma : \phi$, 
		is defined if $\rho, \sigma\models \phi$, which 
		is inductively defined as follows based on the structure of unlabelled formula $\phi$:
		\begin{enumerate}
			\item $\rho, \sigma\models F$, if $F\in \Fmla$ and
			$\rho\models_\boolsem \app(\sigma, F)$.
			\item 
			$\rho, \sigma\models \neg\phi$, if $\rho, \sigma\not\models \phi$. 
			\item 
			$\rho, \sigma\models \phi_1\wedge \phi_2$, if both $\rho, \sigma\models \phi_1$ and $\rho, \sigma\models \phi_2$. 
			\item 
			$\rho, \sigma\models [\alpha]\phi$, if for all terminal paths
			$(\alpha_1, \sigma_1)...(\ter, \sigma_n)$ ($n\ge 1$) over $\Oper$ with $(\alpha_1, \sigma_1)\equiv \rho(\alpha, \sigma)$, 
			$\rho, \sigma_n\models \phi$. 
		\end{enumerate}
		\fi
		
	\end{definition}
	
	\ifx
	\begin{definition}[Satisfaction Relation of $\GDL$ Formulas]
		\label{def:Semantics of dynamic logical formulas}
		Given the Kripke structure $M = (\Conf(\Sigma), \to, \mcl{I})$ of $\GDL$, 
		the 
		satisfaction of a $\GDL$ formula $\phi$ w.r.t. an evaluation $\rho$ and a configuration $\sigma\in \Conf$, denoted by $\rho, \sigma\models_M \phi$ (or simply $\rho, \sigma\models \phi$), 
		is inductively defined as follows:
		\begin{enumerate}
			\item $\rho, \sigma\models_M \phi$ where $\phi\in \Fmla$, if 
			$\rho(\sigma)\in \mcl{I}(\phi)$.
			\item 
			$\rho, \sigma\models \neg\phi$, if $\rho, \sigma\not\models \phi$. 
			\item 
			$\rho, \sigma\models \phi_1\wedge \phi_2$, if both $\rho, \sigma\models \phi_1$ and $\rho, \sigma\models \phi_2$. 
			\item 
			$\rho, \sigma\models_M [\alpha]\phi$, if (1) $\rho(\alpha)$ is $\ter$ and $\rho, \sigma\models \phi$, or (2) for all paths
			$\sigma\xrightarrow{\rho(\alpha)/\alpha_1}...\xrightarrow{\alpha_n/\ter}\sigma'$, 
			$\rho, \sigma'\models_M \phi$. 
		\end{enumerate} 
	\end{definition}
	\fi
	
	Notice that for a $\LDL$ formula $\phi$, how $\rho(\phi)$ is defined explicitly can vary from case to case. 
	In Definition~\ref{def:Semantics of Labelled LDL Formulas}, we only assume that 
	$\rho$ is structural and $\rho(\phi)$ must be a proposition. 
	An example of evaluation $\rho$ is given in Appendix~\ref{section:Formal Definitions of While Programs}. 
	
	According to the semantics of operator $\la\cdot \ra$, 
	we can have that
	$\boolsem(\rho(\la\alpha\ra \phi)) = 1$ iff there exists a terminal path
	$(\alpha_1, \sigma_1)...(\ter, \sigma_n)$ ($n\ge 1$) over $\Oper$ with $(\alpha_1,\sigma_1)\equiv \rho(\alpha, \sigma)$ such that 
	$\boolsem(\rho(\sigma_n : \phi)) = 1$.

	\ifx
	We define the semantics of a $\LDL$ formula 
	naturally based on the Kripke structure of $\GDL$ and the satisfaction of $\GDL$ formulas on the Kripke structure. 
	\fi
	
	\ifx
	\begin{definition}[Semantics of Labelled $\LDL$ Formulas]
		\label{def:Semantics of Labelled LDL Formulas}
		Given the Kripke structure $M = (\Clo(\Conf), \to, \mcl{I})$ of $\LDL$,
		the semantics of a labelled $\LDL$ formula $\sigma : \phi$ is given as the satisfaction relation w.r.t. $M$ and a closed configuration $\sigma'$ as follows:
		$$
		M, \sigma'\models \sigma : \phi\mbox{, if $\sigma'\in \Clo(\sigma)$ implies $M, \sigma'\models \phi$}. 
		$$
		
		Note that in Definition~\ref{def:Semantics of Labelled LDL Formulas}, 
		when $\sigma'\notin \Clo(\sigma)$, $\sigma'$ has no effect to the validity of $\sigma : \phi$. 
		So we always have $M, \sigma'\models \sigma: \phi$. 
		
	\end{definition}
	\fi
	
	\ifx
	\begin{definition}[Semantics of $\LDL$ Formulas]
		\label{def:Semantics of Labelled LDL Formulas}
		The semantics of a $\LDL$ formula $\sigma : \phi$ is given as the satisfaction relation by an evaluation $\rho$ as follows:
		$
		\rho\models \sigma : \phi\mbox{, if $\rho, \sigma\models \phi$}. 
		$
		
		\ifx
		Note that in Definition~\ref{def:Semantics of Labelled LDL Formulas}, 
		when $\sigma'\notin \Clo(\sigma)$, $\sigma'$ has no effect to the validity of $\sigma : \phi$. 
		So we always have $M, \sigma'\models \sigma: \phi$. 
		\fi
		
	\end{definition}
	\fi
	
	\ifx
	The satisfiability and validity of $\LDL$ formulas are introduced in a standard way. 
	A $\LDL$ formula $\sigma: \phi$ is \emph{satisfiable}, if 
	there exists an evaluation $\rho$ such that 
	$\rho\models_M \sigma: \phi$. 
	$\sigma: \phi$ is \emph{valid}, if $\rho\models_M\sigma : \phi$ for all evaluations $\rho$, also denoted by $\models_M \sigma: \phi$.  
	\fi
	
	\ifx
	\begin{definition}[Satisfiability and Validity]
		\label{def:Satisfiability and Validity}
		Given the Kripke structure $M = (\Clo(\Conf), \to, \mcl{I})$ of $\LDL$, 
		
		A $\LDL$ formula/labelled $\LDL$ formula $\phi$ is `satisfiable', if 
		there exists some configuration $\sigma\in \Clo(\Conf)$ such that 
		$M, \sigma\models \phi$ (simply abbreviated as $\sigma\models \phi$). 
		
		A $\LDL$ formula/labelled $\LDL$ formula $\phi$ is `valid', if $M, \sigma\models\phi$ for all configurations $\sigma\in \Clo(\Conf)$, denoted by $M\models \phi$ (or simply $\models \phi$). 
	\end{definition}
	\fi

	
	A $\LDL$ formula $\phi$ is called \emph{valid}, if $\rho\models \phi$ for all evaluation $\rho\in \Eval$. 
	
	\begin{example}[$\LDL$ Specifications]
		\label{example:DLp specifications}
		
		A property of program $\textit{WP}\in \Prog_\WP$ (Table~\ref{table:An Example of Program Structures})
		is described as the following formula
		$$
		v \ge 0 \to \sigma_1 : [\textit{WP}] s = ((v+1)v)/2, 
		$$
		where $\sigma_1 \dddef \{n\mapsto v, s\mapsto 0\}$ with $v$ a free variable. 
		This formula means that 
		given an initial value $v\ge 0$, after executing $\textit{WP}$, $s$ equals to $((v+1)v)/2$, which is the sum of $1 + 2 + ... + v$.
		We will prove this formula in Section~\ref{section:Example One: A While Program}. 
		\ifx
		Recall that $\textit{WP}$ is defined as:
		$$
		\begin{aligned}
			\textit{WP}
			\dddef&
			\{
			\textit{while}\
			(n > 0)\
			\textit{do}\
			s := s + n\ ;\
			n := n - 1\
			\textit{end}\
			\}.
		\end{aligned}
		$$
		\fi
	\end{example}
	
	\section{A Cyclic Proof System for $\LDL$}
	\label{section:A Cyclic Proof System for LDL}
	
	In this section, we propose a cyclic proof system for $\LDL$. 
	In Section~\ref{section:A Proof System for LDL}, we firstly propose a proof system $\pfDLp$ to support reasoning based on program behaviours. 
	Then in Section~\ref{section:Construction of A Cyclic Preproof Structure} we construct a cyclic preproof structure for $\pfDLp$, which support deriving infinite proof trees under certain soundness conditions. 
	In Section~\ref{section:Lifting Process From Program Domains}, we propose a lifting process to allow $\LDL$ support reasoning with existing rules from particular dynamic-logic theories. 
	Section~\ref{section:Proof System and Sequent Calculus} introduces the notion of sequent calculus. 
	
	\subsection{Sequent Calculus}
	\label{section:Proof System and Sequent Calculus}
	
	\textbf{Sequents}. 
	In this paper, we adopt sequents~\cite{Gentzen34} as the derivation form of $\LDL$.
	Sequent is convenient for goal-directed reversed derivation procedure which many general theorem provers, such as Coq and Isabelle, are based on. 
	
	A \emph{sequent} $\nu$ is a logical argumentation of the form:
	$$
	\Gamma\Rightarrow \Delta,
	$$
	where $\Gamma$ and $\Delta$ are finite multi-sets of formulas split by an arrow $\Rightarrow$, called the \emph{left side} and the \emph{right side} of the sequent respectively.  
	We use dot $\cdot$ to express $\Gamma$ or $\Delta$ when they are empty sets. 
	A sequent $\Gamma\Rightarrow \Delta$ expresses the 
	formula $\bigwedge_{\phi\in \Gamma}\phi \to \bigvee_{\psi\in \Delta}\psi$, denoted by $\mfr{P}(\Gamma\Rightarrow \Delta)$, meaning that if all formulas in $\Gamma$ hold, then one of formulas in $\Delta$ holds. 
	
	\textbf{Inference Rules}.
	An \emph{inference rule} is of the form
	$$
	\begin{aligned}
		\infer[]
		{
			\nu
		}
		{
			\nu_1
			&
			...
			&
			\nu_n
		}
		,\end{aligned}$$
	where each of $\nu,\nu_i$ ($1\le i\le n$) is also called a \emph{node}. 
	Each of $\nu_1,...,\nu_n$ is called a \emph{premise}, and $\nu$ is called the \emph{conclusion}, of the rule. 
	The semantics of the rule is that 
	if the validity of formulas $\mfr{P}(\nu_1), ..., \mfr{P}(\nu_n)$ implies the validity of formula $\mfr{P}(\nu)$. 
	A formula pair $(\tau, \tau_i)$ ($1\le i\le n$) with formula $\tau$ in node $\nu$ and formula $\tau_i$ in node $\nu_i$ is called a \emph{conclusion-premise} (CP) pair.

	We use a double-lined inference form:
	$$
	\infer=[]
	{\phi}
	{\phi_1 & ... & \phi_n}
	$$
	to represent both rules
	$$
	\begin{aligned}
		\infer[]
		{
			\Gamma\Rightarrow \phi, \Delta
		}
		{
			\Gamma\Rightarrow \phi_1, \Delta
			&
			...
			&
			\Gamma\Rightarrow \phi_n, \Delta
		}
	\end{aligned}
	\ \ 
	\mbox{and}
	\ \ 
	\begin{aligned}
		\infer[]
		{
			\Gamma, \phi\Rightarrow \Delta
		}
		{
			\Gamma, \phi_1\Rightarrow  \Delta
			&
			...
			&
			\Gamma, \phi_n\Rightarrow \Delta
		}
	\end{aligned}
	,
	$$
	provided any sets $\Gamma$ and $\Delta$. 
	We often call the CP pair $(\phi, \phi_i)$ ($1\le i\le n$) a \emph{target pair}, 
	call $\phi, \phi_i$ \emph{target formulas}. 
	$\Gamma$ and $\Delta$ are called the \emph{context} of a sequent.

	\textbf{Proof Systems}. 
	A \emph{proof tree} is a tree-like structure formed by deducing a node as a conclusion backwardly by consecutively applying inference rules. 
	In a proof tree, each node is the conclusion of a rule. 
	The root node of the tree is called the \emph{conclusion} of the proof. 
	Each leaf node of the tree is called \emph{terminal}, if it is the conclusion of an axiom. 
	A proof tree is called \emph{finite} if all of its leaf nodes terminate. 
	
	
	A \emph{proof system} $P$ consists of a set of inference rules. 
	We say that a node $\nu$ can be derived from $P$, denoted by $\vdash_P \nu$, if 
	a finite proof tree with $\nu$ the root node can be formed by applying the rules in $P$. 
	\ifx
	which satisfies that 
	\begin{enumerate}[(1)]
		\item for each  formula $\phi$ of a node in the proof tree, $\phi\in D$;
		\item each node is the conclusion of a rule in $A$.
	\end{enumerate}
	\fi
	
	Notice that the rules with names shown in the tables below in this paper (e.g. rule $([\alpha])$ in Table~\ref{table:General Rules for LDL}) are actually \emph{rule schemata}, that is, inferences rules with meta variables as parameters (e.g. ``$\Gamma, \Delta, \sigma, \alpha, \phi, \sigma', \alpha'$'' in rule $([\alpha])$). But in the discussions below we usually still call them ``inference rules'', and sometimes do not distinguish them from their instances if no ambiguities are caused. 
	
	\ifx
	\textbf{Sequent Calculi}. 
	In this paper, we adopt sequents~\cite{Gentzen34} as the derivation form of $\LDL$.
	Sequent is convenient for goal-directed reversed derivation procedure which many general theorem provers, such as Coq and Isabelle, are based on. 
	
	A \emph{sequent} is a node of the form:
	$$
	\Gamma\Rightarrow \Delta,
	$$
	where $\Gamma$ and $\Delta$ are finite multisets of formulas split by an arrow $\Rightarrow$, called the \emph{left side} and the \emph{right side} of the sequent respectively.  
	We use dot $\cdot$ to express $\Gamma$ or $\Delta$ when they are empty sets. 
	A sequent $\Gamma\Rightarrow \Delta$ expresses the 
	proposition $\mfr{P}(\Gamma\Rightarrow \Delta)$: 
	``for all evaluation $\rho\in \Eval$, if $\rho\models \phi$ for all $\phi\in \Gamma$, then $\rho\models \psi$ for some $\psi\in \Delta$''. 
	We say a sequent $\Gamma\Rightarrow \Delta$ is \emph{sound} if $\mfr{P}(\Gamma\Rightarrow \Delta)$ is true. 

	we often call a CP pair $(\phi, \phi_i)$ ($1\le i\le n$) a \emph{target pair}, 
	call $\phi, \phi_i$ \emph{target formulas}. 
	$\Gamma$ and $\Delta$ are called the \emph{context} of a sequent. 
	
	We sometimes write $\sigma : \Gamma$ to mean the set of labeled formulas $\{\sigma : \phi\ |\ \phi\in \Gamma\}$. 
	Write $\rho\models \Gamma$ to mean that $\rho\models \phi$ for all $\phi\in \Gamma$. 
	
	\fi

	\subsection{A Proof System for $\LDL$}
	\label{section:A Proof System for LDL}
	\ifx
	As shown in Table~\ref{table:General Rules for LDL}, 
	the proof system of $\LDL$ consists of a set of rules for deriving $\LDL$ labeled dynamic formulas based on programs' operational semantics and a \emph{lifting rule} to support reasoning with special rules in specific domains. 
	\fi
	
	\ifx
	As shown in Table~\ref{table:General Rules for LDL}, 
	the proof system of $\LDL$ consists of a set of rules for deriving $\LDL$ formulas based on program executions according to their operational semantics. 
	\fi
	
	The proof system $\pfDLp$ of $\LDL$ 
	relies on a pre-defined proof system $P_{\Oper,\Terminate}$ 
	for deriving program behaviours $\Oper$ and terminations $\Terminate$. 
	$P_{\Oper,\Terminate}$ is \emph{sound} and \emph{complete} w.r.t. $\Oper$ and $\Terminate$ in the sense that 
	for any transition $(\alpha_1, \sigma_1)\trans (\alpha_2, \sigma_2)\in \Prop$ and termination $(\alpha, \sigma)\termi\ \in \Prop$, $(\alpha_1, \sigma_1)\trans (\alpha_2, \sigma_2)\in \Oper$ iff 
	$\vdash_{P_{\Oper, \Terminate}} \cdot \Rightarrow (\alpha_1, \sigma_1)\trans (\alpha_2, \sigma_2)$, and $(\alpha, \sigma)\termi\ \in \Terminate$ iff $\vdash_{P_{\Oper, \Terminate}}\cdot \Rightarrow (\alpha, \sigma)\termi$. 
	
	Note that in $\pfDLp$, we usually assume that a formula does not contain a program transition or termination. 
	
	\ifx
	$\mcl{E}$, called an \emph{extra theory}, is for deriving domain-specific theories according to the definitions of $\Prog, \Conf$ and $\Fmla$. 
	It helps for 
	building heuristic structures for the purpose of constructing cyclic proof structures (as will be seen later as an example in Section~\ref{section:Example One: A While Program}). 
	However, to obtain a well-suited derivation trace, the forms of extra rules have to be restricted. 
	In Section~\ref{section:Construction of A Cyclic Preproof Structure}, we will propose a so-called ``safe condition'' for an extra rule to follow in order to obtain a cyclic preproof structure. 
	We generally denote an extra rule by $(\Extra : \cdot)$. 
	\fi
	
	\ifx
	relies  a pre-defined proof system $P_{\Oper,\Terminate} = (\mcl{S}, \DLF\cup \AFmla)$ for program behaviours $\Oper$ and terminations $\Terminate$, and a pre-define set $\mcl{E}$ of rules for domain-specific theories according to the definitions of $\Prog, \Conf$ and $\AFmla$, called \emph{extra theory}.  
	$P_{\Oper,\Terminate}$ is \emph{sound} and \emph{complete} w.r.t. $\Oper$ and $\Terminate$, in the sense that 
	for any state $(\alpha, \sigma)$ and closed transition $(\alpha_1, \sigma_1)\trans (\alpha_2, \sigma_2)\in \Beha$, $(\alpha, \sigma)\in \Terminate$ iff $\vdash_{P_{\Oper, \Terminate}}\cdot \Rightarrow (\alpha, \sigma)\termi$, and  $(\alpha_1, \sigma_1)\trans (\alpha_2, \sigma_2)\in \Oper$ iff 
	$\vdash_{P_{\Oper, \Terminate}} \cdot \Rightarrow (\alpha_1, \sigma_1)\trans (\alpha_2, \sigma_2)$. 
	\fi

	\ifx
	for program behaviours $\Oper$. 
	$\mcl{S}$ is complete w.r.t. $\Oper$, in the sense that 
	for any transition $(\alpha, \sigma)\trans (\alpha', \sigma')\in \Oper$, 
	$\vdash_{\mcl{S}} \cdot \Rightarrow (\alpha, \sigma)\trans (\alpha', \sigma')$. 
	\fi
	
	\begin{table}[tb]
		\begin{center}
			\noindent\makebox[\textwidth]{%
				\scalebox{0.9}{
					\begin{tabular}{c}
						\toprule
						$
						\infer[^{(x:=e)}]
						{\Gamma\Rightarrow (x:=e, \sigma)\trans(\ter, \sigma^x_{\sigma^*(e)}), \Delta}
						{
						}
						$
						\ \
						$
						\infer[^{(;)}]
						{\Gamma\Rightarrow (\alpha_1 ; \alpha_2, \sigma)\trans(\alpha'_1 ; \alpha_2, \sigma'), \Delta}
						{
							\Gamma\Rightarrow (\alpha_1, \sigma)\trans (\alpha'_1, \sigma'), \Delta
						}
						$
						\\
						$
						\infer[^{(;\ter)}]
						{\Gamma\Rightarrow(\alpha_1; \alpha_2, \sigma)\trans(\alpha_2, \sigma'), \Delta}
						{
							\Gamma\Rightarrow (\alpha_1, \sigma)\trans (\ter, \sigma'), \Delta
						}
						$
						\ \ 
						$
						\infer[^{(\textit{wh1})}]
						{\Gamma\Rightarrow (\textit{while}\ \phi\ \textit{do}\ \alpha\ \textit{end}, \sigma)\trans(\alpha';\ \textit{while}\ \phi\ \textit{do}\ \alpha\ \textit{end}, \sigma'), \Delta}
						{
							\Gamma, \app_\WP(\sigma, \phi)\Rightarrow (\alpha, \sigma)\trans(\alpha', \sigma'), \Delta
							&
							\Gamma\Rightarrow \app_\WP(\sigma, \phi), \Delta
						}
						$
						\\
						$
						\infer[^{(\textit{wh1}\ter)}]
						{\Gamma\Rightarrow (\textit{while}\ \phi\ \textit{do}\ \alpha\ \textit{end}, \sigma)\trans(\textit{while}\ \phi\ \textit{do}\ \alpha\ \textit{end}, \sigma'), \Delta}
						{
							\Gamma, \app_\WP(\sigma, \phi)\Rightarrow (\alpha, \sigma)\trans(\ter, \sigma'), \Delta
							&
							\Gamma\Rightarrow \app_\WP(\sigma, \phi), \Delta
						}
						$
						\ \ 
						$
						\infer[^{(\textit{wh2})}]
						{\Gamma\Rightarrow (\textit{while}\ \phi\ \textit{do}\ \alpha\ \textit{end}, \sigma)\trans(\ter, \sigma), \Delta}
						{
							\Gamma\Rightarrow \neg \app_\WP(\sigma, \phi), \Delta
						}
						$
						\\
						\bottomrule
					\end{tabular}
				}
			}
		\end{center}
		\caption{Partial Inference Rules for Program Behaviours of \emph{While} programs}
		\label{table:An Example of Inference Rules for Program Behaviours}
	\end{table}
	
	\begin{example}
		\label{example:Proof system for program behaviours}
		As an example, Table~\ref{table:An Example of Inference Rules for Program Behaviours} displays a part of inference rules of the proof system $P_{\Oper_\WP, \Terminate_\WP}$ for the program behaviours $\Oper_\WP$, 
		which are obtained directly from the operational semantics of \emph{While} programs (Table~\ref{table:Operational Semantics of While Programs} of Appendix~\ref{section:Formal Definitions of While Programs}). 
		The whole proof system is shown in Table~\ref{table:Inference Rules for Program Behaviours of While programs}, \ref{table:Inference Rules for Program Terminations of While programs} of Appendix~\ref{section:Formal Definitions of While Programs}. 
		The rules in Table~\ref{table:An Example of Inference Rules for Program Behaviours} will be used in the proof procedure of the example in Section~\ref{section:Example One: A While Program}. 
		In Table~\ref{table:An Example of Inference Rules for Program Behaviours}, 
		$\sigma^x_e$ represents a configuration that store 
		variable $x$ as value $e$, while storing other variables as the same value as $\sigma$. 
		Similarly to $\app_\WP$, $\sigma^*(e)$ is defined to return a term obtained by replacing 
		each free variable $x$ of $e$ by the value of $x$ in $\sigma$ (see Appendix~\ref{section:Formal Definitions of While Programs}). 
	\end{example}

	Table~\ref{table:General Rules for LDL} lists the primitive rules of $\pfDLp$.
	Through $\pfDLp$, 
	a $\LDL$ formula can be transformed into proof obligations as non-dynamic formulas, which can then be encoded and verified accordingly through, for example, an SAT/SMT checking procedure. 
	For easy understanding, we give rule $(\la\alpha\ra)$ instead of the version of rule $([\alpha])$ for the left-side derivations, which can be derived using $(\la\alpha\ra)$ and rules $(\neg L)$ and $(\neg R)$.  
	The rules for other operators like $\vee$, $\to$ can be derived accordingly using the rules in Table~\ref{table:General Rules for LDL}.   
	\ifx
	For easy understanding of the proof rules for operator $\la\cdot \ra$, 
	we only give the rules for the right-side derivations of sequents instead of the rules for negations (rules $(\sigma \neg L)$ and $(\sigma \neg R)$). 
	Other rules of $\mcl{DL}_p$ for the left-side derivations of sequents can be derived based on these rules using the rules for negations. 
	\fi
	
	\begin{table}[tb]
		\begin{center}
			\noindent\makebox[\textwidth]{%
				\scalebox{1.0}{
					\begingroup
					\begin{tabular}{c}
						\toprule
						$
						\begin{aligned}
							\infer[^{1\ ([\alpha])}]
							{\Gamma\Rightarrow \sigma : [\alpha]\phi, \Delta}
							{
								\{
								\Gamma\Rightarrow \sigma' : [\alpha']\phi, \Delta
								\}_{(\alpha', \sigma')\in \Phi}
							}
						\end{aligned}
						$,
						\ \ 
						where
						$
						\Phi\dddef \{(\alpha', \sigma')\ |\ \mbox{$\vdash_{P_{\Oper,\Terminate}} (\Gamma\Rightarrow (\alpha, \sigma)\trans(\alpha', \sigma'), \Delta)$}\}
							$
							\\
							\midrule
							$
							\begin{aligned}
								\infer[^{1\ (\la\alpha\ra)}]
								{\Gamma \Rightarrow \sigma: \la\alpha\ra\phi, \Delta}
								{
									\Gamma\Rightarrow \sigma': \la\alpha'\ra\phi, \Delta
								}
							\end{aligned}
							$,
							\ \ 
							if $\vdash_{P_{\Oper, \Terminate}} (\Gamma\Rightarrow (\alpha, \sigma)\trans(\alpha', \sigma'), \Delta)$
							\\
							\midrule
							$
							\begin{aligned}
								\infer[^{2\ (\textit{Ter})}]
								{\Gamma \Rightarrow \Delta}
								{}
							\end{aligned}
							$
							\ 
							\vline
							\ 
							$
							\begin{aligned}
								\infer=[^{3\ (\textit{Int})}]
								{\sigma : \phi}
								{\app(\sigma, \phi)}
							\end{aligned}
							$
							\ \vline\ 
							$
							\begin{aligned}
								\infer=[^{([\ter])}]
								{\sigma : [\ter]\phi}
								{\sigma : \phi}
							\end{aligned}
							$
							\ \vline\ 
							$
							\begin{aligned}
								\infer[^{4\ (\textit{Sub})}]
								{\Sub(\Gamma)\Rightarrow \Sub(\Delta)}
								{\Gamma\Rightarrow \Delta}
							\end{aligned}
							$
							\ \vline\ 
							$
							\begin{aligned}
								\infer=[^{(\sigma\neg)}]
								{\sigma : (\neg \phi)}
								{\neg(\sigma : \phi)}
							\end{aligned}
							$
							\ifx
							\ \vline\ 
							$
							\begin{aligned}
								\infer=[^{(\sigma\neg 2)}]
								{\neg(\sigma : \phi)}
								{\sigma : (\neg \phi)}
							\end{aligned}
							$
							\fi
							\ \vline\ 
							$
							\begin{aligned}
								\infer=[^{(\sigma\wedge)}]
								{ \sigma : (\phi\wedge\psi)}
								{(\sigma : \phi)\wedge (\sigma : \psi)}
							\end{aligned}
							$
							\ifx
							\ \vline\ 
							$
							\begin{aligned}
								\infer=[^{(\sigma\wedge 2)}]
								{(\sigma : \phi)\wedge (\sigma : \psi)}
								{\sigma : (\phi\wedge\psi)}
							\end{aligned}
							$
							\fi
							\\
							\midrule
							$
							\begin{aligned}
								\infer[^{(\textit{ax})}]
								{\Gamma, \phi\Rightarrow \phi, \Delta}
								{}
							\end{aligned}
							$
							\ \vline\ 
							$
							\begin{aligned}
								\infer[^{(\textit{Cut})}]
								{\Gamma\Rightarrow \Delta}
								{\Gamma\Rightarrow \phi, \Delta & 
									\Gamma, \phi\Rightarrow \Delta}
							\end{aligned}
							$
							\ \vline\ 
							$
							\begin{aligned}
								\infer[^{(\textit{Wk L})}]
								{\Gamma, \phi\Rightarrow  \Delta}
								{\Gamma \Rightarrow \Delta}
							\end{aligned}
							$
							\ \vline\ 
							$
							\begin{aligned}
								\infer[^{(\textit{Wk R})}]
								{\Gamma\Rightarrow \phi, \Delta}
								{\Gamma\Rightarrow \Delta}
							\end{aligned}
							$
							\ \vline\ 
							$
							\begin{aligned}
								\infer=[^{(\textit{Con})}]
								{\phi}
								{\phi, \phi}
							\end{aligned}
							$
							\\
							\midrule
							$
							\begin{aligned}
								\infer[^{(\neg L)}]
								{\Gamma, \neg \phi\Rightarrow \Delta}
								{\Gamma\Rightarrow \phi, \Delta}
							\end{aligned}
							$
							\ \vline\ 
							$
							\begin{aligned}
								\infer[^{(\neg R)}]
								{\Gamma\Rightarrow \neg \phi, \Delta}
								{\Gamma, \phi \Rightarrow \Delta}
							\end{aligned}
							$
							\ \vline\ 
							$
							\begin{aligned}
								\infer[^{(\wedge L)}]
								{\Gamma, \phi\wedge\psi\Rightarrow  \Delta}
								{\Gamma, \phi, \psi\Rightarrow \Delta}
							\end{aligned}
							$
							\ \vline\ 
							$
							\begin{aligned}
								\infer[^{(\wedge R)}]
								{\Gamma\Rightarrow \phi\wedge \psi, \Delta}
								{\Gamma\Rightarrow \phi, \Delta
									&
									\Gamma\Rightarrow \psi, \Delta}
							\end{aligned}
							$
							\\
							\midrule
							\multicolumn{1}{l}{
								\begin{tabular}{l}
									$^{1}$ $\alpha\not\equiv \ter$.\\
									$^{2}$ for each $\phi$ in $\Gamma$ or $\Delta$, $\phi\in \Fmla$; $\mfr{P}(\Gamma\Rightarrow \Delta)$ is valid.\\ 
									$^{3}$ $\phi\in \Fmla$ is a non-dynamic formula.\\
									$^{4}$ $\Sub$ is a substitution by Definition~\ref{def:Abstract Substitution}. 
								\end{tabular}
							}
							\\
							\bottomrule
						\end{tabular}
						\endgroup
					}
				}
			\end{center}
			\caption{Primitive Rules of Proof System $\pfDLp$}
			\label{table:General Rules for LDL}
		\end{table}

		\ifx
		We propose a set of general rules for $\LDL$ as shown in Table~\ref{table:General Rules for LDL}.
		These rules only deals with $\LDL$ formulas, they perform either transformations between $\LDL$ formulas or rewrites of dynamic $\LDL$ formulas based on the operational semantics of programs.
		\fi

		\ifx
		$\sigma : A$ means the multi-set $\{\sigma : \phi\ |\ \phi \in A\}$, where 
		$\phi$ is a $\GDL$ formula in $A$. 
		Rule $(\sigma)$ demonstrates the compatibility of $\LDL$ to other existed dynamic-logic proof systems. The soundness of rule $(\sigma)$ is guaranteed by the following proposition, whose proof is directly by the irrelevance of $\sigma$ (Section~\ref{section:Dynamic Logic LDL}). 
		\fi
		
		Illustration of each rule is as follows. 
		
		Rules $([\alpha])$ and $(\la\alpha\ra)$ deal with dynamic parts of $\LDL$ formulas based on program transitions. 
		Both rules rely on the sub-proof-procedures of system $P_{\Oper, \Terminate}$. 
		In rule $([\alpha])$, 
		$\{...\}_{(\alpha', \sigma')\in \Phi}$ represents the collection of premises for all program states $(\alpha', \sigma')\in \Phi$. 
		By the finiteness of branches of program hehaviours (Definition~\ref{def:Program Properties}-\ref{item:Branching Finiteness}), set $\Phi$ must be finite. 
		So rule $([\alpha])$ only has a finite number of premises. 
		Note that when $\Phi$ is empty, the conclusion terminates.  
		Intuitively, rule $([\alpha])$ says that to prove that $[\alpha]\phi$ holds under configuration $\sigma$, we prove that 
		for each transition from $(\alpha, \sigma)$ to $(\alpha', \sigma')$, 
		$[\alpha']\phi$ holds under configuration $\sigma'$. 
		Compared to rule $([\alpha])$, rule $(\la\alpha\ra)$ has only one premise for $(\alpha', \sigma')$. 
		Intuitively, rule $(\la\alpha\ra)$ says that to prove $\la\alpha\ra\phi$ holds under configuration $\sigma$, we prove that
		after some transition from $(\alpha, \sigma)$ to $(\alpha', \sigma')$, $\la\alpha'\ra\phi$ holds under configuration $\sigma'$. 
		
		\ifx
		we prove that there exists a program state $(\alpha', \sigma')$ such that 
		$(\alpha, \sigma)\trans (\alpha', \sigma')$, and 
		$\la\alpha'\ra\phi$ holds under configuration $\sigma'$.  
		At the same time, we need to guarantee that the termination factor $t'$ cannot grow larger than $t$ w.r.t. relation $\prec$. 
		Note that in both rules, we assume that there always exists a transition $(\alpha, \sigma)\trans (\alpha', \sigma')$ from program state $(\alpha, \sigma)$ since $\alpha$ is neither $\ter$ nor $\abort$. This is guaranteed by the well-definedness of $\alpha$'s operational semantics, as stated in~\ref{item:Well-definedness} of Definition~\ref{def:Program Properties}. 
		\fi
		
		In rule $(\textit{Ter})$, each formula in $\Gamma$ and $\Delta$ is a formula in $\Fmla$.  
		Rule $(\textit{Ter})$ terminates if 
		formula $\mfr{P}(\Gamma\Rightarrow \Delta)$ is valid. 
		The introduction rule $(\textit{Int})$ for labels is applied when $\phi$ is a formula without any dynamic parts. 
		Through this rule we eliminate a configuration $\sigma$ and obtain a formula 
		$\app(\sigma, \phi)$ in $\Fmla$. 
		Rule $([\ter])$ deals with the situation when the program is a termination $\ter$. 
		Its soundness is straightforward by the well-definedness of program behaviours $\Oper$. 
		\ifx
		Rules $([\ter])$, $(\la\ter\ra)$ and $(\sigma[\abort])$ deal with the situations when program $\alpha$ is either a termination $\ter$ or an abortion $\abort$.
		Note that there is no rules for formula $\sigma: \la\abort\ra \phi$. 
		Intuitively, 
		it is not hard to see that $\la \abort\ra\phi$ is false under any configuration because $\abort$ never terminates. 
		\fi
		
		Rule $(\Sub)$ describes a specialization process for $\LDL$ formulas. 
		For a set $A$ of formulas, $\Sub(A)\dddef \{\Sub(\phi)\ |\ \phi\in A\}$, with $\Sub$ a function defined as follows in Definition~\ref{def:Abstract Substitution}. 
		Intuitively, if formula $\mfr{P}(\Gamma\Rightarrow \Delta)$ is valid, then its one of special cases $\mfr{P}(\Sub(\Gamma)\Rightarrow \Sub(\Delta))$ through an abstract version of substitutions $\Sub$ on formulas is valid. 
		Rule $(\Sub)$ plays an important role in constructing a bud in a cyclic preproof structure (Section~\ref{section:Construction of A Cyclic Preproof Structure}). 
		See Section~\ref{section:Example One: A While Program} as an example.  
		
		\begin{definition}[Substitution]
			\label{def:Abstract Substitution}
			A structural function $\eta : \TA\to \TA$ is called a `substitution',  if 
			for any evaluation $\rho\in \Eval$, there exists a $\rho'\in \Eval$ such that 
			$\rho(\eta(\phi))\equiv \rho'(\phi)$ for each formula $\phi\in \DLF$. 
		\end{definition}
		
		Rules $(\sigma \neg)$ and $(\sigma\wedge)$ are for transforming labeled formulas. Their soundness are direct according to the semantics of labeled formulas (Definition~\ref{def:Semantics of Labelled LDL Formulas}). 
		
		From rule $(\textit{ax})$ to $(\sigma \wedge R)$ are the rules inherited from traditional first-order logic. 
		Note that rule $(\textit{Cut})$ and rules $(\textit{Wk} L)$ and $(\textit{Wk} R)$ have no target pairs. 
		The meanings of these rules are classical and we omit their discussions here. 
		
		\ifx
		Rules from $(\textit{ax})$ to $(\vee)$ are direct from the traditional first-order logic. 
		Note that both rules $(\neg L)$ and $(\neg R)$ are required to derive the rules for the left-side target formulas of sequents. 
		Rules $(\textit{Sub} L)$ and $(\textit{Sub} R)$ are the versions of substitutions without quantifiers.  
		The meanings of these rules are classical and we omit them here. 
		\fi
		
		\ifx
		Rules $(\textit{Ter})$ and $(\textit{ax})$ declare a termination of a proof branch. 
		In rule $(\textit{Ter})$, 
		when each formula in $\Gamma$ and $\Delta$ is a formula in $\Fmla$.  
		We can conclude the proof branch if 
		the proof obligation $\mfr{P}(\Gamma\Rightarrow \Delta)$ is true. 
		Rule $(\textit{ax})$ is the `labeled' version of the corresponding rule in traditional propositional logic. 
		Rule $(\textit{Cut})$, the labeled version of the corresponding traditional ``cut rule'', provides a mechanism to derive by providing additional lemmas. 
		\fi

		\ifx
		The introduction rule $(\textit{Int})$ for labels is applied when $\phi$ is a formula in $\Fmla$ without any dynamic parts. 
		Through this rule we eliminate a configuration $\sigma$ and obtain a formula 
		$\app(\sigma, \phi)$ in $\Fmla$. 
		\ifx
		Rule $(\sigma\textit{Lif})$ lifts a rule in specific domains to a rule labelled with a free configuration $\sigma$ w.r.t. set $\Conf_\free(\Sigma)$ (Definition~\ref{def:free configurations}), where $\Sigma \dddef \Gamma\cup \Delta\cup \{\phi, \psi\}$.  
		We define $\sigma : A\dddef \{\sigma : \phi\ |\ \phi\in A\}$ for a multi-set $A$ of formulas. 
		Rule $(\sigma\textit{Lif})$ is useful, when a structure-based rule exists, we can simply lift it and make derivations in its labelled forms. 
		\fi
		Rule $(\sigma \cfeq)$ replaces a configuration $\sigma$ with one that is equivalent as $\sigma$ under evaluations. 
		It is critical for the proof system because by looking for a suitable $\sigma'$, 
		one can successfully find a back-link, thus closes one branch in a cyclic proof structure (these concepts will be introduced in Section~\ref{section:Construction of A Cyclic Preproof Structure}). 
		In Section~\ref{section:Case Studies}, our examples will illustrate how to construct a cyclic proof structure making use of this rule. 
		\fi

		\ifx
		Rules $(\neg\sigma L)$, $(\neg\sigma R)$, $(\wedge\sigma)$ and $(\vee\sigma)$ deal with logical connectives $\neg$, $\wedge$ and $\vee$ in a labelled $\LDL$ formula.
		They correspond to the rules in traditional propositional logic (without labels) for $\neg$, $\wedge$ and $\vee$ respectively.  
		\fi
		
		\ifx
		The lifting rule $(\sigma)$ lifts a rule $(\textit{sp})$ in specific domains to a rule labelled with an irrelevant configuration $\sigma$. 
		$\sigma : A$ means the multi-set $\{\sigma : \phi\ |\ \phi \in A\}$, where 
		$\phi$ is a $\GDL$ formula in $A$. 
		Rule $(\sigma)$ demonstrates the compatibility of $\LDL$ to other existed dynamic-logic proof systems. The soundness of rule $(\sigma)$ is guaranteed by the following proposition, whose proof is directly by the irrelevance of $\sigma$ (Section~\ref{section:Dynamic Logic LDL}). 
		\fi
		
		\ifx
		\begin{proposition}
			If a configuration $\sigma$ is irrelevant to a $\GDL$ formula $\phi$, 
			then $\models \sigma : \phi$ iff $\models \phi$.  
		\end{proposition}
		\fi

		The soundness of the rules in $\pfDLp$ is stated as follows. 
		
		\begin{theorem}
			\label{theo:soundness of rules for LDL}
			Provided that proof system $P_{\Oper, \Terminate}$ is sound, all rules in $\pfDLp$ (Table~\ref{table:General Rules for LDL}) are sound. 
		\end{theorem}
		
		Following the above explanations, Theorem~\ref{theo:soundness of rules for LDL} can be proved according to the semantics of $\LDL$ in Section~\ref{section:Syntax and Semantics of LDL}. 
		Appendix~\ref{section:Other Propositions and Proofs} gives the proofs for rules $([\alpha])$ and $(\la\alpha\ra)$ as examples. 
		
		\ifx
		Despite reminds for some special rules, the soundness of all the rules can be easily obtained according to 
		the semantics of $\LDL$ in Section~\ref{section:Semantics of LDL}. We will omit any details about it in this paper. 
		\fi

		\ifx
		Rules $(\textit{Int})$ and $(\sigma\textit{Eli})$ are for introducing and eliminating labels for $\LDL$ formulas. 
		They bridge the general labelled derivations proposed in Table~\ref{table:General Rules for LDL} and unlabelled derivations which are special for particular programs. 
		Rule $(\textit{Int})$ is an introduction rule for labels. 
		It is used when $\phi$ is a formula in $\Fmla$ without any dynamic parts, 
		to prove that $\phi$ holds under configuration $\sigma$, we only need to prove 
		that the explanation of $\phi$ under $\sigma$ is a valid formula. 
		Rule $(\sigma\textit{Eli})$ is an elimination rule for labels, where $X$ is a fresh configuration variable.  
		This rule is required when we wish to apply labelled rules for a $\LDL$ formula. 
		\fi

		\ifx
		Rules $(\sigma R)$ and $(\sigma L)$ are for introducing and eliminating a free-variable label $X$ for a formula $\phi$ during a deductive procedure.
		Their soundness is a direct result from Proposition~\ref{prop:basic idea}.
		Intuitively, if formula $\phi$ holds under any configuration $X$, then it also holds for a specific configuration.
		Rules $(\sigma R)$ introduces a free-variable label for a formula $\phi$ on the right side of a sequent.
		This is useful when deriving a non-labelled formula through the operational semantics of programs in $\LDL$ (will be seen in the example ??? of Section~\ref{???}).
		Rule $(\sigma L)$ eliminates a free-variable label for a formula $\phi$ on the left side of a sequent.
		\fi
		
		\ifx
		Rule $(\textit{Con})$ declares the congruence of labels w.r.t. $\LDL$ formulas. 
		It means that if we have a derivation between the $\LDL$ formulas $\phi$ and $\psi$, 
		we also have the same derivation between their labelled versions: $\sigma: \phi$ and $\sigma : \psi$. 
		This rule is useful when we have other derivations about unlabelled $\LDL$ formulas, we can also apply them to labelled $\LDL$ formulas. 
		As indicated in Section~\ref{section:Summery}, 
		rule $(\textit{Con})$ allows $\LDL$ to support the traditional syntactic-based derivations between labelled $\LDL$ formulas. 
		\fi
		
		\ifx
		Rule $(\sigma \sigma)$ is a direct result from Proposition~\ref{???}.
		Intuitively, if we can derive $\psi$ from $\phi$, then we can also derive $\psi:\sigma$ from $\phi : \sigma$ for any label $\sigma$.
		This conditional inference rule is quite useful: One can conduct the same derivation for formula $\phi$ labelled by $\sigma$ according to an already-existed rule for the unlabelled $\phi$.
		As already indicated in Section~\ref{section:Summery}, it also shows that labelled $\LDL$ formulas are compatible with traditional syntactic-based inferences.
		\fi
		
		\ifx
		Four rules $(-\sigma R)$, $(- \sigma L)$, $(+\sigma R)$ and $(+\sigma L)$ are for introducing and eliminating a label $\sigma$ for a formula $\phi$ during a deductive proceduce.
		Their soundness is a direct result from Proposition~\ref{???}.
		Rules $(- \sigma R)$ and $(-\sigma L)$ introduce a label for a formula $\phi$ on either side of a sequent.
		They are useful when one wants to conduct derivations using the rules for labelled formulas.
		This is usually the first step when deriving a non-labelled formula through the operational semantics of programs in $\LDL$ (will be seen in the example ??? of Section~\ref{???}).
		In rule $(-\sigma L)$, when $\phi$ is on the right side of a sequent, the label can be of any term $\sigma$, not just a free variable $X$.
		This is because
		\fi

		
		\ifx
		This is usually the first step when we want to derive a non-labelled formula through the operational semantics of programs in $\LDL$.
		Note that in $\Conf$ we require that there must exist a set of configuration variables, and any transitions from $(\alpha, X)$ are possible
		because the operational semantics of any programs in $\Prop$ is assumed to be well defined.
		\fi
		
		\ifx
		\subsection{Infiniteness of Proof Structures in $\LDL$ and Well-founded Relations}
		\label{section: Well-founded Relations}
		
		Based on the semantics of $\LDL$ given in Section~\ref{section:Semantics of LDL}, 
		it is not hard to prove that each proof rule of Table~\ref{table:General Rules for LDL} is sound. 
		When each branch of a proof tree terminates with a valid sequent in $\Fmla$, 
		the whole proof is sound.

		\ifx
		The proof rules of Table~\ref{table:General Rules for LDL} provides a general framework for reasoning about programs through operational semantics. 
		Each branch of a proof tree terminates when it ends with a valid sequent $\Gamma\Rightarrow \Delta$ in a discussed domain $\Assn$. 
		A sound proof is obtained when all branches of its proof tree terminate. 
		\fi
		\ifx
		However, during the proof of a set of specific programs and formulas, additional rules that are special for certain program and formula structures are needed.
		For example, as mentioned above, if a first-order quantifier $\forall$ appears in formulas in $\Fmla$, then we need another rules for labelled formulas with quantifiers, such as:
		$$
		???.
		$$
		The same goes for programs,
		as will be seen later in case studies in Section~\ref{???}, when we need to perform term rewrites or deductions by syntactic structures.
		\fi
		
		However, a branch of a proof tree in $\LDL$ system does not always terminate, since 
		the process of symbolically executing a program via rule $([\alpha])$ or/and rule $(\la\alpha\ra)$ might not stop. 
		This is well known when a program has a loop structure that may runs infinitely, 
		for example, a while program $\textit{while}\ \textit{true}\ \textit{do}\ x := x + 1\ \textit{end}$. 
		
		To avoid potentially infinite derivations in $\LDL$, in this paper, we adopt the so-called \emph{cyclic proof} approach (cf.~\cite{Brotherston07}), a technique to insure a valid conclusion even when its proof tree contains infinite derivations. 
		A \emph{preproof} is 
		a proof tree with a finite structure (which means having a finite number of nodes) but containing infinite derivation paths. 
		A preproof structure can lead to a valid conclusion if a certain soundness condition is met. 
		This condition guarantees that any counterexamples from an invalid conclusion would cause an infinite descent sequence w.r.t. a well-founded relation (Definition~\ref{def:Well-foundedness}) that is related to the semantics of $\LDL$, which, however, is a contradiction to the definition of the well-foundedness itself. 
		We will expand it in detail in Section~\ref{section:Construction of A Cyclic Preproof Structure}. 
		
		\ifx
		The main idea behind this approach is that, 
		if the preproof structure does not lead to a sound conclusion, then it can be shown that each derivation path will cause an infinite descent sequence w.r.t. a well-founded relation (Definition~\ref{def:Well-foundedness}) that is related to the semantics of $\LDL$, which, however, is a contradiction to the definition of the well-foundedness itself. 
		\fi
		
		\ifx
		However, the derivations of a $\LDL$ formula is potentially infinite, as the process of symbolically executing a program $\alpha$ via applying rule $([\alpha])$ or/and rule $(\la\alpha\ra)$ might not terminate.
		This is well known when a program has a loop structure that may runs infinitely, e.g., a while program $\textit{while}\ \textit{true}\ \textit{do}\ x := x + 1\ \textit{end}$. 
		To avoid potentially infinite derivations in $\LDL$, we adopt the so-called \emph{cyclic proof} approach~\cite{???}, a technique to insure a valid proof even when its proof tree contains infinite derivations. 
		A proof tree with a finite structure (which means having a finite number of nodes) but containing infinite derivation paths is called a \emph{preproof}. 
		The main idea behind this approach is that, 
		if the preproof structure is not a sound proof, then it can be shown that each derivation path will lead to an infinite descent sequence w.r.t. a well-founded relation that is related to the semantics of $\LDL$, which, however, is a controdiction to the definition of well-foundedness relations itself. 
		\fi
		
		\ifx
		when certain conditions are met for a preproof, we can claim that the possibly-infinite derivations induced by the preproof structure lead to a valid conclusion.
		Otherwise, it would cause an infinite descent sequence in a well-founded set related to the semantics of $\LDL$, which is a contradicton (as we will prove).
		\fi
		
		To apply the cyclic proof approach, 
		we need 3 critical well-founded relations that are related to the semantics of $\LDL$. 
		One of them, relation $\prec$, has already been introduced in Section~\ref{section:Definitions of LDL}.  
		With these relations, 
		in Section~\ref{section:Construction of A Cyclic Preproof Structure}, we will define a preproof structure, and give the soundness condition in which a preproof is a sound proof. 
		
		\ifx
		In the following, we first introduce the so-called ``well-founded'' sets/relations, a key concept related to a circular proof.
		Then we define a preproof structure, and give the conditions in which a preproof is a proof. That is, it can lead to a valid conclusion.
		\fi
		

		\ifx
		Following~\cite{Aczel77,Brotherston07} (Proposition 1.2.1), the relation $\prec$ between inductive terms introduced in Definition~\ref{def:Definition of Relation Between Inductive Terms} is a well-founded relation in $\Pred$. 
		\fi
		
		In the following, we introduce another 2 well-founded relations. 
		The first relation $\psuf$ is the suffix relation between two paths;
		while the second relation $\pmult$ is between two finite sets of paths.
		
		\ifx
		\paragraph{Relation $\psuf$}
		As mentioned in Section~\ref{section:Construction of LDL}, in $\Conf$ we assume a set of inductively-defined predicates $\Pred$. 
		In $\Pred$, we introduce an ordering $\psuf$ as follows.
		\fi
		
		\ifx
		\begin{definition}
			In $\Pred$, a relation between any two terms $P_1$ and $P_2$, denoted by 
			$P_1\psuf P_2$, is defined such that either 
			(1) there exists a rule $(\mfr{P}, P_2)\in \Phi$ with $P_1\in \mfr{P}$, 
			or (2) there is another predicate $P$ such that $P_1\psuf P\psuf P_2$. 
		\end{definition}
		
		Following~\cite{An introduction to Inductive Definitions} (Proposition 1.2.1), it is not hard to see that the relation $\psuf$ is well-founded in $\Pred$.
		\fi
		
		\ifx
		\begin{proposition}
			For any predicate $P\in \Pred$, there exists only a finite number of predicates $P_1,...,P_k$ such that $P\psufr P_1\psufr ... \psufr P_k$, with $P_k$ a minimum element (w.r.t. $\psuf$) of $\Pred$.
		\end{proposition}
		\fi
		
		\ifx
		Formally, for any predicate $P\in \Pred$, there exists only a finite number of predicates $P_1,...,P_k$ such that $P\psufr P_1\psufr ... \psufr P_k$, with $P_k$ a minimum element of $\Pred$.
		\fi
		\ifx
		$P_1\psuf P_2$ if either (1) there exists a rule $(\mfr{P}, P_2)\in \Phi$ with $P_1\in \mfr{P}$, or (2) there is another predicate $P$ such that $P_1\psuf P\psuf P_2$.
		Following~\cite{An introduction to Inductive Definitions} (Proposition 1.2.1), it is not hard to see that the relation $\psuf$ is well-founded in $\Pred$.
		Formally, for any predicate $P\in \Pred$, there exists only a finite number of predicates $P_1,...,P_k$ such that $P\psufr P_1\psufr ... \psufr P_k$, with $P_k$ a minimum element of $\Pred$.
		\fi
		
		\begin{definition}[Relation $\psuf$]
			Given two paths $tr_1$ and $tr_2$, 
			relation $tr_1 \psuf tr_2$ is defined if $tr_2$ is a proper suffix of $tr_1$. 
			Write $tr_1\suf tr_2$ if $tr_1$ is a suffix of $tr_2$. 
		\end{definition}
		
		It is trivial that in a set of finite paths, relation $\psuf$ is well-founded. Because 
		every finite path has only a finite number of suffixes. 
		
		\ifx
		\begin{proposition}
			For any finite path $tr$, there exists only a finite number of paths $tr_1,...,tr_k$ such that $tr\psufr tr_1\psufr ... \psufr tr_k$, with $tr_k$ a minimum element (w.r.t. $\psuf$) of $\Pred$.
		\end{proposition}
		\fi
		
		Relation $\pmult$ is based on the definition of relation $\psuf$. 
		
		\begin{definition}[Relation $\pmult$]
			\label{def:Relation pmult}
			Given two finite sets $\cnt_1$ and $\cnt_2$ of finite paths,  
			$\cnt_1\pmult \cnt_2$ is defined if set $\cnt_1$ can be obtained from $\cnt_2$ by replacing (or removing) one or more elements of $\cnt_2$ with a finite number of elements, such that
			for each replaced element $tr\in \cnt_2$, its replacements $tr_1,...,tr_n$ ($n\ge 1$) in $\cnt_1$ satisfies that $tr_i\psuf tr$ for any $i$, $1\le i\le n$. 
			
			
			Write $\cnt_1\mult \cnt_2$ if $\cnt_1 = \cnt_2$ or $\cnt_1\pmult \cnt_2$. 
		\end{definition}
		
		Note that in Definition~\ref{def:Relation pmult} it is not hard to see that $\cnt_1\pmult \cnt_2$ implies $\cnt_1\neq \cnt_2$.

		\ifx
		\paragraph{Relation $\pmult$}
		Given two execution traces $tr_1$ and $tr_2$, let 
		$tr_1 < tr_2$ represent that trace $tr_2$ is a suffix of trace $tr_1$. 
		Define $\cnt(\alpha, \sigma, \phi)$ to be a finite set of execution traces $t_1t_2...t_n$ 
		satisfying that 
		(1) it starts from $t_1 = (\alpha, \sigma)$ and ends with $t_n = (\ter, \sigma')$, and $\Sem{\phi}_{\sigma'}$ is not valid;
		(2) there is no trace $tr'$ such that $tr' < t_1t_2...t_n$ and $tr'\in \cnt(\alpha, \sigma, \phi)$. 
		\fi
		
		\ifx
		Let $\cnt(\alpha, \sigma, \phi)$ be a finite set of execution traces $t_1t_2...t_n$ starting from $t_1 = (\alpha, \sigma)$ and ending with $t_n = (\ter, \sigma')$, in which $\sigma'$ makes $\Sem{\phi}_{\sigma'}$ not valid.
		In addition, it satisfies that such a trace $t_1t_2...t_n$ is a shortest such trace in the sense that there is no any of its subtraces belonging to $\cnt(\alpha, \sigma, \phi)$.
		In other words, let $tr_1 < tr_2$ represent that trace $tr_2$ is a suffix of trace $tr_1$, then there is no trace $tr'$ such that $tr' < t_1t_2...t_n$ and $tr'\in \cnt(\alpha, \sigma, \phi)$.
		\fi
		
		\ifx
		From the suffix relation $<$ and the definition of $\cnt(\alpha, \sigma, \phi)$ we induce the ordering $\pmult$. 
		
		\begin{definition}[Relation $\psuf$]
			Given two finite sets of traces $\cnt_1(\alpha_1, \sigma_1, \phi_1)$ and $\cnt_2(\alpha_2, \sigma_2, \phi_2)$, $\cnt_1\pmult \cnt_2$ is defined if 
			set $\cnt_1$ can be obtained from $\cnt_2$ by replacing a finite set of traces $A$ of $\cnt_2$ with a set of traces $A'$. 
			The function $': A\to A'$ satisfies that 
			for any trace $tr\in A$, $tr' < tr$. 
			
		\end{definition}
		\fi
		\ifx
		$C_1\pmult C_2$ if $C_1$ can be obtained by replacing one or more traces of $C_2$ with a finite number of traces (including an empty set), each of which is smaller than the replaced traces in the sense of the relation $<$.
		\fi
		
		For example, let $C_1=\{tr_1, tr_2, tr_3\}$, where $tr_1 \dddef (\alpha, \sigma)\trans(\alpha_1, \sigma_1)\trans(\alpha_2, \sigma_2)\trans(\alpha_3, \sigma_3)\trans(\ter, \sigma_4), tr_2 \dddef (\alpha, \sigma)\trans(\alpha_1, \sigma_1)\trans(\beta_1, \delta_1)\trans(\beta_2, \delta_2)\trans(\ter, \delta_3)$ and $tr_3 \dddef (\alpha, \sigma)\trans(\ter, \tau)$; 
		$C_2 = \{tr'_1, tr'_2\}$, where 
		$tr'_1 \dddef (\alpha_1, \sigma_1)\trans(\alpha_2, \sigma_2)\trans(\alpha_3, \sigma_3)\trans(\ter, \sigma_4), tr'_2 \dddef (\alpha_1, \sigma_1)\trans(\beta_1, \delta_1)\trans(\beta_2, \delta_2)\trans(\ter, \delta_3)$. 
		We see that $tr'_1 \psuf tr_1$ and $tr'_2 \psuf tr_2$. 
		$C_2$ can be obtained from $C_1$ by replacing $tr_1$ and $tr_2$ with $tr'_1$ and $tr'_2$  respectively, and removing $tr_3$. 
		Hence $C_2\pmult C_1$. 
		
		Relation $\pmult$ is in fact a special case of the ``multiset ordering'' introduced in~\cite{Dershowitz79}, where it has shown that multiset ordering is well-founded. 
		Therefore, $\pmult$ is well-founded. 
		
		
		\ifx
		Relation $\prec_3$ is actually a special case of the well-founded relation on multisets proposed in~\cite{Proving termination on multiset orderings}. 
		We omit the proof of the well-foundedness of $\prec_3$. 
		\fi
		
		\ifx
		To see $\prec_3$ is indeed well-founded, observe that since $\pmult$ on a set of finite paths is well-founded, for any finite path $tr$ with length $n$, there are only $n$ paths that are its suffixes. So if $\cnt_1\prec_3 \cnt_2$, then 
		\fi
		
		
		\ifx
		\begin{proposition}
			Starting from a set $C(\alpha, \sigma, \phi)$, there exists only a finite number of sets $C_1, ..., C_k$ such that $C\pmultr C_1\pmultr ... \pmultr C_k$, with $C_k$ a minimum element.
		\end{proposition}
		\fi
		\ifx
		Formally, starting from a set $C(\alpha, \sigma, \phi)$, there exists only a finite number of sets $C_1, ..., C_k$ such that $C\pmultr C_1\pmultr ... \pmultr C_k$, with $C_k$ a minimum element.
		\fi
		
		\fi
		
		\ifx
		\begin{definition}[Cyclic Condition for Extra Rules]
			A rule in $\mcl{E}$ satisfies ``cyclic condition'', if it is of the form:
			$$
			\begin{aligned}
				\infer[^{(\sigma [\textit{Ext}])}]
				{\Gamma\Rightarrow \sigma : [\alpha]\phi, \Delta}
				{\Gamma'\Rightarrow \sigma' : [\alpha]\phi, \Delta'}
			\end{aligned}
			\mbox{\ \ or\ \ }
			\begin{aligned}
				\infer[^{(\sigma \la\textit{Ext}\ra)}]
				{\Gamma\Rightarrow \sigma : \la\alpha\ra\phi, \Delta}
				{\Gamma'\Rightarrow \sigma' : \la\alpha\ra\phi, \Delta'}
			\end{aligned}, 
			$$
			which satisfies that 
			for any $\rho\in \Eval$ with $\rho\models \Gamma$, there is a $\rho'\in \Eval$ such that $\rho'\models \Gamma'$ and $\rho'(\sigma)\equiv \rho(\sigma)$. 
			
			We call target pair $(\sigma : [\alpha]\phi, \sigma':[\alpha]\phi)$ (resp. $(\sigma : \la\alpha\ra\phi, \sigma':\la\alpha\ra\phi)$) a ``critical pair'' of this rule. 
		\end{definition}
		\fi
		
		\subsection{Construction of A Cyclic Preproof Structure in $\pfDLp$}
		\label{section:Construction of A Cyclic Preproof Structure}
		\ifx
		A \emph{proof tree} is a tree structure constructed according to the proof rules of Table~\ref{table:General Rules for LDL} and/or other additional proof rules for a special type of programs.
		The root of a proof tree is the conclusion of a proof.
		Each leave node of a proof tree is either \emph{terminal}, if it is an instance of an axiom rule, or \emph{non-terminal} otherwise.
		A non-terminal node is also called a \emph{bud}~\cite{???}.
		\fi

		In the proof system $\pfDLp$, given a sequent $\Gamma\Rightarrow \Delta$, we expect 
		a finite proof tree to prove $\Gamma\Rightarrow \Delta$. 
		However, a branch of a proof tree does not always terminate.
		Because 
		the process of symbolically executing a program via rule $([\alpha])$ or/and rule $(\la\alpha\ra)$ might not stop. 
		This is well known when a program has an explicit/implicit loop structure that may run infinitely. 
		For example, a while program $W \dddef \textit{while}\ \textit{true}\ \textit{do}\ x := x + 1\ \textit{end}$ will proceed infinitely as the following transitions:
		$(W, \{x \mapsto 0\})\trans (W, \{x\mapsto 1\})\trans ...$. 
		When deriving it, 
		one may generate an infinite proof branch like this: 
		$$
		\begin{aligned}
			\infer[^{([\alpha])}]
			{\cdot\Rightarrow \{x\mapsto 0\} : [W]\true}
			{
				\infer[]
				{\cdot\Rightarrow \{x\mapsto 1\} : [W]\true}
				{
					\infer*[]
					{}
					{
						\infer[]
						{}
						{
							\infer[^{([\alpha])}]
							{\cdot\Rightarrow \{x\mapsto n\} : [W]\true}
							{...}
						}
					}
				}
			}
		\end{aligned}
		$$
		
		In this paper we apply the cyclic proof approach (cf.~\cite{Brotherston07}) to deal with a potential infinite proof tree in $\LDL$. 
		Cyclic proof is a technique to ensure a valid conclusion of a special type of 
		infinite proof trees, called `preproof', under some certain conditions called ``soundness conditions''. 
		Below we firstly introduce a preproof structure, 
		then based on it we propose a `cyclic' preproof --- a structure that satisfies certain soundness conditions.  
		The key idea is that, if the conclusion is invalid, then the structure of a cyclic preproof, as will be shown in Section~\ref{section:Proof of Theorem - theo:Soundness of A Cyclic Preproof} in detail, would induce an infinite descent sequence of elements of a well-founded set, which would contradict the definition of `well-foundedness' (Section~\ref{section:Proof of Theorem - theo:Soundness of A Cyclic Preproof}) itself. 

		\ifx
		A \emph{preproof} is 
		a proof tree with a finite structure (which means having a finite number of nodes) but containing infinite derivation paths. 
		A preproof structure can lead to a sound conclusion if a certain soundness condition is met. 
		This condition guarantees that any counterexamples from an invalid conclusion would cause an infinite descent sequence w.r.t. a well-founded relation (Definition~\ref{def:Well-foundedness}) that is related to the semantics of $\LDL$, which, however, is a contradiction to the definition of the well-foundedness itself. 
		We will expand it in detail in Section~\ref{section:Construction of A Cyclic Preproof Structure}. 
		\fi
		
		
		\textbf{Preproof \& Derivation Traces}. 
		A \emph{preproof} (cf.~\cite{Brotherston07}) is an infinite proof tree expressed as a tree structure with finite nodes, in which there exist non-terminal leaf nodes, called \emph{buds}. 
		Each bud is identical to one of its ancestors in the tree. 
		The ancestor identical to a bud $N$ is called a \emph{companion of $N$}. 
		
		
		A \emph{derivation path} in a preproof is an infinite sequence of nodes $\nu_1\nu_2...\nu_m...$ ($m\ge 1$) starting from the root node $\nu_1$ of the preproof, 
		where for each node pair $(\nu_i, \nu_{i+1})$ ($i\ge 1$), 
		$\nu_i$ is the conclusion and $\nu_{i+1}$ is a premise, of a rule. 

		
		
		A \emph{derivation trace} over a derivation path $\mu_1\mu_2...\mu_k\nu_1\nu_2...\nu_m...$ ($k \ge 0, m\ge 1$) is an infinite sequence $\tau_1\tau_2...\tau_m...$ of formulas starting from formula $\tau_1$ in node $\nu_1$ such that 
		each CP pair $(\tau_i, \tau_{i+1})$ ($i\ge 1$) appearing on the right-sides of nodes $(\nu_i, \nu_{i+1})$ respectively is either 
		a target pair of a rule in $\pfDLp$, or a CP pair satisfying $\tau_i\equiv \tau_{i+1}$ of a rule in $\pfDLp$. 

		\ifx
		\begin{definition}[Derivation Traces]
			A ``derivation trace'' over a derivation path $\mu_1\mu_2...\mu_k\nu_1\nu_2...\nu_m...$ ($k \ge 0, m\ge 1$) is an infinite sequence $\tau_1\tau_2...\tau_m...$ of formulas starting from formula $\tau_1$ in node $\nu_1$ such that 
			each pair $(\tau_i, \tau_{i+1})$ ($i\ge 1$) with $\tau_i, \tau_{i+1}$ appearing on the right-sides of nodes $\nu_i, \nu_{i+1}$ respectively is either a target pair of an instance of a rule in $\mcl{DL}_p$, or a critical pair of an instance of a rule in $\mcl{E}$, or a non-target pair satisfying $\tau_i\equiv \tau_{i+1}$. 
			
		\end{definition}
		
		Note that rules $(\textit{Cut})$ and $(\textit{Wk})$ have no target pairs. So a pair of a derivation trace in an instance of these two rules must be a non-target pair. 
		\fi

		\textbf{Progressive Steps}. 
		A crucial prerequisite of defining cyclic preproofs is to introduce the notion of \emph{(progressive) derivation traces} in a preproof of $\pfDLp$. 
		\ifx
		The key idea is that if each derivation trace of a preproof structure is progressive, then it is impossible for the conclusion of the preproof to be invalid.
		Otherwise,
		an infinite descent sequence of counterexamples of finite sets of execution traces of a program can be constructed, contradicting the fact that the relation $\psuf$ or $\pmult$ mentioned in Section~\ref{section: Well-founded Relations} is well-founded.
		\fi
		
		\begin{definition}[Progressive Step/Progressive Derivation Trace in $\LDL$]
			\label{def:Progressive Step/Progressive Derivation Trace}
			In a preproof of $\pfDLp$, given a derivation trace $\tau_1\tau_2...\tau_m...$ over a derivation path $...\nu_1\nu_2...\nu_m...$ ($m\ge 1$) starting from $\tau_1$ in node $\nu_1$, 
			a formula pair $(\tau_i, \tau_{i+1})$ ($1\le i\le m$) over $(\nu_i, \nu_{i+1})$ is called a ``progressive step'', if $(\tau_i, \tau_{i+1})$ is a target pair 
			of an instance of rule $([\alpha])$: 
			$$
			\begin{gathered}
				\infer[^{([\alpha])}]
				{\nu_i: \Gamma\Rightarrow \tau_i : (\sigma : [\alpha]\phi), \Delta,}
				{
					...
					&
					\nu_{i+1}: \Gamma\Rightarrow \tau_{i+1} : (\sigma' : [\alpha']\phi), \Delta
					&
					...
				}
			\end{gathered}
			;
			$$
			or the target pair of an instance of rule $(\la\alpha\ra)$:
			$$
			\begin{gathered}
				\infer[^{(\la\alpha\ra)}]
				{\nu_i: \Gamma\Rightarrow \tau_i : (\sigma : \la\alpha\ra\phi), \Delta}
				{
					\nu_{i+1}: \Gamma\Rightarrow \tau_{i+1} : (\sigma' : \la\alpha'\ra\phi), \Delta
				}
			\end{gathered}
			, 
			$$
			provided with an additional condition that 
			$\vdash_{P_{\Oper,\Terminate}}(\Gamma\Rightarrow (\alpha', \sigma')\termi, \Delta)$. 
			
			\ifx
			are the formulas appearing in the sequent pair $(\nu_i, \nu_{i+1})$ respectively, which is either
			$$
			\begin{gathered}
				\nu_i: \Gamma\Rightarrow \sigma : [\alpha]\phi, \Delta, \\
				\nu_{i+1}: \Gamma\Rightarrow \sigma' : [\alpha']\phi, \Delta,
			\end{gathered}
			$$
			of an instance of rule $([\alpha])$
			where $\tau_i$ is formula $\sigma : [\alpha]\phi$ and $\tau_{i+1}$ is formula $\sigma' : [\alpha']\phi$; 
			or
			$$
			\begin{gathered}
				\nu_i: \Gamma \Rightarrow \sigma: \la\alpha\ra\phi, \Delta,\\
				\nu_{i+1}: \Gamma\Rightarrow \sigma': \la\alpha'\ra\phi, \Delta,
			\end{gathered}
			$$
			of an instance of rule $(\la\alpha\ra)$ where
			$\tau_i$ is formula $\sigma: \la\alpha\ra\phi$, $\tau_{i+1}$ is formula $\sigma': \la\alpha'\ra\phi$ and $(\alpha', \sigma')\prec (\alpha, \sigma)$. 
			\fi

			\ifx
			A formula pair $(\tau_i, \tau_{i+1})$ ($i\ge 1$) of the derivation trace is called a ``progressive step'' if $(\tau_i, \tau_{i+1})$ is as in the following sequents of instances of the rules $([\alpha])$ and $(\la\alpha\ra)$ as follows:
			$$
			\begin{gathered}
				\infer[^{([\alpha])}]
				{\Gamma\Rightarrow \tau_i: (\sigma : [\alpha]\phi), \Delta}
				{
					\Gamma\Rightarrow \tau_{i+1}: (\sigma' : [\alpha']\phi), \Delta
				},\\ 
				\infer[^{(\la\alpha\ra)}]
				{\Gamma \Rightarrow \tau_i: (\sigma\{P\}: \la\alpha\ra\phi), \Delta}
				{
					\Gamma\Rightarrow \tau_{i+1}: (\sigma'\{P'\}: \la\alpha'\ra\phi), \Delta
				},
			\end{gathered}
			$$
			\mbox{ where $P, P'\in \Pred$ and $P'\psuf P$}, $\sigma\{Q\}$ means that $Q$ appears in $\sigma$. 
			\fi
			
			If a derivation trace has an infinite number of progressive steps, we say that the trace is `progressive'.
		\end{definition}

		As we will see in Section~\ref{section:Proof of Theorem - theo:Soundness of A Cyclic Preproof} and Appendix~\ref{section:Other Propositions and Proofs}, 
		the additional condition of the instance of rule $(\la\alpha\ra)$ is the key to prove the case for modality $\la\cdot\ra$ in Theorem~\ref{theo:Soundness of A Cyclic Preproof}.

		
		As an example, we propose a proof system for termination checking of \emph{While} programs shown in Table~\ref{table:Inference Rules for Program Terminations of While programs} of Appendix~\ref{section:Formal Definitions of While Programs}.

		\textbf{Cyclic Preproofs}. 
		We define cyclic preproof structures as the following definition. 
		
		\begin{definition}[A Cyclic Preproof for $\LDL$]
			\label{def:Cyclic Preproof}
			In $\pfDLp$, a preproof is a `cyclic' one, if there exists a progressive trace over every derivation path.
		\end{definition}
		
		The cyclic proof system $\pfDLp$ is sound, as stated in the following theorem. 
		
		\begin{theorem}[Soundness of the Cyclic Preproof for $\LDL$]
			\label{theo:Soundness of A Cyclic Preproof}
			Provided that $P_{\Oper, \Terminate}$ is sound, 
			a cyclic preproof of $\pfDLp$ has a sound conclusion.  
		\end{theorem}
		
		
		We still denote by $\vdash_{\pfDLp} \nu$ that $\nu$ can be proved through a cyclic preproof. 
		In Section~\ref{section:Proof of Theorem - theo:Soundness of A Cyclic Preproof}, we will prove Theorem~\ref{theo:Soundness of A Cyclic Preproof}.  
		
		Proof system $P_{\Oper, \Omega}$ itself that $\pfDLp$ relies on can also be cyclic. 
		In Appendix~\ref{section:Formal Definitions of While Programs}, we propose a cyclic proof system for termination checking of \emph{While} programs in Table~\ref{table:Another Type of Inference Rules for Program Terminations of While programs}, which does not rely on programs' syntactic structures. 
		

		\textbf{About Other Properties of $\pfDLp$}. 
		Since $\LDL$ is not a specific logic,
		it is impossible to discuss about its decidability, completeness or whether it is cut-free without any restrictions on sets $\Prog, \Conf$ and $\Fmla$. 
		For example, for the completeness, in $\Domain{\WP}$, whether a $\LDL$ formula can be proved depends on the definitions of $\Conf_\WP$, more precisely, on whether we can successfully find proper forms of configurations using rule $(\Sub)$ from which a preproof structure can be constructed. 
		One of our future work will focus on analyzing under what restrictions, 
		these properties can be obtained in a general sense. 
		
		\ifx
		About the completeness of $\pfDLp$, 
		since $\LDL$ is not a specific logic in normal sense,
		it is impossible to discuss it without explicit definitions of or any restrictions on structures $\Prog, \Conf$ and $\Fmla$. 
		The completeness depends on whether a cyclic preproof can be built for every valid $\LDL$ formula. 
		One of our future work is to analyze under which restrictions, a completeness result can be obtained for $\LDL$.  
		\fi
		
		\ifx
		since 
		$\LDL$ is not a specific logic in normal sense, 
		the completeness of $\pfDLp$ relies on the definitions of $\Prog, \Conf$ and $\AFmla$ particularly. 
		It depends on whether a cyclic preproof can be built for every valid $\LDL$ formula. 
		One of our future work will try to analyze under which conditions of $\Prog, \Conf$ and $\AFmla$, a completeness or relative completeness result for $\LDL$ can be achieved.  
		\fi
		
		\ifx
		An interesting problem is to ask whether or in what conditions our proof system is complete \emph{related to} the structures $\Prog, \Conf$ and $\Fmla$. 
		Currently, we think that this problem is rather challenging, 
		\fi

		\subsection{Lifting Rules From Dynamic Logic Theories}
		\label{section:Lifting Process From Program Domains}
		We introduce a technique of lifting rules from particular dynamic-logic theories to labeled ones in $\LDL$. 
		It makes possible for embedding (at least partial) existing  dynamic-logic theories into $\LDL$, 
		without losing their abilities of deriving based on programs' syntactic structures. 
		This in turn facilitates deriving $\LDL$ formulas in particular program domains by making use of special structural rules. 
		\ifx
		This is useful in two aspects: 
		firstly, it allows making use of existing rules to facilitate deriving $\LDL$ formulas in particular program domains;
		secondly, it makes possible for embedding (at least partial) existing  dynamic-logic theories into $\LDL$, 
		without losing their abilities of deriving based on programs' syntactic structures. 
		\fi
		Below, we propose a lifting process for general inference rules under a certain condition of configurations (Proposition~\ref{prop:lifting process 2}). 
		One example of applications of this technique will be given later in Section~\ref{section:Lifting Process in While Programs}. 
		
		\ifx
		In a particular program domains like $\TA_\WP$ (Section~\ref{section:Examples of Term Structures}), we are interested in adapting a special rule 
		to a labeled one as a part of extra theory $\mcl{E}$ of $\LDL$.
		As pointed out in Section~\ref{section:Summery}, it allows $\LDL$ to make use of existing rules of a particular theory to facilitate derivations of $\LDL$ formulas, thus providing a flexible verification framework in which both symbolic-execution-based reasoning and structural reasoning are possible. 
		Below we introduce a process of lifting logical consequences $\phi\to \psi$ and general inference rules under certain conditions. 
		Logical consequences are commonly found as axioms in dynamic logics and their variations (cf.~\cite{???}). 
		One example of applications of the lifting process will be given later in Section~\ref{section:Lifting Process in While Programs}. 
		\fi
		
		We first introduce the basic concepts. 
		In the following definitions of this section, we always assume $\rho\models \phi$ for an unlabeled formula $\phi$ is well defined in some particular dynamic-logic theory.

		\begin{definition}[Relation $\se$]
			\label{def:Same Effects}
			Given evaluations $\rho, \rho'\in \Eval$, a configuration $\sigma\in \Conf$ and a set $A$ of unlabeled formulas, 
			evaluation-configuration pair $(\rho, \sigma)$ ``has the same effect'' on $A$ as evaluation $\rho'$, denoted by $(\rho, \sigma) \se_A \rho'$, if for any $\phi\in A$, $\rho\models \sigma : \phi$ iff $\rho'\models \phi$. 
			
		\end{definition}
		
		
		\ifx
		\begin{definition}[Standard Configuration]
			\label{def:Simple Configurations}
			We say a configuration $\sigma\in\Conf$ is `standard' w.r.t. a set $A$ of formulas, if it satisfies that for any $\rho\in \Eval$, there exists a $\rho'\in \Eval$ such that 
			$(\rho, \sigma)\se_A \rho'$. 
			
			Simply call $\sigma$ `standard' if $A = \DLF$. We denote the set of all standard configurations w.r.t. $A$ as $\Conf_\std(A)$. 
		\end{definition}
		
		The effect of a standard configuration on formulas is equivalent to the effect of an evaluation. 
		\fi
		

		\begin{example}
			\label{example:normal configurations}
			In program domain $\Domain{\WP}$, let $\rho, \rho'\in \Eval_\WP$ be two evaluations satisfying $\rho(x) = 5, \rho(y) = 1$ and $\rho'(x) = 6, \rho'(y) = 1$. 
			Let $\sigma = \{x\mapsto x + 1\}\in \Conf_\WP$, formula $\phi = (x + y + 1 > 7)\in \Fmla_\WP$. 
			Then $(\rho, \sigma)\se_{\{\phi\}}\rho'$, since 
			$\rho(\app_\WP(\sigma, \phi))\equiv \rho((x + 1) + y + 1 > 7)\equiv (8 > 7)\equiv \rho'(\phi)$. 
			
			\ifx
			$\sigma$ is standard, because for any evaluation $\xi\in \Eval_\WP$, let $\xi'\dddef \xi^x_{\xi(x) + 1}$, 
			then we always have $(\xi, \sigma)\se \xi'$. 
			It is not hard to see that 
			the configurations in $\Conf_\WP$ and $\Conf_\E$ introduced in Section~\ref{section:Examples of Term Structures} are actually standard. 
			See a proof for $\Conf_\WP$ in Lemma~\ref{???} (Appendix~\ref{section:Formal Definitions of While Programs}). 
			\fi
		\end{example}


		\begin{definition}[Free Configurations]
			\label{def:free configurations}

			A configuration $\sigma$ is called `free' w.r.t. a set $A$ of formulas if it satisfies that 
			\begin{enumerate}
				\item\label{item:free configurations cond 1} for any $\rho\in \Eval$, there is a $\rho'\in \Eval$ such that $(\rho, \sigma)\se_A \rho'$; 
				\item\label{item:free configurations cond 2} for any $\rho\in \Eval$, there is a $\rho'\in \Eval$ such that 
				$(\rho', \sigma)\se_A \rho$.
			\end{enumerate}
			
			We denote the set of all free configurations w.r.t. $A$ as $\Conf_\free(A)$. 
		\end{definition}
		
		Intuitively, a free configuration in some sense neither strengthen nor weaken a formula w.r.t. its boolean semantics after affections. 
		
		For a configuration $\sigma\in \Conf_\WP$, 
		recall that $e^x_e$ is explained in Example~\ref{example:Proof system for program behaviours}.  
		
		\begin{example}
			In Example~\ref{example:normal configurations}, 
			the configuration $\sigma$ is free w.r.t. $\{\phi\}$.
			On one hand, for any evaluation $\xi\in \Eval_\WP$, let $\xi'\dddef \xi^x_{\xi(x) + 1}$, 
			then we have $(\xi, \sigma)\se_{\{\phi\}} \xi'$ (Definition~\ref{def:free configurations}-\ref{item:free configurations cond 1});
			On the other hand, 
			for any evaluation $\xi$, let $\xi'\dddef \xi^x_{\xi(x) - 1}$, 
			then $(\xi', \sigma)\se_{\{\phi\}} \xi$ holds (which is Definition~\ref{def:free configurations}-\ref{item:free configurations cond 2}). 
			Let $\sigma' = \{x \mapsto 0, y\mapsto 0\}$, then 
			$\sigma'$ is not free w.r.t. $\{\phi\}$, since it makes $\phi$ too strong so that there does not exist an evaluation $\xi$ such that 
			$(\xi, \sigma')\se_{\{\phi\}} \rho'$, violating Definition~\ref{def:free configurations}-\ref{item:free configurations cond 2}. 
			
		\end{example}
		
		\ifx
		The lifting process for logical consequences $\phi\to \psi$ and inference rules are described by the following two propositions respectively. 
		
		\begin{proposition}
			\label{prop:lifting process 1}
			Given any formulas $\phi, \psi$, if formula $\phi\to \psi$ is valid, 
			then rule 
			$$
			\begin{aligned}
				\infer[]
				{\Gamma\Rightarrow \sigma : \psi, \Delta}
				{\Gamma\Rightarrow \sigma : \phi, \Delta}
			\end{aligned}
			$$
			is sound for any configuration $\sigma\in \Conf_\std(\Gamma\cup \Delta\cup\{\phi, \psi\})$. 
		\end{proposition}
		\fi
		
		For a set $A$ of unlabeled formulas, 
		we write $\sigma : A$ to mean the set of labeled formulas $\{\sigma : \phi\ |\ \phi\in A\}$.

		\begin{proposition}[Lifting Process]
			\label{prop:lifting process 2}
			Given a sound rule of the form
			$$
			\begin{aligned}
				\infer[]
				{\Gamma\Rightarrow \Delta}
				{\Gamma_1\Rightarrow \Delta_1
					&...&
					\Gamma_n\Rightarrow \Delta_n}
			\end{aligned}, \mbox{ $n\ge 1$}, 
			$$
			in which all formulas are unlabelled, rule 
			$$
			\begin{aligned}
				\infer[]
				{\sigma : \Gamma\Rightarrow \sigma : \Delta}
				{\sigma : \Gamma_1\Rightarrow \sigma : \Delta_1
					&...&
					\sigma : \Gamma_n\Rightarrow \sigma : \Delta_n}
			\end{aligned}
			$$
			is sound for any configuration $\sigma\in \Conf_\free(\Gamma\cup \Delta\cup \Gamma_1\cup \Delta_1\cup...\cup\Gamma_n\cup\Delta_n)$. 
		\end{proposition}
		
		Proposition~\ref{prop:lifting process 2} is proved in Appendix~\ref{section:Other Propositions and Proofs} based on the notion of free configurations defined above. 
		From the proof of Proposition~\ref{prop:lifting process 2}, we can see that when the rule is an axiom, we actually only need Definition~\ref{def:free configurations}-\ref{item:free configurations cond 1}, as stated in the following proposition, where we call a configuration $\sigma$ \emph{standard} w.r.t. $A$ if it satisfies  Definition~\ref{def:free configurations}-\ref{item:free configurations cond 1}, and write $\Conf_\std(A)$ as the set of all standard configurations w.r.t. $A$. 
		
		\begin{proposition}[Simple Version of Proposition~\ref{prop:lifting process 2}]
			\label{prop:lifting process simple}
			Given a sound axiom
			$
			\begin{aligned}
				\infer[]
				{\Gamma\Rightarrow \Delta}
				{}
			\end{aligned}, 
			$
			rule 
			$
			\begin{aligned}
				\infer[]
				{\sigma : \Gamma\Rightarrow \sigma : \Delta}
				{}
			\end{aligned}
			$
			is sound for any standard configuration $\sigma\in \Conf_\std(\Gamma\cup \Delta)$. 
		\end{proposition}

		\section{Case Studies}
		\label{section:Case Studies}
		
		In this section, we give some examples to show how $\LDL$ can be useful in reasoning about programs. 
		In the first example (Section~\ref{section:Example One: A While Program}), 
		we show how a cyclic deduction of a traditional \emph{While} program 
		is carried out. 
		Section~\ref{section:Lifting Process in While Programs} gives an example showing by applying the results obtained in Section~\ref{section:Lifting Process From Program Domains} that how lifting rules from the theory of FODL can be helpful in deriving programs in $\LDL$. 
		In Section~\ref{section:Encoding of Complex Configurations}, we show the powerfulness of $\LDL$ by encoding a type of more complex configurations in the programs of separation logic~\cite{Reynolds02} than in \emph{While} programs. This example shows the potential of $\LDL$ to be applied in real programming languages with the ability to manipulate storage. 
		
		\ifx
		In this section, we show how the proposed $\LDL$ can be used to derive programs step by step based on both symbolic-execution-based reasoning and structure-based reasoning. 
		The first example is a traditional \emph{while} program, with which we display a complete cyclic deduction procedure of a property described by modality $[\cdot]$ in $\LDL$. 
		While the second example is an Esterel program with a more complex non-compositional behaviour. 
		We will explain the reason why it is non-compositional. 
		Due to the limit of space, we will leave its detailed deduction procedure to Appendix~\ref{section:Other Details of Example Two}. 
		
		In the first example (Section~\ref{section:Example One: A While Program}), we show a complete proof procedure of a traditional program --- a \emph{while} program that has already been shown in Table~\ref{table:An Example of Program Structures}. 
		This example gives a clear picture of how cyclic deductions of $\pfDLp$ can be carried out. 
		\fi
		
		\ifx
		for modality $[\cdot]$ and the other for modality $\la \cdot\ra$. 
		
		Due to the limit of space, we only give an example for modality $[\cdot]$. 
		For modality $\la \cdot\ra$ (which involves termination factors) and a more interesting example of Esterel programs, refer to the technique report~\cite{zhang2024parameterizeddynamiclogic}.  
		\fi
		
		\subsection{Cyclic Deduction of A While Program}
		\label{section:Example One: A While Program}
		We prove the property given in Example~\ref{example:DLp specifications}, stated as the following sequent
		$$
		\nu_1\dddef \cdot \Rightarrow v \ge 0 \to \sigma_1 : [\textit{WP}] (s = ((v+1)v)/2).
		$$
		Recalling that $\sigma_1 \dddef \{n\mapsto v, s\mapsto 0\}$ with $v$ a free variable; $\textit{WP}$ is defined as
		$$\begin{aligned}
			\textit{WP}
			\dddef&
			\{
			\textit{while}\
			(n > 0)\
			\textit{do}\
			s := s + n\ ;\
			n := n - 1\
			\textit{end}\
			\}.
		\end{aligned}$$ 
		\ifx
		We prove a property of program $\textit{WP}\in \Prog_\WP$ (Table~\ref{table:An Example of Program Structures}):
		$$
		\begin{aligned}
			\textit{WP}
			\dddef&
			\{
			s := 0\ ;\
			\textit{while}\
			(n > 0)\
			\textit{do}\
			s := s + n\ ;\
			n := n - 1\
			\textit{end}\
			\},
		\end{aligned}
		$$
		described in a sequent of $\LDL$ formulas as follows:
		$$
		\nu_1\dddef \sigma_1 : (n \ge 0) \Rightarrow \sigma_1 : [\textit{WP}] (s = \frac{(n+1)n}{2}), 
		$$
		where $\sigma_1 \dddef \{n\mapsto v, s\mapsto u\}$ with $v, u$ free variables. 
		$\nu_1$ says that given an initial value $n\ge 0$, after executing $\textit{WP}$, $s$ equals to $((n+1)n)/2$, which is the sum of $1 + 2 + ... + n$.
		\fi

		\ifx
		\begin{table}[tb]
			\begin{center}
				\noindent\makebox[\textwidth]{%
					\scalebox{1.0}{
						\begin{tabular}{|c|}
							\toprule
							\\
							$
							\infer[^{(x:=e)}]
							{\Gamma\Rightarrow (x:=e, \sigma)\trans(\ter, \sigma[x\mapsto e_\sigma]), \Delta}
							{
							}
							$
							\ \
							$
							\infer[^{(;\ter)}]
							{\Gamma\Rightarrow (\alpha; \beta, \sigma)\trans(\beta, \sigma'), \Delta}
							{
								\Gamma\Rightarrow (\alpha, \sigma)\trans (\ter, \sigma'), \Delta
							}
							$
							\\
							\\
							$
							\infer[^{(\textit{wh1})}]
							{\Gamma\Rightarrow (\textit{while}\ \phi\ \textit{do}\ \alpha\ \textit{end}, \sigma)\trans(\alpha';\ \textit{while}\ \phi\ \textit{do}\ \alpha\ \textit{end}, \sigma'), \Delta}
							{
								\Gamma\Rightarrow \phi_\sigma, \Delta
								&
								\Gamma\Rightarrow (\alpha, \sigma)\trans(\alpha', \sigma'), \Delta
							}
							$
							\\
							\\
							$
							\infer[^{(\textit{wh1}\ter)}]
							{\Gamma\Rightarrow (\textit{while}\ \phi\ \textit{do}\ \alpha\ \textit{end}, \sigma)\trans(\textit{while}\ \phi\ \textit{do}\ \alpha\ \textit{end}, \sigma'), \Delta}
							{
								\Gamma\Rightarrow \phi_\sigma, \Delta
								&
								\Gamma\Rightarrow (\alpha, \sigma)\trans(\ter, \sigma'), \Delta
							}
							$
							\\
							\\
							$
							\infer[^{(\textit{wh2})}]
							{\Gamma\Rightarrow (\textit{while}\ \phi\ \textit{do}\ \alpha\ \textit{end}, \sigma)\trans(\ter, \sigma), \Delta}
							{
								\Gamma\Rightarrow (\neg\phi)_\sigma, \Delta
							}
							$
							\\
							\ifx
							\midrule
							\multicolumn{1}{|l|}{
								\begin{tabular}{l}
									$^{1}$ if (1) for each $\LDL$ formula $\psi : \sigma$ in $\Gamma$ or $\Delta$, $\psi\in \Fmla$; (2) $\Gamma\Rightarrow \Delta$ is valid.\\
									$^{2}$ where $\alpha$ is neither $\ter$ nor $\abort$. \\
									$^{3}$ where $X \in \ConfV$ does not appear in $\phi$, $\Gamma$ and $\Delta$.
								\end{tabular}
							}
							\\
							\fi
							\bottomrule
						\end{tabular}
					}
				}
			\end{center}
			\caption{Some rewrite rules for a traditional \emph{while} program}
			\label{table:Some rewrite rules for a traditional while program}
		\end{table}
		\fi
		
		\ifx
		A configuration $\sigma$ of a while program is a storage that maps a variable to a value of the integer domain $\mbb{Z}$.
		For example, $\{n \mapsto 1, s\mapsto 0\}$ denotes a configuration that maps $n$ to $1$ and $s$ to $0$.
		Part of the operation semantics of Program $\textit{WP}$ are described as inference rules in Table~\ref{table:Examples of programs, configurations and operational semantics}. 
		The interpretation $\app(\sigma)$ of a configuration $\sigma$ is in the usual sense. 
		For example, $\app(\{n\mapsto 1, s\mapsto 0\})$ is the mapping that maps $n, s$ to $1, 0$ respectively, and maps other variables to themselves.
		We have that, for instance, $\app(\{n\mapsto 1, s\mapsto 0\})(n\ge 0) = 1\ge 0$ by  assigning $n$ to $1$ in $n\ge 0$. 
		\fi

		\ifx
		It is easy to see that the operational semantics of a program following these rewrite rules satisfies the two program properties \Romann{1} and \Romann{2} in Section~\ref{section:Construction of LDL}.
		Specifically, 1. any assignment, sequence program or while program can proceed following these rewrite rules; and 2. any program that proceeds by these rewrite rules has finite branches.
		\fi

		\begin{table}[tb]
			\noindent\makebox[\textwidth]{%
				\scalebox{0.9}{
					\begin{tabular}{l|l}
						\toprule
						\begin{tabular}{c}
							\begin{tikzpicture}[->,>=stealth', node distance=3cm]
								\node[draw=none] (txt2) {
									$
										\infer[^{(\to)}]
										{\mbox{$\nu_1$: 1}}
										{
											\infer[^{(\Sub)}]
											{2}
											{
												\infer[^{(\textit{Cut})}]
												{3}
												{
													\infer[^{(\textit{Wk} R)}]
													{4}
													{
														\infer[^{(\textit{Ter})}]
														{19}
														{}
													}
													&    
													\infer[^{(\vee L)}]
													{5}
													{\infer[^{([\alpha])}]
														{6}
														{
															\infer[^{([\alpha])}]
															{11}
															{
																\infer[^{(\textit{Cut})}]
																{12}
																{
																	\infer[^{(\textit{Wk} L)}]
																	{13}
																	{
																		\infer[^{(\Sub)}]
																		{16}
																		{
																			\infer[^{(\textit{Wk} L)}]
																			{17}
																			{18}
																		}
																	}
																	&
																	\infer[^{(\textit{Wk} R)}]
																	{14}
																	{15}
																}
															}
														}
														&
														\infer[^{([\alpha])}]
														{7}
														{
															\infer[^{([\ter])}]
															{8}
															{\infer[^{(\textit{Int})}]
																{9}
																{
																	\infer[^{(\textit{Ter})}]
																	{10}
																	{
																	}
																}
															}
														}
													}
												}
											}
										}
									$
								};
								

								\path
								;
								
								
								\draw[dotted,thick,red] ([xshift=-1.3cm, yshift=1.75cm]txt2.center) -- 
								([xshift=-1.6cm, yshift=1.75cm]txt2.center) --
								([xshift=-1.6cm, yshift=-1cm]txt2.center) --
								([xshift=-1.1cm, yshift=-1cm]txt2.center);

							\end{tikzpicture}
						\end{tabular}
						&
						\begin{tabular}{l}
							Definitions of other symbols:
							\\
							$\alpha_1\dddef s := s + n;\ n := n - 1$
							\\
							$\phi_1 \dddef (s = ((v + 1)v)/2)$
							\\
							$\sigma_1(v) \dddef \{n\mapsto v, s\mapsto 0\}$
							\\
							$\sigma_2(v) \dddef \{n\mapsto v, s\mapsto 0\}$
							\\
							$\sigma_3(v)\dddef \{n\mapsto v - m,
							s\mapsto (2v - m + 1)m/2\}$
							\\
							$\sigma_4(v) \dddef \{n\mapsto v - m,
							s\mapsto (2v - (m + 1) + 1)(m + 1)/2\}$
							\\
							$\sigma_5(v)\dddef \{n\mapsto v - (m + 1),
							s\mapsto (2v - (m + 1) + 1)(m + 1)/2\}$
							\\
						\end{tabular}
						\\
						\midrule
						\multicolumn{2}{c}{
							\begin{tabular}{l l l l}
								1: & $\cdot$ & $\Rightarrow$ & $v\ge 0\to \sigma_1(v) : [\textit{while}\
								(n > 0)\
								\textit{do}\
								\alpha_1\
								\textit{end}\ ]\phi_1$
								\\
								2: & $v\ge 0$ & $\Rightarrow$ & $\sigma_1(v) : [\textit{while}\
								(n > 0)\
								\textit{do}\
								\alpha_1\
								\textit{end}\ ]\phi_1$
								\\
								3: & $v - m\ge 0$ & $\Rightarrow$ & $\mbox{\ul{$\sigma_3(v) : [\textit{while}\
										(n > 0)\
										\textit{do}\
										\alpha_1\
										\textit{end}\ ]\phi_1$}}$
								\\
								4: & $v - m\ge 0$ & $\Rightarrow$ & 
								$\left\{\begin{array}{l}
									\mbox{$\sigma_3(v) : [\textit{while}\
										(n > 0)\
										\textit{do}\
										\alpha_1\
										\textit{end}\ ]\phi_1$},\\
									v - m > 0\vee v - m\le 0
								\end{array}\right\}$
								\\
								19: & $v - m\ge 0$ & $\Rightarrow$ & $v - m> 0 \vee v - m\le 0$
								\\
								\midrule
								5: & $v - m\ge 0, v - m > 0\vee v - m\le 0$ & $\Rightarrow$ & $\mbox{\ul{$\sigma_3(v) : [\textit{while}\
										(n > 0)\
										\textit{do}\
										\alpha_1\
										\textit{end}\ ]\phi_1$}}$
								\\
								6: & $v - m\ge 0, v - m > 0$ & $\Rightarrow$ & $\mbox{\ul{$\sigma_3(v) : [\textit{while}\
										(n > 0)\
										\textit{do}\
										\alpha_1\
										\textit{end}\ ]\phi_1$}}$
								\\
								11: & $v - m\ge 0, v - m > 0$ & $\Rightarrow$ & $\mbox{\ul{$\sigma_4(v) : [n := n - 1;\ \textit{while}\
										(n > 0)\
										\textit{do}\
										\alpha_1\
										\textit{end}\ ]\phi_1$}}$
								\\
								12: & $v - m\ge 0, v - m > 0$ & $\Rightarrow$ & $\mbox{\ul{$\sigma_5(v) : [\textit{while}\
										(n > 0)\
										\textit{do}\
										\alpha_1\
										\textit{end}\ ]\phi_1$}}$
								\\
								13: & $\left\{\begin{array}{l}
									v - m\ge 0, v - m > 0,\\
									v - (m + 1)\ge -1, v - (m + 1) \ge 0
								\end{array}\right\}$ & $\Rightarrow$ & $\mbox{\ul{$\sigma_5(v) : [\textit{while}\
										(n > 0)\
										\textit{do}\
										\alpha_1\
										\textit{end}\ ]\phi_1$}}$
								\\
								16: & $v - (m + 1)\ge -1, v - (m + 1) \ge 0$ & $\Rightarrow$ & $\mbox{\ul{$\sigma_5(v) : [\textit{while}\
										(n > 0)\
										\textit{do}\
										\alpha_1\
										\textit{end}\ ]\phi_1$}}$
								\\
								17: & $v - m\ge -1, v - m \ge 0$ & $\Rightarrow$ & $\mbox{\ul{$\sigma_3(v) : [\textit{while}\
										(n > 0)\
										\textit{do}\
										\alpha_1\
										\textit{end}\ ]\phi_1$}}$
								\\
								18: & $v - m\ge 0$ & $\Rightarrow$ & $\mbox{\ul{$\sigma_3(v) : [\textit{while}\
										(n > 0)\
										\textit{do}\
										\alpha_1\
										\textit{end}\ ]\phi_1$}}$
								\\
								\midrule
								14: & $v - m\ge 0, v - m > 0$ & $\Rightarrow$ & 
								$\left\{\begin{array}{l}
									\mbox{$\sigma_5(v) : [\textit{while}\
										(n > 0)\
										\textit{do}\
										\alpha_1\
										\textit{end}\ ]\phi_1$},\\
									v - (m + 1)\ge -1, v - (m + 1) \ge 0
								\end{array}\right\}$
								\\
								15: & $v - m\ge 0, v - m > 0$ & $\Rightarrow$ & 
								$
								v - (m + 1)\ge -1, v - (m + 1) \ge 0
								$
								\\
								\midrule
								7: & $v - m\ge 0, v - m \le 0$ & $\Rightarrow$ & $\sigma_3(v) : [\textit{while}\
								(n > 0)\
								\textit{do}\
								\alpha_1\
								\textit{end}\ ]\phi_1$
								\\
								8: & $v - m\ge 0, v - m \le 0$ & $\Rightarrow$ &  $\sigma_3(v) : [\ter]\phi_1$
								\\
								9: & $v - m\ge 0, v - m \le 0$ & $\Rightarrow$ & $\sigma_3(v) : (s=((v+1)v)/2)$
								\\
								10: & $v - m\ge 0, v - m \le 0$ & $\Rightarrow$ &  $((2v -m + 1)m) /2=((v+1)v)/2)$
								\\
							\end{tabular}
						}
				\\
				\bottomrule
			\end{tabular}
		}
	}
	\caption{Derivations of Property $\nu_1$}
	\label{figure:The derivation of Example 1}
\end{table}

Table~\ref{figure:The derivation of Example 1} shows the derivation of this formula.
We omit all sub-proof-procedures in the proof system $P_{\Oper_\WP,\Terminate_\WP}$ derived by the inference rules in 
Table~\ref{table:An Example of Inference Rules for Program Behaviours} when deriving via rule $([\alpha])$. 
We write term $t(x)$ if $x$ is a variable that has free occurrences in $t$. 
Non-primitive rules $(\to R)$ and $(\vee L)$ are can be derived by the rules for $\neg$ and $\wedge$ accordingly. 
For example, rule $(\vee L)$ can be derived as follows:
$$
\begin{aligned}
	\infer[^{(\neg L)}]
	{\Gamma, \phi\vee \psi\Rightarrow \Delta}
	{
		\infer[^{(\wedge R)}]
		{\Gamma\Rightarrow (\neg\phi) \wedge (\neg\psi), \Delta}
		{
			\infer[^{(\neg R)}]
			{\Gamma\Rightarrow \neg\phi, \Delta}
			{\Gamma, \phi\Rightarrow \Delta}
			&
			\infer[^{(\neg R)}]
			{\Gamma\Rightarrow \neg\psi, \Delta}
			{\Gamma, \psi\Rightarrow \Delta}
		}
	}
\end{aligned}
.
$$

The derivation from sequent 2 to 3 is according to rule
$$
\begin{aligned}
	\infer[^{(\Sub)}]
	{\Gamma[e/x]\Rightarrow \Delta[e/x]}
	{\Gamma\Rightarrow\Delta}, 
\end{aligned}
$$
where function $(\cdot)[e/x]$ is an instantiation of abstract substitution defined  in Definition~\ref{def:Abstract Substitution}. 
Recalling that for any formula $\phi$, $\phi[e/x]$ substitutes each free variable $x$ of $\phi$ with term $e$ (Section~\ref{section:Examples of Term Structures}). 
Its formal definition and the proof of instantiation are given in Appendix~\ref{section:Formal Definitions of While Programs}. 
We observe that sequent 2 can be writen as: 
$$(v - m\ge 0)[0/m]\Rightarrow (\sigma_3(v) : [\textit{while}\
(n > 0)\
\textit{do}\
\alpha_1\
\textit{end}\ ]\phi_1)[0/m], $$
as a special case of sequent 3. 
Intuitively, sequent 3 captures the situation after the $m$th loop ($m\ge 0$) of program $\WP$. 
This step is crucial as starting from sequent 3, we can find a bud node --- 18 --- that is identical to node 3. 
\ifx
to help constructing a bud for each potential infinite proof branch. 
Here, the substitution is instantiated as $t[e/x]$, which returns a term by replacing each free variable $x$ of $t$ with term $e$. 
Its formal definition is given in Appendix~\ref{section:Formal Definitions of While Programs}. 

The derivation from sequent 3 to 4 is according to rule
$$
\begin{aligned}
	\infer[^{(\Extra:\sigma\cfeq)}]
	{\Gamma\Rightarrow \sigma : [\alpha]\phi, \Delta}
	{\Gamma\Rightarrow \sigma' : [\alpha]\phi, \Delta}
\end{aligned}, \mbox{if $\sigma\cfeq \sigma'$}, 
$$
where 
equivalence $\sigma\cfeq\sigma'$ means that $\sigma$ and $\sigma'$ have the same affect on formulas, whose definition is given in Appendix~\ref{section:Formal Definitions of While Programs}. 
In this example, intuitively, we see that $\sigma_3$ actually stores the same values for $n$ and $s$ as $\sigma_2$, because by replacing $n$ with its stored value $v$ in the expression $((n+v+1)(n-v))/2$, we have  $((v+v+1)(v-v))/2 = 0$. 
So $\sigma_3$ maps $s$ to $0$ as $\sigma_2$, which means that for any formula $\phi$, we have $\app_\WP(\sigma_2, \phi)\equiv \app_\WP(\sigma_3, \phi)$. 
This observation, however, relies on the explicit definition of configurations $\Conf_\WP$ and interpretation $\app_\WP$. 
Configuration $\sigma_3(v)$ constructed in node 4 is crucial, as starting from it, we can find a bud node --- 17 ---
that is identical to node 5. 
\fi
The derivation from sequent 3 to \{4, 5\} provides a lemma: 
$v - m > 0\vee v - m\le 0$, which is obvious valid. 
Sequent 16 indicates the end of the $(m+1)$th loop of program $\WP$. 
From node 12 to 16, we rewrite the formulas on the left side into an obvious logical equivalent form in order to apply rule $(\Sub)$. 
From sequent 16 to 17 rule $(\Sub)$ is applied, with 16 can be written as:
$$
(v - m \ge -1)[m + 1/m], (v - m\ge 0)[m + 1/m]\Rightarrow (\sigma_3 : [\textit{while}\
(n > 0)\
\textit{do}\
\alpha_1\
\textit{end}\ ]\phi_1)[m + 1/m].
$$



The whole proof tree is a cyclic preproof because the only derivation path: $1,2,3,5, 6, 11, 12, 13, 16, 17, 18, 3, ...$
has a progressive derivation trace whose elements are underlined in Table~\ref{figure:The derivation of Example 1}.


The configurations in $\LDL$ subsume the concept of `updates' in dynamic logics, which were adopted in work like~\cite{Beckert13}. 
In this example, if we allow a configuration to be a meta variable, say $X$, 
then it is just an update, which has no explicit structures but carries a series of ``delayed substitutions''. 
For instance, we can have a configuration $X[0/s][1/x]$ that carries two substitutions in sequence.
When applying to formulas, these substitutions are passed to formulas:  
$\app(X[0/s][1/x], s + x + 1 > 0) \dddef ((s + x + 1 > 0)[0/s])[1/x] \equiv (0 + x + 1 > 0)[1/x] \equiv 0 + 1 + 1 > 0$ through a well-defined $\app$.  
From this we see that one of the advantages of $\LDL$ is that it allows explicit structures of program status such as stores or stacks we have seen in Section~\ref{section:Examples of Term Structures}. 

Appendix~\ref{section:Example Two: A Synchronous Loop Program} introduces another example of derivations of a program with implicit loop structures, 
which can better highlight the advantages brought by the cyclic proof framework of $\LDL$. 

\ifx
[This paragraph needs to be changed later???]
The semantics of a \emph{while} program supports a compositional reasoning based on its syntactic structures. 
For example, we can have a sound deduction 
$
\begin{aligned}
	\infer[^{([;])}]
	{n\ge 0\Rightarrow [s:=0 ; \alpha_1] \phi_1}
	{n\ge 0\Rightarrow [s:=0][\alpha_1]\phi_1}  
\end{aligned}
$
from node $1$ in a special dynamic logic for \emph{while} programs (e.g.~\cite{Beckert2016}). 
This derivation is compositional in the sense that it dissolves the structure $s := 0; \alpha_1$ into its sub-structures $s:=0$ and $\alpha_1$ in two different modalities $[\cdot ]$. 
In $\LDL$, we can lift it as a labelled derivation based on configuration $\sigma_1$:
$
\begin{aligned}
	\infer[^{(\sigma\textit{Lif})}]
	{\sigma_1: n\ge 0\Rightarrow \sigma_1: [s:=0 ; \alpha_1] \phi_1}
	{\sigma_1: n\ge 0\Rightarrow \sigma_1: [s:=0][\alpha_1]\phi_1}
\end{aligned}
$, 
by applying rule $(\sigma\textit{Lif})$ and the fact that $\sigma_1$ is a free configuration. 
Echoing Section~\ref{section:Summery}, this example demonstrates that the framework of $\LDL$ is compatible with the existing `unlabeled' rules through a lifting process. 
\fi

\ifx
Readers might notice that different configurations can be chosen. In this example, we can also prove $P_1 : X$ with $X$
a free configuration variable.
In this case, the configuration is just treated as a state from the set of all variables to values.
A cyclic preproof can also be constructed by choosing a suitable configuration with parameters.
This is a good example showing that a configuration in $\LDL$ can be of an arbitrary structure, not just a state.
\fi

\subsection{Lifting Structural Rules From FODL}
\label{section:Lifting Process in While Programs}

We give two examples from the theory of FODL (cf.~\cite{Harel00}) to illustrate how existing rules in FODL 
can be helpful for deriving compositional programs in $\LDL$. 


\ifx
\textbf{Example 1:} For some derived result in 

In the following, we assume a special domain $\TA_X$, in which for each unlabeled formula $\phi$ discussed, $\rho\models \phi$ is well defined. 

In Section~\ref{section:Example One: A While Program}, we have seen how a \emph{While} program specification can be derived in $\pfDLp$ based on proof system $(P_{\Oper, \Terminate})_\WP$ for \emph{While} program behaviours. 
In this subsection, we show that lifting processes are helpful for  converting existed dynamic logic rules into labeled forms in $\LDL$. 
\fi

In FODL, consider an axiom
$$
\infer[^{([\seq])}]
{[\alpha][\beta]\phi\Rightarrow [\alpha\seq \beta]\phi}
{},
$$
which comes from the valid formula $[\alpha][\beta]\phi\rightarrow [\alpha\seq \beta]\phi$, as a useful rule appearing in many dynamic logic calculi that are based on FODL (e.g.~\cite{Beckert2016}). 
By Proposition~\ref{prop:lifting process simple}, we lift $([;])$ 
as 
$$
\infer[]
{\sigma : [\alpha][\beta]\phi\Rightarrow \sigma : [\alpha\seq \beta]\phi}
{}
$$
for any standard configuration $\sigma\in \Conf_\std(\{[\alpha][\beta]\phi, [\alpha\seq\beta]\phi\})$. 
And by Lemma~\ref{lemma:from logical consequence to sequent derivation} (Appendix~\ref{section:Other Propositions and Proofs}), we obtain a labeled rule
$$
\infer[^{(\sigma[\seq])}]
{\Gamma\Rightarrow \sigma : [\alpha\seq\beta]\phi, \Delta}
{\Gamma\Rightarrow \sigma : [\alpha][\beta]\phi, \Delta}, 
$$
which provides a compositional reasoning for sequential programs in $\pfDLp$ as an additional rule. 
It is useful when verifying a property like $\Gamma, \sigma' : [\beta]\phi\Rightarrow \sigma : [\alpha\seq \beta]\phi, \Delta$, in which we might finish the proof by only symbolic executing program $\alpha$ as: 
$$
\infer[^{(\sigma[\seq])}]
{\Gamma, \sigma' : [\beta]\phi\Rightarrow \sigma : [\alpha\seq\beta]\phi, \Delta}
{
	\infer[^{([\alpha])}]
	{
		\Gamma, \sigma' : [\beta]\phi\Rightarrow \sigma : [\alpha][\beta]\phi, \Delta
	}
	{
		\infer*[]
		{...}
		{
			\infer[^{([\alpha])}]
			{...}
			{
				\infer[^{([\ter])}]
				{\Gamma, \sigma' : [\beta]\phi\Rightarrow \sigma' : [\ter][\beta]\phi, \Delta}
				{
					\infer[^{(\textit{ax})}]
					{\Gamma, \sigma' : [\beta]\phi\Rightarrow \sigma' : [\beta]\phi, \Delta}
					{}
				}
			}
		}
	}
}, 
$$
especially when verifying program $\beta$ can be very costly. 

Another example is rule
$$
\infer[^{([\textit{Gen}])}]
{[\alpha]\phi\Rightarrow [\alpha]\psi}
{\phi\Rightarrow \psi}
$$
for generating modality $[\cdot]$, 
which were used for deriving the structural rule of loop programs in FODL (cf.~\cite{Harel00}). 
By Proposition~\ref{prop:lifting process 2}, we lift $([\textit{Gen}])$ as the following rule
$$
\infer[^{(\sigma [\textit{Gen}])}]
{\sigma : [\alpha]\phi\Rightarrow \sigma : [\alpha]\psi}
{\sigma : \phi\Rightarrow \sigma : \psi}, 
$$
where $\sigma\in \Conf_\free(\{[\alpha]\phi, [\alpha]\psi, \phi, \psi\})$. 
It is useful, for example, when deriving a property
$\sigma : [\alpha]\la\beta\ra \phi\Rightarrow \sigma : [\alpha]\la\beta'\ra\psi$, where 
we can skip the verification of program $\alpha$ as follows: 
$$
\infer[^{(\sigma [\textit{Gen}])}]
{\sigma : [\alpha]\la\beta\ra \phi\Rightarrow \sigma : [\alpha]\la\beta'\ra\psi}
{\sigma : \la\beta\ra\phi\Rightarrow \sigma : \la\beta'\ra\psi} 
$$
and directly focus on verifying program $\beta'$.

\ifx
\subsection{Embedding of First-Order Dynamic Logic}
\label{section:Embedding First-Order Dynamic Logic}
First-Order Dynamic Logic (FODL) is the very base theory for many dynamic logic variations such as~\cite{???}. 
We embed the theory of FODL into $\LDL$ by proposing a set of rules for program behaviours of regular expressions with tests based on its semantics. 
Moreover, as mentioned in Section~\ref{section:Summery}, we show that the generality of $\LDL$ allows the adaption of the proof rules of FODL into $\LDL$ through a lifting process. 
This means, by choosing a suitable configuration, one can embed the whole theory of a logic into $\LDL$, for a better support of both symbolic-execution-based reasoning based on behaviours and compositional reasoning based on syntax. 

\begin{table}[tb]
	\begin{center}
		\noindent\makebox[\textwidth]{%
			\scalebox{0.9}{
				\begin{tabular}{c}
					\toprule
					\multicolumn{1}{l}{Syntax of FODL:} 
					\\
					$
					\begin{aligned}
						\alpha\dddef&\ x := e\ |\ \psi?\ |\ \alpha\seq\alpha\ |\ \alpha\cho\alpha\ |\ \alpha^*,\\
						\psi\dddef&\ p\ |\ \neg\psi\ |\ \psi\wedge \psi\ |\ \forall x.\psi\\
						\phi\dddef&\ \psi\ |\ \neg\phi\ |\ \phi\wedge \phi\ |\ [\alpha]\phi.
					\end{aligned}
					$
					\\
					\midrule
					\multicolumn{1}{l}{Proof System of FODL:}
					\\
					\begin{tabular}{l | l}
						\\
						Axioms: & Inference Rules:
						\\
						\begin{tabular}{l}
							(1) Axioms for propositional logic\\
							(2) $[x := e]\phi\leftrightarrow \phi[e/x]$\\
							(3) $[\phi?]\leftrightarrow (\phi\to \psi)$\\
							(4) $[\alpha](\phi\to \psi) \to ([\alpha]\phi\to [\alpha]\psi)$\\
							(5) $[\alpha](\phi\wedge \psi) \leftrightarrow [\alpha]\phi\wedge [\alpha]\psi$\\
							(6) $[\alpha\cho \beta]\phi \leftrightarrow [\alpha]\phi \wedge [\beta]\phi$\\
							(7) $[\alpha\seq \beta]\phi\leftrightarrow [\alpha][\beta]\phi$\\
							(8) $\phi\wedge [\alpha][\alpha^*]\phi \leftrightarrow [\alpha^*]\phi$\\
							(9) $\phi\wedge [\alpha^*](\phi\to [\alpha]\phi)\to [\alpha^*]\phi$
						\end{tabular}
						&
						\begin{tabular}{l}
							($\textit{MP}$)\ \ $\begin{aligned}
								\infer[]
								{\psi}
								{\phi & \phi\to \psi}
							\end{aligned}$\\
							($\forall$)\ \ $\begin{aligned}
								\infer[]
								{\forall x.\phi}
								{\phi}
							\end{aligned}$\\
							($[\textit{Gen}]$)\ \ $\begin{aligned}
								\infer[]
								{[\alpha]\phi}
								{\phi}
							\end{aligned}$\\
							($\la\textit{Cov}\ra$)\ \ $\begin{aligned}
								\infer[]
								{\phi(n)\to \la\alpha^*\ra\phi(0)}
								{\phi(n+1)\to \la\alpha\ra\phi(n)}
							\end{aligned}$
						\end{tabular}
						\\
					\end{tabular}
					\\
					\bottomrule
				\end{tabular}
			}
		}
	\end{center}
	\caption{Theory of First-Order Dynamic Logic}
	\label{table:Theory of First-Order Dynamic Logic}
\end{table}

Table~\ref{table:Theory of First-Order Dynamic Logic} gives the syntax and the proof system of FODL. 
The program models $\alpha$ are regular expressions with tests consisting of primitive terms assignments $x := e$, tests $\psi?$, and compositional terms of primitive terms using operators sequence $\seq$, choice $\cho$ and star $*$. 
Intuitively, a test $\psi?$ proceeds if $\psi$ is true and aborts otherwise; sequential program $\alpha\seq \beta$ proceeds $\alpha$ first, after $\alpha$ terminates, $\beta$ proceeds; choice program $\alpha\cup \beta$ either proceeds $\alpha$ or $\beta$ non-deterministically; star program $\alpha^*$ performs $\alpha$ for a random finite number of times. 
The proof rules of FODL are in a Hilbert-style, where each node is a formula, not a sequent. 
Formal definitions and more details of FODL can be found in references like~\cite{???}.

\begin{table}[tbp]
	\begin{center}
		\noindent\makebox[\textwidth]{%
			\scalebox{0.9}{
				\begin{tabular}{c}
					\toprule
					$
					\infer[^{(x:=e)}]
					{\Gamma\Rightarrow (x := e, \sigma)\trans (\ter, \sigma[x\mapsto \sigma(e)]), \Delta}
					{}
					$
					\ \ 
					$
					\infer[^{(\seq 1)}]
					{\Gamma\Rightarrow (\alpha\seq\beta, \sigma)\trans (\beta, \sigma'), \Delta}
					{\Gamma\Rightarrow (\alpha, \sigma)\trans (\ter, \sigma'), \Delta}
					$
					\ \ 
					$
					\infer[^{(\seq 2)}]
					{\Gamma\Rightarrow (\alpha\seq\beta, \sigma)\trans (\alpha'\seq \beta, \sigma'), \Delta}
					{\Gamma\Rightarrow (\alpha, \sigma)\trans (\alpha', \sigma'), \Delta}
					$
					\\
					$
					\infer[^{(\cho 1)}]
					{\Gamma\Rightarrow (\alpha\cho\beta, \sigma)\trans (\alpha', \sigma'), \Delta}
					{\Gamma\Rightarrow (\alpha, \sigma)\trans (\alpha', \sigma'), \Delta}
					$
					\ \ 
					$
					\infer[^{(\cho 2)}]
					{\Gamma\Rightarrow (\alpha\cho\beta, \sigma)\trans (\beta', \sigma'), \Delta}
					{\Gamma\Rightarrow (\beta, \sigma)\trans (\beta', \sigma'), \Delta}
					$
					\ \ 
					$
					\infer[^{(* 1)}]
					{\Gamma\Rightarrow (\alpha^*, \sigma)\trans (\ter, \sigma'), \Delta}
					{\Gamma\Rightarrow (\alpha, \sigma)\trans (\ter, \sigma'), \Delta}
					$
					\ \ 
					$
					\infer[^{(* 2)}]
					{\Gamma\Rightarrow (\alpha^*, \sigma)\trans (\alpha'\seq \alpha^*, \sigma'), \Delta}
					{\Gamma\Rightarrow (\alpha, \sigma)\trans (\alpha', \sigma'), \Delta}
					$
					\\
					\bottomrule
				\end{tabular}
			}
		}
	\end{center}
	\caption{Rules for Program Behaviours of FODL}
	\label{table:Rules for Program Behaviours of FODL}
\end{table}

The program behaviours $\Oper_{\fodl}$ of regular expressions with tests of FODL are defined as inference rules shown in Table~\ref{table:Rules for Program Behaviours of FODL}. 
We adopt similar definitions as in \emph{While} programs (Appendix~\ref{section:Formal Definitions of While Programs}) for evaluations $\Eval_\fodl$, configurations $\Conf_\fodl$ and interpretations $\app_\fodl$ of FODL. 
Their formal definitions only differ from those of \emph{While} programs in the structures of programs, and so we omit them. 

\begin{proposition}
	In FODL, given any formulas $\phi, \psi$, if $\phi\to \psi$ is valid, 
	then rule 
	$$
	\begin{aligned}
		\infer[]
		{\Gamma\Rightarrow \sigma : \psi, \Delta}
		{\Gamma\Rightarrow \sigma : \phi, \Delta}
	\end{aligned}
	$$
	is sound for any configuration $\sigma\in \Conf_\fodl$ and contexts $\Gamma, \Delta$. 
\end{proposition}

\begin{proposition}
	In FODL, given a sound rule of the form
	$$
	\begin{aligned}
		\infer[]
		{\Gamma\Rightarrow \Delta}
		{\Gamma'\Rightarrow \Delta'}
	\end{aligned}, 
	$$
	then rule 
	$$
	\begin{aligned}
		\infer[]
		{\sigma : \Gamma\Rightarrow \sigma : \Delta}
		{\sigma : \Gamma'\Rightarrow \sigma : \Delta'}
	\end{aligned}, 
	$$
	is sound for any configuration $\sigma\in \Conf_\free(\Gamma\cup \Delta\cup \Gamma'\cup \Delta')$. 
\end{proposition}

\fi

\subsection{Encoding of Complex Configurations}
\label{section:Encoding of Complex Configurations}

The examples of configurations we have seen in both \emph{While} programs and Esterel programs (Section~\ref{section:Examples of Term Structures}) are both a simple version of variable storage. 
We show an example of defining a type of more complex configurations to capture the notion of `heaps' for separation logic~\cite{Reynolds02} --- an extension of Hoare logic for reasoning about program data structures. 
We will encode a partial separation logic into $\LDL$. 
From this example, we see that a configuration interpretation $\app$ can be more than just substitutions of terms like $\app_\WP$ and $\app_\E$ seen in previous examples. 

\ifx
we encode a part of separation logic~\cite{Reynolds02} --- an extension of Hoare logic for reasoning about program data structures --- in which a type of more complex configurations are required to capture the notion of `heaps'. 
From this encoding, we shall see that a configuration interpretation can be defined more than just as substitutions of terms like $\app_\WP$ and $\app_\E$ in the examples in Section~\ref{section:Examples of Term Structures}. 
\fi

\textbf{Separation Logic}. 
In the following, we assume readers are familiar with separation logic and only give an informal explanations of its semantics.
For simplicity, we only introduce a partial separation logic primitives: 
for formulas, we only introduce an atomic formula $e\allocto e'$ and a critical operator $\sepc$ (while omitting another critical operator $\sepi$); 
for programs, we only introduce atomic program statements: $x := \Alloc(e)$, $x := [e]$, $[e] := e'$ and $\disAlloc(e)$ that have direct influences on heaps, and ignore compositional programs that varies from case to case. 
After the introductions of the semantics, we give the encoding of separation logic into $\LDL$ as labeled formulas and explain how they actually capture the same meaning. 

Below we follow some conventions of notations: 
Given a partial function $f : A \partto B$, we use $\dom(f)$ to denote the domain of $f$. 
For a set $C$, partial function $f|_C : A\partto B$ is the function $f$ restricted on domain $\dom(f)\cap C$.  
$f[x\mapsto e]$ represents the partial function that maps $x$ to $e$, and maps the other variables in its domain to the same value as $f$. 

Let $\mbb{V} = \mbb{Z}\cup \textit{Addr}$ be the set of values, where 
$\textit{Addr}$ is a set of addresses. 
We assume $\textit{Addr}$ to be expressed with an infinite set of integer numbers. 
In separation logic, a store $s : V\to \mbb{V}$ is a function that maps each variable to a value of $\mbb{V}$, 
a heap $h : \textit{Addr}\partto \mbb{V}$ is a partial function that maps each address to a value of $\mbb{V}$, expressing that the value is stored in a memory indicated by the address. 
$\dom(h)$ is a finite subset of $\textit{Addr}$. 
Let $e$ be an arithmetic expression with or without variables, 
$s(e)$ returns the value by replacing each variable $x$ of $e$ with value $s(x)$.  
A state is a store-heap pair $(s, h)$. 
The \emph{disjoint relation} $h_1\disj h_2$ is defined if $\dom(h_1)\cap\dom(h_2) = \emptyset$. 

Here we informally explain the behaviours of program statements. 
Given a state $(s,h)$, 
statement $x := \Alloc(e)$ allocates a memory addressed by a new integer $n$ in $h$ to store the value of expression $e$ (thus obtaining a new heap $h\cup \{(n, s(e))\}$ where $n\notin \dom(h)$), 
and assigns $n$ to $x$.
Statement $x := [e]$ assigns the value of the address $e$ in $h$ (i.e. $h(s(e))$) to variable $x$. 
$[e] := e'$ means to assign the value $e'$ to the memory of the address $e$ in $h$ (i.e. obtaining a new heap $h[s(e)\mapsto s(e')]$). 
$\disAlloc(e)$ means to de-allocate the memory of address $e$ in heap 
(i.e. obtaining a new heap $h|_{\dom(h)\setminus\{s(e)\}}$). 

Formula $e\allocto e'$ means that value $e'$ is stored in the memory of address $e$. 
Given a state $(s, h)$, $s, h\models e\allocto e'$ is defined if $h(s(e)) = s(e')$. 
For any separation logical formulas $\phi$ and $\psi$, $s,h\models \phi\ \sepc\ \psi$ if 
there exist heaps $h_1, h_2$ such that $h = h_1\cup h_2$, $h_1\disj h_2$, and $s,h_1\models \phi$ and $s, h_2\models \psi$.

\begin{example}
	Let $(s, h)$ be a state such that $s(x) = 3, s(y) = 4$ and $h = \emptyset$, then 
	the following table shows the information of each state about focused variables and addresses during the process of the following executions: 
	$$
	(s, h)\xrightarrow{x := \Alloc(1)}
	(s_1,h_1)\xrightarrow {y := \Alloc(1)}
	(s_2,h_2)\xrightarrow {[y] := 37}
	(s_3,h_3)\xrightarrow {y := [x + 1]}
	(s_4,h_4)\xrightarrow{\disAlloc(x+1)}
	(s_5,h_5).
	$$
	\begin{center}
		\begin{tabular}{| c | l | c | l |}
			\toprule
			& \multicolumn{1}{c|}{Store} & & \multicolumn{1}{c|}{Heap}\\
			\hline
			$s$ & $x: 3$, $y: 4$ & $h$ & empty\\
			\hline
			$s_1$ & $x: 37$, $y: 4$ & $h_1$ & $37 : 1$\\
			\hline
			$s_2$ & $x : 37$, $y : 38$ & $h_2$ & $37: 1$, $38 : 1$\\
			\hline
			$s_3$ & $x : 37$, $y : 38$ & $h_3$ & $37 : 1$, $38 : 37$\\
			\hline
			$s_4$ & $x : 37$, $y : 37$ & $h_4$ & $37: 1$, $38 : 37$\\
			\hline
			$s_5$ & $x : 37$, $y : 37$ & $h_5$ & $37 : 1$\\
			\bottomrule
		\end{tabular}   
	\end{center}
	Let $\phi\dddef (x\allocto 1\ \sepc\ y\allocto 1)$, $\psi\dddef (x\allocto 1\ \wedge\ y\allocto 1)$, 
	we have $s_2, h_2\models \phi$ and $s_2, h_2\models \psi$, $s_5, h_5\models \psi$, but $s_5, h_5\not\models \phi$ since $x$ and $y$ point to the single memory storing value $1$.  
\end{example}

\begin{table}[tbp]
	\begin{center}
		\noindent\makebox[\textwidth]{%
			\scalebox{0.9}{
				\begin{tabular}{c}
					\toprule
					\multicolumn{1}{l}{Configurations $\sigma_\Sep$: $(s, h)$}\\
					\midrule
					\multicolumn{1}{l}{Interpretations $\app_\Sep$:}\\
					\begin{tabular}{l}
						$\app_\Sep((s, h), e \allocto e')\dddef h(s(e)) = s(e')$\\
						$\app_\Sep((s,h), \phi\ \sepc\ \psi)\dddef \exists X, Y. (h = X\cup Y) \wedge (X\bot Y) \wedge \app_\Sep((s,X), \phi) \wedge \app_\Sep((s, Y), \psi)$
					\end{tabular}
					\\
					\midrule
					\multicolumn{1}{l}{Inference rules for behaviours of program statements:}
					\\
					$
					\infer[^{1\ (\Alloc)}]
					{\Gamma \Rightarrow (x := \bff{cons}(e), (s, h))\trans (\ter, (s[x\mapsto n], h\cup \{(n,s(e))\})), \Delta}
					{}
					$
					\ \ 
					$
					\infer[^{(x := [e])}]
					{\Gamma \Rightarrow (x := [e], (s, h))\trans (\ter, (s[x\mapsto h(s(e))], h)), \Delta}
					{}
					$
					\\ 
					$
					\infer[^{([e] := e')}]
					{\Gamma \Rightarrow ([e] := e', (s, h))\trans (\ter, (s, h[s(e)\mapsto s(e')])), \Delta}
					{}
					$
					\ \ 
					$
					\infer[^{(\disAlloc)}]
					{\Gamma \Rightarrow (\disAlloc(e), (s, h))\trans (\ter, (s, h|_{\dom(h)\setminus\{s(e)\}})), \Delta}
					{}
					$
					\\
					\midrule
					\multicolumn{1}{l}{
						$^{1}$ $n$ is new w.r.t. $h$.
					}
					\\
					\bottomrule
				\end{tabular}
			}
		}
	\end{center}
	\caption{Encoding of A Part of Separation Logic in $\LDL$}
	\label{table:Encoding of A Part of Separation Logic}
\end{table}

\textbf{Encoding of Separation Logic in $\LDL$}. 
In $\LDL$, we choose the store-heap pairs as the configurations namely $\sigma_\Sep$. 
Table~\ref{table:Encoding of A Part of Separation Logic} lists the rules for  the program behaviours of the atomic  program statements. 
To capture the semantics of separation logical formulas, interpretations $\app_\Sep$ are defined in an inductive way according to their syntactical structures. 
Note that the formula $\app_\Sep((s,h), \phi\ \sepc\ \psi)$ requires variables $X, Y$ ranging over heaps. 



To tackle a formula like $\app_\Sep((s,h), \phi\ \sepc\ \psi)$, additional rules can be proposed. 
In this example, we can propose a rule
$$
\begin{aligned}
	\infer[^{(\sigma \sepc)}]
	{\Gamma\Rightarrow (s, h_1\cup h_2) : \phi\ \sepc\ \psi, \Delta}
	{
		\Gamma\Rightarrow h_1\bot h_2, \Delta
		&
		\Gamma\Rightarrow (s,h_1) : \phi, \Delta
		&
		\Gamma\Rightarrow (s, h_2) : \psi, \Delta
	}   
\end{aligned}
$$
to decompose the heap's structure, or a rule
$$
\begin{aligned}
	\infer[^{(\sigma \textit{Frm})}]
	{\Gamma\Rightarrow (s, h) : \phi\ \sepc\ \psi, \Delta}
	{\Gamma\Rightarrow (s, h) : \phi, \Delta}
\end{aligned}, \mbox{ if no variables of $\dom(h)$ appear in $\psi$}, 
$$
to simplify the formula's structure. 
These rules are inspired from their
counterparts for programs in separation logic. 

In practice, configurations can be more explicit structures than store-heap pairs $(s, h)$. Similar encoding can be obtained accordingly. 
From this example, we envision that the entire theory of separation logic can actually be embedded into $\LDL$, where additional rules like the above ones will support a ``configuration-level reasoning'' of separation logical formulas. 


\section{Soundness of Cyclic Proof System $\pfDLp$}
\label{section:Proof of Theorem - theo:Soundness of A Cyclic Preproof}
We analyze and prove the soundness of $\pfDLp$ stated as Theorem~\ref{theo:Soundness of A Cyclic Preproof}.  
We need to show that
if a preproof is cyclic (Definition~\ref{def:Cyclic Preproof}), that is, if every derivation path of the preproof is followed by a safe and progressive trace, then the conclusion is valid. 
We follow the main idea behind~\cite{Brotherston08} to carry out a proof by contradiction: 
suppose the conclusion is invalid, then 
we will show that there exists an \emph{invalid derivation path} in which each node is invalid, 
and 
one of its progressive traces starting from a formula $\tau$ in a node $\nu$ leads to an infinite descent sequence of elements along this trace ordered by a well-founded relation $\mult$ (introduced in Section~\ref{section:Well-founded Relation pmult}), which violates the definition of well-foundedness itself. 

Below, we only consider the case when $\tau\in \{\sigma : [\alpha]\phi, \sigma: \la\alpha\ra\phi\}$ in a node $\nu$ of the form $\Gamma\Rightarrow \tau, \Delta$ for arbitrary $\Gamma$ and $\Delta$. 
Other cases for $\tau$ are trivial. 

\ifx
In this subsection, we analyze and prove Theorem~\ref{theo:Soundness of A Cyclic Preproof}. 
We only focus on the case when the conclusion is of the form $\Gamma \Rightarrow \sigma : [\alpha]\phi$ or $\Gamma\Rightarrow \sigma : \la\alpha\ra$, where dynamic $\LDL$ formulas $ \sigma : [\alpha]\phi$ and $\sigma : \la\alpha\ra\phi$ are the only formula on the right side of the sequent.
Other cases are trivial.

To prove Theorem~\ref{theo:Soundness of A Cyclic Preproof}, we need to show that
if a preproof is cyclic (Definition~\ref{def:Cyclic Preproof}), that is, if any derivation path is followed by a progressive trace, then the conclusion is sound. 
We carry out the proof by contradiction following the main idea behind~\cite{Brotherston08}: Suppose the conclusion is not sound, that is, proposition $\mfr{P}(\Gamma\Rightarrow \sigma: [\alpha]\phi)$ (resp. $\mfr{P}(\Gamma\Rightarrow\sigma: \la\alpha\ra\phi)$) is not true,
then we will show that it induces an infinite descent sequence of elements ordered by 
the well-founded relation $\mult$ introduced as follows, which violates Definition~\ref{def:Well-foundedness}. 
\fi

In the rest of this section, we firstly introduce the well-founded relation $\pmult$, then focus on the main skeleton of proving Theorem~\ref{theo:Soundness of A Cyclic Preproof}. 
Other proof details are given in Appendix~\ref{section:Other Propositions and Proofs}. 

\subsection{Well-founded Relation $\pmult$}
\label{section:Well-founded Relation pmult}

\textbf{Well-foundedness}. 
A relation $\preceq$ on a set $S$ is \emph{partially ordered}, if 
it satisfies the following properties: 
(1) Reflexivity. $t\preceq t$ for each $t\in S$. 
(2) Anti-symmetry. For any $t_1, t_2\in S$, if $t_1\preceq t_2$ and $t_2\preceq t_1$, then $t_1$ and $t_2$ are the same element in $S$, we denote by $t_1 = t_2$.
(3) Transitivity. For any $t_1, t_2, t_3\in S$, if $t_1\preceq t_2$ and $t_2\preceq t_3$, then $t_1\preceq t_3$. 

$t_1 \prec t_2$ is defined as $t_1\preceq t_2$ and $t_1\neq t_2$. 


Given a set $S$ and a partial-order relation $\preceq$ on $S$, 
$\preceq$ is called a \emph{well-founded relation} over $S$, if for any element $a$ in $S$, there is no infinite descent sequence: $a \succ a_1 \succ a_2 \succ ...$ in $S$. 
Set $S$ is called a \emph{well-founded} set w.r.t. $\preceq$. 

\textbf{Relation $\mult$}.  
Given two paths $tr_1$ and $tr_2$, 
relation $tr_1 \suf tr_2$ is defined if $tr_2$ is a suffix of $tr_1$. 
Obviously $\suf$ is partially ordered. 
Relation $tr_1\psuf tr_2$ expresses that $tr_1$ is a proper suffix of $tr_2$. 
In a set of finite paths, relation $\suf$ is well-founded,  because every finite path has only a finite number of suffixes. 

\begin{definition}[Relation $\mult$]
	\label{def:Relation pmult}
	Given two finite sets $\cnt_1$ and $\cnt_2$ of finite paths,  
	$\cnt_1\mult \cnt_2$ is defined if either (1) $\cnt_1 = \cnt_2$; or (2) set $\cnt_1$ can be obtained from $\cnt_2$ by replacing (or removing) one or more elements of $\cnt_2$ each with a finite number of elements, such that
	for each replaced element $tr$, its replacements $tr_1,...,tr_n$ ($n\ge 1$) in $\cnt_1$ satisfies that $tr_i\psuf tr$ for any $i$, $1\le i\le n$.  
	
\end{definition}

\begin{proposition}
	\label{prop:relation mult is partially ordered}
	$\mult$ is a partial-order relation. 
\end{proposition}

The proof of Proposition~\ref{prop:relation mult is partially ordered} is given in Appendix~\ref{section:Other Propositions and Proofs}. 


\begin{example}
	Let $C_1=\{tr_1, tr_2, tr_3\}$, where $tr_1 \dddef (\alpha, \sigma)(\alpha_1, \sigma_1)(\alpha_2, \sigma_2)(\alpha_3, \sigma_3)(\ter, \sigma_4), tr_2 \dddef (\alpha, \sigma)(\alpha_1, \sigma_1)(\beta_1, \delta_1)(\beta_2, \delta_2)(\ter, \delta_3)$ and $tr_3 \dddef (\alpha, \sigma)(\ter, \tau)$; 
	$C_2 = \{tr'_1, tr'_2\}$, where 
	$tr'_1 \dddef (\alpha_1, \sigma_1)(\alpha_2, \sigma_2)(\alpha_3, \sigma_3)(\ter, \sigma_4), tr'_2 \dddef (\alpha_1, \sigma_1)(\beta_1, \delta_1)(\beta_2, \delta_2)(\ter, \delta_3)$. 
	We see that $tr'_1 \psuf tr_1$ and $tr'_2 \psuf tr_2$. 
	$C_2$ can be obtained from $C_1$ by replacing $tr_1$ and $tr_2$ with $tr'_1$ and $tr'_2$  respectively, and removing $tr_3$. 
	Hence $C_2\mult C_1$. 
	Since $C_1\neq C_2$, $C_2\pmult C_1$. 
\end{example}

\begin{proposition}
	\label{prop:well-foundedness of relation pmult}
	The relation $\mult$ between any two finite sets of finite paths is a  well-founded relation. 
\end{proposition}
We omit the proof since 
relation $\mult$ is in fact a special case of the ``multi-set ordering'' introduced in~\cite{Dershowitz79}, where it has been proved that multi-set ordering is well-founded. 

\subsection{Proof of Theorem~\ref{theo:Soundness of A Cyclic Preproof}}
Below we give the main skeleton of the proof by skipping the details of the critical Lemma~\ref{lemma:infinite descent sequence}, whose proof can be found in Appendix~\ref{section:Other Propositions and Proofs}. 

Following the main idea in the beginning of Section~\ref{section:Proof of Theorem - theo:Soundness of A Cyclic Preproof}, 
next we first introduce the notion of ``counter-example paths'' that makes a formula $\sigma : [\alpha]\phi$ or $\sigma : \la\alpha\ra \phi$ invalid.   


\begin{definition}[Counter-example Paths]
	\label{def:counter-example set}
	A ``counter-example'' $\CT(\rho, \tau, \nu)$ of a formula $\tau\in \{\sigma: [\alpha]\phi, \sigma:\la\alpha\ra\phi\}$ in a sequent $\nu\dddef (\Gamma\Rightarrow \tau, \Delta)$ is a set of minimum paths defined as:
	$
	\CT(\rho, \tau, \nu)\dddef \{\rho(\alpha, \sigma)...(\ter, \sigma')\ |\ \rho\models \Gamma, \rho\not\models \sigma' : \phi
	\}. 
	$
\end{definition}

We write $\rho\models \Gamma$ to mean that $\rho\models \phi$ for all $\phi\in \Gamma$.
Recall that a path being minimum is defined in Section~\ref{section:Syntax and Semantics of LDL}.

\begin{lemma}
	\label{lemma:infinite descent sequence}
	In a cyclic preproof (where there is at least one derivation path), 
	let $(\tau, \tau')$ be a pair of a derivation trace over a node pair $(\nu, \nu')$ of an invalid derivation path, where $\tau\in \{\sigma : [\alpha]\phi, \sigma : \la\alpha\ra\phi\}$ and $\nu\dddef (\Gamma\Rightarrow \tau, \Delta)$. 
	For any counter-example $\CT(\rho, \tau, \nu)$ of $\tau$, 
	there exists an evaluation $\rho'$ and a counter-example $\CT(\rho', \tau', \nu')$ of $\tau'$ such that $\CT(\rho', \tau', \nu')\mult \CT(\rho, \tau, \nu)$. Moreover, if $(\tau, \tau')$ is a progressive step, then $\CT(\rho', \tau', \nu')\pmult \CT(\rho, \tau, \nu)$. 
\end{lemma}

\ifx
\begin{lemma}
	\label{lemma:infinite descent sequence}
	In a cyclic preproof (where there is at least one derivation path), let $\tau\in \{\sigma: [\alpha]\phi, \sigma:\la\alpha\ra\phi\}$ be a formula in a node $\nu\dddef (\Gamma\Rightarrow \tau)$, 
	$\rho$ is an evaluation satisfying that $\rho\models \phi$ for all $\phi\in \Gamma$ and  
	$F(\rho, \alpha, \sigma, \phi)$ is a non-empty, finite counter-example of $\tau$. 
	Then there exists a formula pair $(\tau, \tau')$ appearing in a node pair $(\nu, \nu')$ and a set $F(\rho, \alpha', \sigma', \phi)$ such that $F(\rho, \alpha', \sigma', \phi)$ is a non-empty, finite counter-example of $\tau'$, and $F(\rho, \alpha', \sigma', \phi) \mult F(\rho, \alpha, \sigma, \phi)$. 
	Moreover, if $(\tau, \tau')$ is a progressive step, then $F(\rho, \alpha', \sigma', \phi) \pmult F(\rho, \alpha, \sigma, \phi)$. 
\end{lemma}
\fi

In Lemma~\ref{lemma:infinite descent sequence}, note that if 
there is no derivation paths, then the preproof itself must be a finite proof tree. 
Then its soundness is straightforward according to the soundness of rules in Table~\ref{table:General Rules for LDL}. 

\ifx
Lemma~\ref{lemma:infinite descent sequence} shows the critical step of proving Theorem~\ref{theo:Soundness of A Cyclic Preproof} that 
along a derivation step $(\nu, \nu')$, we can obtain a counter-example $\CT(\rho', \tau', \nu')$ of a formula $\tau'$ in $\nu'$ from the counter-example $\CT(\rho, \tau, \nu)$ of a formula $\tau$ in $\nu$. 

Intuitively, Lemma~\ref{lemma:infinite descent sequence} shows that along an arbitrary derivation step from a dynamic formula $\tau$ to a formula $\tau'$, the invalidity of $\tau$ leads to the invalidity of $\tau'$, and the counter-example paths $F(\rho, \alpha', \sigma', \phi)$ of $\tau'$ is less than the counter-example paths $F(\rho, \alpha, \sigma, \phi)$ of $\tau$ w.r.t. relation $\pmult$. Moreover, a progressive derivation step guarantees a strictly less-than relation between $F(\rho, \alpha', \sigma', \phi)$ and $F(\rho, \alpha, \sigma, \phi)$. 
\fi

Based on Theorem~\ref{theo:soundness of rules for LDL} and Lemma~\ref{lemma:infinite descent sequence}, we give the proof of Theorem~\ref{theo:Soundness of A Cyclic Preproof} as follows. 

\begin{proof}[Proof of Theorem~\ref{theo:Soundness of A Cyclic Preproof}]
	As stated previously, we prove by contradiction and only focus on the cases where 
	the first formula $\tau_1$ of a progressive trace is of the form: $\sigma : [\alpha]\phi$ or $\sigma : \la\alpha\ra\phi$. 
	Let the progressive trace be $\tau_1\tau_2...\tau_m...$ over a derivation path
	$...\nu_1\nu_2...\nu_m...$ ($m\ge 1$) starting from $\tau_1$ in $\nu_1$, where $\nu_1\dddef (\Gamma\Rightarrow \tau_1, \Delta)$, and formula $\mfr{P}(\nu_i)$  $(i\ge 1)$ for each node $\nu_i$ is invalid.

	From that $\mfr{P}(\Gamma\Rightarrow \tau_1, \Delta)$ is invalid, 
	there is an evaluation $\rho$ satisfying that $\rho\models \Gamma$ but $\rho\not\models \tau_1$, so $\CT_1(\rho, \tau_1, \nu_1)$ is well defined. 
	By Lemma~\ref{lemma:infinite descent sequence}, 
	from $\CT_1$ we can obtain an infinite sequence of counter-examples: 
	$\CT_1,\CT_2,...,\CT_m,...$, with each $\CT_i$ ($i\ge 1$) being the counter-example of formula $\tau_i$, and which satisfies that
	$\CT_1\multr \CT_2\multr...\multr \CT_m\multr...$. 
	Moreover, since $\tau_1\tau_2...\tau_m...$ is progressive (Definition~\ref{def:Cyclic Preproof}), 
	there must be an infinite number of strict relation $\pmult$ among these $\mult$s. This thus forms an infinite descent sequence w.r.t. $\pmult$, violating 
	the well-foundedness of relation $\mult$ (Proposition~\ref{prop:well-foundedness of relation pmult}). 
	
	\ifx
	Let $\tau\in \{\sigma:[\alpha]\phi, \sigma:\la\alpha\ra\phi\}$. 
	Suppose $\mfr{P}(\Gamma\Rightarrow \tau)$ is not true. 
	Then there is an evaluation $\rho$ satisfying that 
	$\rho\models \Gamma$ but $\rho\not\models \tau$, 
	making the counter-example set $F(\rho, \alpha, \sigma, \phi)$ non-empty. 
	By~\ref{item:Termination Finiteness} of Definition~\ref{def:Program Properties}, 
	$F(\rho, \alpha, \sigma, \phi)$ is also finite. 
	According to Lemma~\ref{lemma:infinite descent sequence}, 
	we can obtain an infinite derivation trace $\tau_1\tau_2...\tau_m...$ over a derivation path
	$\nu_1\nu_2...\nu_m...$ ($m\ge 1$) and each set $F_i$ (with $F_1\dddef F(\rho, \alpha, \sigma, \phi)$) as a non-empty finite counter-example of formula $\tau_i$ ($i\ge 1$). 
	By Lemma~\ref{lemma:infinite descent sequence},  from these counter-examples we obtain an infinite relation chain:
	$F_1\multr F_2\multr...\multr F_m\multr...$. 
	Moreover, since $\tau_1\tau_2...\tau_m...$ is progressive (Definition~\ref{def:Cyclic Preproof}), there must be an infinite number of strict relation $\pmultr$ among these relations. This thus forms an infinite descent sequence w.r.t. relation $\pmultr$, violating 
	the well-foundedness of relation $\pmultr$ (Proposition~\ref{prop:well-foundedness of relation pmult}). 
	
	\fi
\end{proof}


\ifx
Therefore, it is impossible to discuss the completeness of its proof system, unless the structures of formulas, programs and configurations are given specifically.
Whether the proof system of a ``specific'' $\LDL$, which is given by specific formulas, programs and configurations, is complete or not, depends on whether we can build a cyclic preproof for any formula in this logic, by constructing suitable configurations so that every infinite derivation path has a progressive derivation trace.
\fi

\section{Related Work}
\label{section:Related Work} 
The idea of reasoning about programs directly based on program  behaviours is not new. 
Previous work such as~\cite{Rosu12,Rosu13,Stefanescu14,X.Chen19} in last decade has addressed this issue using theories based on rewriting logics~\cite{Meseguer12}. 
Matching logic~\cite{Rosu12} is based on patterns and pattern matching. 
Its basic form, a reachability rule $\varphi\Rightarrow \varphi'$ (where $\Rightarrow$ means differently from in sequents here), captures whether pattern $\varphi'$ is reachable from pattern $\varphi$ in a given pattern reachability system. 
The concept of `patterns' in matching logic has a wider scope of meaning than the concepts of `programs', `configurations' and `formulas' in $\LDL$. 
Based on matching logic, one-path and all-paths reachability logics~\cite{Rosu13,Stefanescu14} were developed by enhancing the expressive power of the reachability rule. 
A more powerful matching $\mu$-logic~\cite{X.Chen19} was proposed by adding a least fixpoint $\mu$-binder to matching logic. 

In terms of expressiveness, the semantics of modality $[\cdot]$ in dynamic logic cannot be fully captured by matching logic and one-path reachability logic when the programs are non-deterministic. 
We conjecture that matching $\mu$-logic can encode $\LDL$.
As a generic calculus, it has been claimed that matching $\mu$-logic can encode traditional dynamic logics (cf.~\cite{X.Chen19}). 
Compared to these theories based on rewriting logics, 
$\LDL$ has a more specific logical structure and belongs to a type of modal logics, with modalities $[\cdot]$ and $\la\cdot\ra$ to efficiently capture and reason about program specifications. 
On the other hand, to encode a $\LDL$ specification $\phi\to (\sigma : \la\alpha\ra \psi)$ in matching logic for example,  
one needs to somehow `mix up' all structures together to form a pattern, in a form like: $\varphi\dddef \alpha\wedge \sigma\wedge \phi$ and $\varphi'\dddef \ter\wedge \sigma'\wedge \psi$ for some configuration $\sigma'$, in order to conduct a reachability rule $\varphi\Rightarrow{} \varphi'$.

\ifx
Compared to matching, reachability and matching-$\mu$ logics, $\LDL$ is based on a different proof theory that is deeply rooted in the theories of dynamic logics and cyclic reasoning rather than coinduction~\cite{Rosu13}. 
One advantage of $\LDL$ is that $\LDL$ comes to program deductions in a more natural way in the sense that program models and logical formulas are explicitly expressed and clearly separated by modality $[\cdot]$. 
On the other hand, to express a $\LDL$ specification, like $(\sigma : \phi)\Rightarrow{} (\sigma : \la\alpha\ra \psi)$, in matching logic for example, 
one needs to somehow `mix up' all structures together to form a pattern, in a form like: $\varphi\dddef \alpha\wedge \sigma\wedge \phi$ and $\varphi'\dddef \ter\wedge \sigma'\wedge \psi$ for some configuration $\sigma'$, in order to conduct the reachability rule $\varphi\Rightarrow{} \varphi'$.  
\fi

\ifx
This is against the philosophy we uphold in this paper 
to separate programs and formulas throughout the entire derivation process as in traditional dynamic logics.   
\fi

\cite{Moore18}~proposed a general program verification framework based on coinduction. 
Using the terminology in this paper, 
in~\cite{Moore18} a program specification can be expressed as a pair $(c, P)$ with $c$ a program state and $P$ a set of program states, capturing the exact meaning of 
formula $\sigma : [\alpha]\phi$ in $\LDL$ if we let $c$ be $(\alpha, \sigma)$ and let $P$ represent the semantics of formula $\phi$.   
The authors designed a method to derive a program specification $(c, P)$ in a coinductive way 
according to the operational semantics of $c$. 
Following~\cite{Moore18}, \cite{X.Li21}~also proposed a general framework for program reasoning,  but via big-step operational semantics. 
Unlike the frameworks in~\cite{Moore18} and~\cite{X.Li21}, which are directly based on mathematical set theory, 
$\LDL$ is a logic, which in some cases is more convenient to catch program properties.   
In terms of expressiveness, $\LDL$ can also express the negation of a dynamic formula $\sigma : [\alpha]\phi$ as the dual formula $\sigma : \la\alpha\ra\neg \phi$, whose meaning cannot be captured in the framework of~\cite{Moore18}.

The structure `updates' in dynamic logics, used in Java dynamic logic~\cite{Beckert2016}, differential dynamic logic~\cite{Platzer07b}, dynamic trace logic~\cite{Beckert13}, etc., works as a special case of the configurations proposed in this paper. 
As illustrated in Section~\ref{section:Example One: A While Program}, a configuration in $\LDL$ can be more than just a ``delay substitution'' (cf.~\cite{Beckert13}) of variables and terms. 

The proof system of $\LDL$ relies on the cyclic proof theory which firstly arose in~\cite{Stirling91} and was later developed in 
different logics such as~\cite{Brotherston07} and \cite{Brotherston08}. 
Traditional dynamic logics' proof systems are not cyclic ones.  
In~\cite{Jungteerapanich09}, 
Jungteerapanich proposed a complete cyclic proof system for $\mu$-calculus, which subsumes PDL~\cite{Fischer79} in its expressiveness. In~\cite{Docherty19}, Docherty and Rowe proposed a complete labelled cyclic proof system for PDL. 
Both $\mu$-calculus and PDL, as mentioned in Section~\ref{section:Summery}, are logics designed for dissolving regular expressions as their program models.  
In terms of program structures, however, $\LDL$ is relatively simpler than PDL, in which programs can also depend on dynamic formulas. 
The labelled form of $\LDL$ formula $\sigma : [\alpha]\phi$ is inspired from~\cite{Docherty19}, where a labelled PDL formula is of the form $s : [\alpha]\phi$, with $s$ a world in a Kripke structure. 
While in $\LDL$, a configuration $\sigma$ can be an explicit term. 

There has been some other work for generalizing the theories of dynamic logics, such as~\cite{MossakowskiEA09,Hennicker19}. 
However, generally speaking, they neither consider structures as general as allowed by configurations and programs in $\LDL$, nor adopt a similar approach for reasoning about programs. 
\cite{MossakowskiEA09} proposed a complete generic dynamic logic for monads of programming languages. 
In~\cite{Hennicker19}, the authors proposed a dynamic logic for reasoning about program models that are more general than regular expressions of PDL, also based on program behaviours (where it is called ``interaction-based'' behaviours). 
But there the program states are still abstract worlds of Kripke structures as in PDL, unlike the configurations in $\LDL$. 
And yet no proof systems have been proposed and analyzed for the logic.  

\section{Some Discussions \&\ Future Work}
\label{section:Discussions and Future Work}

In this paper, we propose a dynamic-logic-like parameterized formalism allowing instantiations of a rich set of programs and formulas in interested domains. 
$\LDL$ provides a flexible verification framework to encompass existing dynamic-logic theories, by supporting a direct symbolic-execution-based reasoning according to programs behaviours, while still preserving structure-based reasoning through lifting processes. 
In practical verification, trade-off can be made between these two types of reasoning.  
\ifx
Essentially, $\LDL$ is a second-ordered logic restricted to a particular form that facilitates expressing and reasoning about programs and systems in a dynamic-logic style. $\LDL$ provides a flexible verification framework to encompass different dynamic logic theories, in which different types of reasoning are offered to make a trade-off between structure-based reasoning and symbolic-execution-based reasoning. 
\fi
\ifx
$\LDL$ can be a very powerful logic calculus due to its general forms. 
However, we do not attempt to propose a ``unified framework'' for programming and verifications, just as what matching-$\mu$ logic and UTP theory~\cite{Hoare98} aim for. 
Instead, we would rather treat $\LDL$ as a second-ordered logic restricted to a particular form that facilitates expressing and reasoning about programs and systems in a dynamic-logic style. 
\fi
Through examples displayed in this paper, one can see the potential of $\LDL$ to be used for other types of programs, especially those do not support structural rules, such as transitional models~\cite{Hennicker19} or neural networks~\cite{XiyueZhang22}. 
On the other hand, due to its generality, $\LDL$ loses the completeness that a dynamic logic usually possesses, and must rely on additional lifted rules to support a compositional reasoning.

\ifx
Though we only display an example in this paper, it can be easily seen that $\LDL$ can be used to reason about many other types of programs, especially those do not support structural rules, such as neural networks, as done recently in~\cite{XiyueZhang22}, and also many other abstract program models, such as Kleene algebra with tests~\cite{kozen97} and CCS~\cite{Milner82}.  
\fi

\ifx
$\LDL$ can be taken as a general theory that can subsume the existed dynamic logic theories in which 
program models satisfy the program properties declared in Definition~\ref{def:Program Properties}. Any rule in a dynamic logic theory can be lifted as a labelled rule in $\LDL$, which is derived without changing the labels of a formula.  
For instance, in the example discussed in Section~\ref{section:Summery}, 
to derive formula $\psi\dddef (\{x\mapsto t\} : x\ge 0\Rightarrow{} \{x\mapsto t\} : [x := x + 1]x > 0)$, one can also apply the ``lifting version'' of the assignment rule
($x := e$) as: 
$\begin{aligned} \infer[]
	{\sigma : [x := e]\psi}
	{\sigma : \psi[x / e]}\end{aligned}$, 
and $\psi$ then becomes $\psi'''\dddef (\{x\mapsto t\} : x\ge 0)\to (\{x\mapsto t\} : x + 1 > 0)$, which also can be transformed into formula $\phi'$: $t\ge 0\to t + 1 > 0$ after the applications of $\{x\mapsto t\}$. 
This ``lifting ability'' provides $\LDL$ with a flexible framework in which 
different inference rules can be applied to make a trade-off between 
structure-based reasoning and symbolic executions in practical deduction processes. 
More work will be on analyzing how program verification can benefit from this flexibility of derivation. 
\fi

\ifx
This ``lifting ability'' allows $\LDL$ to subsume special dynamic logic theories and 
provides a flexible framework in which 
different inference rules can be applied to make a trade-off between 
compositional reasoning and symbolic executions in practical deduction processes. 
\fi

\ifx
notice that unlike matching-$\mu$ logic and UTP theory~\cite{???}, 
$\LDL$ never intend to be a ``unified framework'' for programming and verification. 
We keep stressing that the main advantage of $\LDL$ is that 
a $\LDL$ formula is still in a dynamic logical form $[\alpha]\phi$ (with labels $\sigma$), which in our opinion is convenient for program derivations. 
\fi


About future work, 
currently we plan a full mechanization of $\LDL$ using Coq~\cite{Bertot04}, which on one hand can help verifying the correctness of the theory itself, and on the other hand provides a verification framework built upon Coq to support verifying general programs. 
Our long-term goal is to develop an independent program verifier based on the theory of $\LDL$, which could fully make use of the advantages of symbolic-execution-based reasoning and cyclic proofs to maximize automation of verifications. 
One possible way is based on \textit{Cyclist}~\cite{Brotherston12}, a proof engine supported by an efficient cyclic-proof-search algorithm. 
To see the full potential of $\LDL$, we are also trying to use $\LDL$ to describe and verify more types of programs or system models.



\ifx
The general framework provided by $\LDL$ could be helpful for maximizing the automation of these theorem provers for program verification. 
To see the full potential of $\LDL$, we are also trying to use $\LDL$ to describe and verify more types of programs or structures. 
\fi

\ifx
\section{Conclusion and Future Work}
\label{section:Conclusion}
We propose a dynamic-logic-like formalism called labelled dynamic logic ($\LDL$) which provides a general framework to support direct program reasoning via operational semantics. 
$\LDL$ allows general forms of programs, configurations and formulas as `parameters' 
and inherits from traditional dynamic logics the modal operator $[\cdot]$. 
We build a cyclic proof system for $\LDL$, which relies a cyclic preproof structure to gurrantee its soundness. 
We formally prove the soundness of $\LDL$, making use of the well-foundedness of partial-order relations related to the semantics of $\LDL$ dynamic formulas. 
Like other dynamic logics and Hoare-style logics, $\LDL$ separates program models and formulas in its expressions to allow a clear observation of program evolution during the deduction processes. 
We analyze two case studies to show how $\LDL$ can be used to reason about different types of programs based on their operational semantics.

$\LDL$ can be a very powerful tool due to its general forms. 
One future work will focus on mechanization of $\LDL$ in theorem provers like Coq~\cite{Bertot10}. 
The general framework provided by $\LDL$ could be helpful for maximizing the automation of these theorem provers for program verification. 
To see the full potential of $\LDL$, we are also trying to use $\LDL$ to describe and verify more types of programs or structures. 
\fi

\paragraph{Acknowledgment}
This work is partially supported by the Youth Project of National Science Foundation of China (No. 62102329), the Project
of National Science Foundation of China (No. 62272397), and the New Faculty Start-up Foundation of NUAA (No. 90YAT24003).


\bibliographystyle{splncs04}
\bibliography{main}

\appendix

\section{Other Propositions and Proofs}
\label{section:Other Propositions and Proofs}

\ifx
\begin{lemma}
	For any evaluation $\rho_A$ and substitution $t[t'/x]$ given free-variable set $A$, $t'\in \TA$ and $x\in \Var$, the following properties hold: 
	\begin{enumerate}
		\item $t\equiv t[x/x]$;
		\item $t\equiv t[t'/x]$, if $x\notin \FV(t)$;
		\item $\rho_A(t[t'/x]) \equiv \rho_{A\setminus\{x\}}(t)[\rho_A(t')/x]$.
	\end{enumerate}
\end{lemma}

\begin{proof}
	???
\end{proof}
\fi

\begin{lemma}
	\label{lemma:from logical consequence to sequent derivation}
	For any formulas $\phi_1,...,\phi_n, \phi$, 
	if formula $\phi_1\wedge ...\wedge\phi_n\to \phi$ is valid, 
	then sequent 
	$$
	\begin{aligned}
		\infer[]
		{
			\Gamma\Rightarrow \phi, \Delta
		}
		{
			\Gamma\Rightarrow \phi_1, \Delta
			&
			...
			&
			\Gamma\Rightarrow \phi_n, \Delta
	}\end{aligned}$$ is sound for any contexts $\Gamma, \Delta$. 
\end{lemma}

\begin{proof}
	Assume $\mfr{P}(\Gamma\Rightarrow \phi_1, \Delta)$,...,$\mfr{P}(\Gamma\Rightarrow \phi_n, \Delta)$ are valid, 
	for any evaluation $\rho\in \Eval$, let $\rho\models \Gamma$ and $\rho\not\models \psi$ for all $\psi\in \Delta$, 
	we need to prove $\rho\models \phi$. 
	From the assumption we have $\rho\models \phi_1$, ..., $\rho\models \phi_n$. 
	$\rho\models \phi$ is an immediate result since
	$\phi_1\wedge ...\wedge \phi_n\to \phi$. 
\end{proof}

{\parindent 0pt
	\textbf{Content of Theorem~\ref{theo:soundness of rules for LDL}: }
	Provided that proof system $P_{\Oper, \Terminate}$ is sound, all rules in $\mcl{DL}_p$ (Table~\ref{table:General Rules for LDL}) are sound. 
}

Below we only prove the soundness of rules $([\alpha])$, $(\la\alpha\ra)$ and $(\Sub)$. 
Other rules can be proved similarly based on the semantics of $\LDL$ formulas (Definition~\ref{def:Semantics of Labelled LDL Formulas}).  
By Lemma~\ref{lemma:from logical consequence to sequent derivation}, 
it is sufficient to prove $\phi_1\wedge...\wedge\phi_n\to \phi$, that is, for any evaluation $\rho$, if $\rho\models \phi_1$, ..., $\rho\models\phi_n$, then $\rho\models \phi$. 

\begin{proof}[Proof of Theorem~\ref{theo:soundness of rules for LDL}]        
	For rule $([\alpha])$, by the soundness of proof system $\vdash_{P_{\Oper, \Terminate}} \Gamma\Rightarrow (\alpha, \sigma)\trans (\alpha', \sigma'), \Delta$, it is sufficient to prove that 
	for any evaluation $\rho$, if $\rho\models \sigma' : [\alpha']\phi$ and $(\rho(\alpha,\sigma)\trans\rho(\alpha',\sigma'))\in \Oper$ for all $(\alpha', \sigma')\in \Phi$, then $\rho\models \sigma : [\alpha]\phi$. 
	However, this is straightforward by Definition~\ref{def:Semantics of Labelled LDL Formulas} and the completeness of $P_{\Oper, \Terminate}$ w.r.t. $\Oper$, because since $\alpha\not\equiv \ter$, 
	any path starting from $\rho(\alpha, \sigma)$ is of the form: 
	$\rho(\alpha, \sigma)\rho(\alpha', \sigma')...(\ter, \sigma'')$ for some $(\alpha', \sigma')\in \Phi$. 

	For rule $(\la\alpha\ra)$, by the soundness of proof system $\vdash_{P_{\Oper, \Terminate}} \Gamma\Rightarrow (\alpha, \sigma)\trans (\alpha', \sigma'), \Delta$, it is sufficient to prove that for any evaluation $\rho$, 
	if $\rho\models \sigma' : \la\alpha'\ra \phi$ and $(\rho(\alpha,\sigma)\trans \rho(\alpha', \sigma'))\in \Oper$ for some $(\alpha', \sigma')$, then 
	$\rho\models \sigma : \la\alpha\ra\phi$. 
	However, this is direct by Definition~\ref{def:Semantics of Labelled LDL Formulas}. 
	
	For rule $(\Sub)$, it is sufficient to prove that for any evaluation $\rho\in \Eval$, there exists an evaluation $\rho'\in \Eval$ such that for each formula $\phi\in \Gamma\cup \Delta$, 
	we have $\rho(\Sub(\phi))\equiv \rho'(\phi)$. 
	But this just matches the definition of $\Sub$ in Definition~\ref{def:Abstract Substitution}. 
	
	\ifx
	For rule $(\sigma \textit{Sub})$, 
	for any evaluation $\rho\in \Eval$, let $\rho'\dddef \rho[x\mapsto \rho(t)]$. 
	By the definition of substitutions in Section~\ref{section:Terms and Substitutions}, we can have that 
	for any term $u\in \TA$, $\rho(u[t/x]) \equiv \rho|_{\FV(u)\setminus\FV_x(u)}(u)[\rho(t)/x] \equiv \rho'^*|_{\FV(u)\setminus\FV_x(u)}(u)[\rho(t)/x] \equiv \rho'^*(u)$. 
	Therefore, it is easy to prove that $\mfr{P}(\Gamma\Rightarrow \Delta)$ implies $\mfr{P}(\Gamma[t/x]\Rightarrow \Delta[t/x])$. 
	\fi
\end{proof}

\ifx
{\parindent 0pt
	\textbf{Content of Lemma~\ref{lemma:infinite descent sequence}:}  
	In a cyclic preproof (where there is at least one derivation path), 
	let $(\tau, \tau')$ be a pair of a derivation trace over a node pair $(\nu, \nu')$ of an invalid derivation path, where $\tau\in \{\sigma : [\alpha]\phi, \sigma : \la\alpha\ra\phi\}$ and $\nu\dddef (\Gamma\Rightarrow \tau, \Delta)$. 
	For any counter-example $\CT(\rho, \tau, \nu)$ of $\tau$, 
	there exists an evaluation $\rho'$ and a counter-example $\CT(\rho', \tau', \nu')$ of $\tau'$ such that $\CT(\rho', \tau', \nu')\mult \CT(\rho, \tau, \nu)$. Moreover, if $(\tau, \tau')$ is a progressive step, then $\CT(\rho', \tau', \nu')\pmult \CT(\rho, \tau, \nu)$.  
}

Before proving Lemma~\ref{lemma:infinite descent sequence}, we introduce the concept of \emph{corresponding path} and prove Lemma~\ref{lemma:construction of terminal path}. 

Intuitively, a corresponding path is a result of program transitions made by the derivations of instances of rules $([\alpha])$ and $(\sigma\la\sigma\ra)$ along a derivation trace. 

\begin{definition}[Corresponding Paths]
	\label{def:Corresponding Path}
	Given a derivation trace $\tau_1\tau_2...\tau_m...$ over a derivation path $\nu_1\nu_2...\nu_m...$ ($m\ge 1$), where $\tau_i\equiv\sigma_i:[\alpha_i]\phi$ (resp.  $\tau_i\equiv\sigma_i:\la\alpha_i\ra\phi$) for each $i\ge 1$. 
	Each pair $(\tau_i, \tau_{i+1})$ either satisfies $\tau_i\equiv \tau_{i+1}$ or $(\nu_i, \nu_{i+1})$ is a conclusion-premise pair of an instance of rule $([\alpha])$ (resp. $(\la\alpha\ra)$). 
	The ``corresponding sequence'' of the derivation trace $\tau_1\tau_2...\tau_m...$ is a sequence of the form: $(l_1, \alpha_1, \sigma_1)...(l_m, \alpha_m, \sigma_m)...$ ($m\ge 1$), where each $l_i\in \{0, 1\}$ satisfies that 
	\begin{enumerate}
		\item $l_i = 0$, if $\tau_i\equiv \tau_{i+1}$;
		\item $l_i = 1$, if $(\nu_i, \nu_{i+1})$ is a conclusion-premise pair of an instance of rule $([\alpha])$ (resp. $(\la\alpha\ra)$). 
	\end{enumerate}
	Given an evaluation $\rho$, 
	the ``corresponding path'' of the derivation trace $\tau_1\tau_2...\tau_m...$ is the sequence of program states obtained from the corresponding sequence $(l_1, \alpha_1, \sigma_1)...(l_m, \alpha_m, \sigma_m)...$ by only keeping pairs $(\rho(\alpha_i), \rho(\sigma_i))$ whose $l_i = 1$ (while maintaining the order of these pairs the same as the corresponding elements in the sequence).

	\ifx
	Given a derivation trace $\tau_1\tau_2...\tau_m...$ over a derivation path $n_1n_2...n_m...$ ($m\ge 1$), where $\tau_1\in \{\sigma:[\alpha]\phi, \sigma:\la\alpha\ra\phi\}$, each pair $(\tau_i, \tau_{i+1})$ ($i\ge 1$) either satisfies $\tau_i\equiv \tau_{i+1}$ or $(n_i, n_{i+1})$ is a conclusion-premise pair of an instance of rule $([\alpha])$ or $(\la\alpha\ra)$,  
	a ``corresponding path'' $(\alpha_1, \sigma_1)...(\alpha_n, \sigma_n)...$ ($n\ge 1$) of the derivation trace $\tau_1\tau_2...\tau_m...$ satisfies that 
	$\alpha_1\in \Clo(\alpha), \sigma_1\in \Clo(\sigma)$, and for any $(\alpha_k,\sigma_k)\trans (\alpha_{k+1}, \sigma_{k+1})$ ($k \ge 1$) of the path and formula $\tau_i\dddef \sigma_i : [\alpha_i]\phi$ (resp. $\tau_i\dddef \sigma_i : \la\alpha_i\ra\phi$, $i\ge 1$) such that $\alpha_k\in \Clo(\alpha_i)$ and $\sigma_k\in \Clo(\sigma_i)$, 
	there is a pair $(\tau_j, \tau_{j+1})$ over $(n_j, n_{j+1})$ ($i \le j$) such that 
	$\tau_i\equiv ...\equiv \tau_j$, $(n_j, n_{j+1})$ is a conclusion-premise pair of an instance of rule $([\alpha])$ (resp. rule $(\la\alpha\ra)$), and $\tau_{j+1}\dddef \sigma_{j+1}[\alpha_{j+1}]\phi$ (resp. $\tau_{j+1}\dddef\sigma_{j+1}\la\alpha_{j+1}\ra\phi$) such that
	$\alpha_{k+1}\in \Clo(\alpha_{j+1})$ and $\sigma_{k+1}\in \Clo(\sigma_{j+1})$. 
	\fi
\end{definition}

\begin{lemma}
	\label{lemma:construction of terminal path}
	In a cyclic preproof of sequent $\Gamma \Rightarrow \sigma:\la\alpha\ra \phi$, 
	given an evaluation $\rho$, 
	any corresponding path $tr$ starting from $(\rho(\alpha), \rho(\sigma))$ of an infinite derivation trace $\tau_1\tau_2...\tau_m...$ ($m\ge 1$) starting from $\tau_1\dddef \sigma: \la\alpha\ra \phi$ eventually terminates.   
\end{lemma}

\begin{proof}
	By Definition~\ref{def:Corresponding Path} and rule $(\la\alpha\ra)$, 
	path $tr$ is of the form $(\alpha_1, \sigma_1)...(\alpha_n, \sigma_n)...$ satisfying that $\alpha_1\equiv \rho(\alpha)$, $\sigma_1\equiv \rho(\sigma)$ and 
	$(\alpha_1,\sigma_1)\succeq ...\succeq (\alpha_n, \sigma_n)\succeq...$. 
	If $tr$ does not terminate, since the derivation trace $\tau_1\tau_2...\tau_m...$ is progressive (by Definition~\ref{def:Cyclic Preproof}), 
	then alone these relation $\succeq$ between the elements of path $tr$, there must be 
	an infinite number of relations $\succ$. 
	This violates the well-foundedness of relation $\prec$ (see Definition~\ref{def:Well-foundedness}). 
	
	\ifx
	Since every derivation trace is progressive, according to rule $(\la\alpha\ra)$ and Definition~\ref{def:Corresponding Path}, if path $tr$ does not terminate, 
	then along the path there exists an infinite sequence of program states $(\alpha_1, \sigma_1),...,(\alpha_n, \sigma_n),...$ 
	satisfying that 
	$(\alpha_1, \sigma_1)\succ...\succ (\alpha_n, \sigma_n)\succ...$. 
	This violates the well-foundedness of relation $\prec$ (see Definition~\ref{def:Well-foundedness}). 
	\fi
\end{proof}
\fi

\ifx
{\parindent 0pt
	\textbf{Content of Proposition~\ref{prop:relation mult is partially ordered}:}  
	The following properties hold:
	\begin{enumerate}
		\item given two finite sets $\cnt_1, \cnt_2$ of finite paths, 
		$\cnt_1\multeq \cnt_2$ iff $\cnt_1 = \cnt_2$;
		\item $\mult$ is a partial-order relation. 
	\end{enumerate} 
}
\fi

{\parindent 0pt
	\textbf{Content of Proposition~\ref{prop:relation mult is partially ordered}:}  
	$\mult$ is a partial-order relation. 
}

\begin{proof}[Proof of Proposition~\ref{prop:relation mult is partially ordered}]
	The reflexivity is trivial. 
	The transitivity can be proved by the definition of `replacements' as described in Definition~\ref{def:Relation pmult} and the transitivity of relation $\psuf$.  
	Below we only prove the anti-symmetry. 
	
	For any finite sets $D_1, D_2$ of finite paths, if $D_1\mult D_2$ but $D_1\neq D_2$, let $f_{D_1,D_2} : D_1\to D_2$ be the function 
	defined such that for any $tr\in D_1$, either (1) $f_{D_1,D_2}(tr) \sufeq tr$; or (2) $tr$ is one of the replacements of a replaced element $f_{D_1,D_2}(tr)$ in $D_2$ with $tr\psuf f_{D_1,D_2}(tr)$. 
	
	For the anti-symmetry, 
	suppose $\cnt_1\mult\cnt_2$ and $\cnt_2\mult\cnt_1$ but $\cnt_1\neq \cnt_2$. 
	Let $tr\in \cnt_1$ but $tr\notin \cnt_2$. 
	Then from $\cnt_1\mult \cnt_2$, we have $tr\psuf f_{\cnt_1, \cnt_2}(tr)$. 
	Since $f_{\cnt_1, \cnt_2}(tr)$ is the replaced element, so $f_{\cnt_1, \cnt_2}(tr)\notin \cnt_1$. 
	By $\cnt_2\mult\cnt_1$, we have $f_{\cnt_1, \cnt_2}(tr)\psuf f_{\cnt_2, \cnt_1}(f_{\cnt_1, \cnt_2}(tr))$. 
	Continuing this process, we can construct an infinite descent sequence
	$tr\psuf f_{\cnt_1, \cnt_2}(tr)\psuf f_{\cnt_2, \cnt_1}(f_{\cnt_1, \cnt_2}(tr))\psuf ...$ w.r.t. relation $\suf$, which violates its well-foundedness. 
	So the only possibility is $\cnt_1 = \cnt_2$. 
	
	\ifx
	\textbf{Proof of 1}. 
	If $\cnt_1 = \cnt_2$, obviously we have $\cnt_1\multeq\cnt_2$ because both $\cnt_1\mult \cnt_2$ and $\cnt_2\mult \cnt_1$ hold. 
	
	For any finite sets $D_1, D_2$ of finite paths, if $D_1\mult D_2$ but $D_1\neq D_2$, let $f_{D_1,D_2} : D_1\to D_2$ be the function 
	defined such that for any $tr\in D_1$, either (1) $f_{D_1,D_2}(tr) \sufeq tr$; or (2) $tr$ is one of the replacements of $f_{D_1,D_2}(tr)$ in $D_2$ as described in Definition~\ref{def:Relation pmult} such that $tr\psuf f_{D_1,D_2}(tr)$. 
	
	If $\cnt_1\multeq\cnt_2$ but $\cnt_1 \neq \cnt_2$, let $tr\in \cnt_1$ but $tr\notin \cnt_2$. 
	Then from $\cnt_1\mult \cnt_2$, we have $tr\psuf f_{\cnt_1, \cnt_2}(tr)$. 
	Since $f_{\cnt_1, \cnt_2}(tr)$ is the replaced element, so $f_{\cnt_1, \cnt_2}(tr)\notin \cnt_1$. 
	By $\cnt_2\mult\cnt_1$, we have $f_{\cnt_1, \cnt_2}(tr)\psuf f_{\cnt_2, \cnt_1}(f_{\cnt_1, \cnt_2}(tr))$. 
	Continuing this process, we can construct an infinite descent sequence
	$tr\psuf f_{\cnt_1, \cnt_2}(tr)\psuf f_{\cnt_2, \cnt_1}(f_{\cnt_1, \cnt_2}(tr))\psuf ...$ w.r.t. relation $\suf$, which violates its well-foundedness. 
	So $\cnt_1 = \cnt_2$.

	\textbf{Proof of 2}. 
	\fi
	
\end{proof}

{\parindent 0pt
	\textbf{Content of Lemma~\ref{lemma:infinite descent sequence}:}  
	In a cyclic preproof (where there is at least one derivation path), 
	let $(\tau, \tau')$ be a pair of a derivation trace over a node pair $(\nu, \nu')$ of an invalid derivation path, where $\tau\in \{\sigma : [\alpha]\phi, \sigma : \la\alpha\ra\phi\}$ and $\nu\dddef (\Gamma\Rightarrow \tau, \Delta)$. 
	For any counter-example $\CT(\rho, \tau, \nu)$ of $\tau$, 
	there exists an evaluation $\rho'$ and a counter-example $\CT(\rho', \tau', \nu')$ of $\tau'$ such that $\CT(\rho', \tau', \nu')\mult \CT(\rho, \tau, \nu)$. Moreover, if $(\tau, \tau')$ is a progressive step, then $\CT(\rho', \tau', \nu')\pmult \CT(\rho, \tau, \nu)$.  
}

\begin{proof}[Proof of Lemma~\ref{lemma:infinite descent sequence}]
	Consider the rule application from node $\nu$, actually the only non-trivial cases are when it is an instance of rule $([\alpha])$ (namely ``case 1'') , rule $(\la\alpha\ra)$ (namely ``case 2'') or a substitution rule $(\sigma\Sub)$ (namely ``case 3'').

	\textbf{Case 1: }
	If from node $\nu$ rule $([\alpha])$ is applied with $\tau\dddef \sigma : [\alpha]\phi$ the target formula, 
	let $\tau' = (\sigma' : [\alpha']\phi)$ for some $(\alpha', \sigma')$, so $\nu' = (\Gamma\Rightarrow \tau', \Delta)$. 
	For any $\rho$ that violates $\mfr{P}(\nu)$, 
	$\CT(\rho, \tau, \nu)$ is obviously non-empty. 
	So $\CT(\rho, \tau', \nu')$ is also non-empty (since $\rho\models \Gamma$ and $\mfr{P}(\nu')$ is invalid). 
	By Definition~\ref{def:Program Properties}-\ref{item:Termination Finiteness}, $\CT(\rho, \tau, \nu)$ is also finite.  
	By the derivation $\vdash_{P_{\Oper,\Terminate}} (\Gamma\Rightarrow (\alpha, \sigma)\trans(\alpha', \sigma'),\Delta)$, 
	for each path $tr = \rho(\alpha',\sigma')...(\ter, \sigma_n)\in \CT(\rho, \tau', \nu')$ ($n\ge 1$), 
	$tr' = \rho(\alpha, \sigma)tr$ is a path in $\CT(\rho, \tau, \nu)$ that has $tr$ as its proper suffix.
	So we obtain the finiteness of $\CT(\rho, \tau', \nu')$ from the finiteness of $\CT(\rho, \tau, \nu)$, and thus we have 
	$\CT(\rho, \tau', \nu')\pmult \CT(\rho, \tau, \nu)$.

	\textbf{Case 2: }
	If from node $\nu$ rule $(\la\alpha\ra)$ is applied with $\tau\dddef \sigma : \la\alpha\ra\phi$ the target formula, 
	let $\tau' = (\sigma' : \la\alpha'\ra\phi)$ for some $(\alpha', \sigma')$, so $\nu' = (\Gamma\Rightarrow \tau', \Delta)$. 
	For any $\rho$ that violates $\mfr{P}(\nu)$, $\rho\models \Gamma$, 
	so both $\CT(\rho, \tau, \nu)$ and $\CT(\rho, \tau', \nu')$ are well defined. 
	By Definition~\ref{def:Program Properties}-\ref{item:Termination Finiteness}, $\CT(\rho, \tau, \nu)$ is finite.
	According to the derivation $\vdash_{P_{\Oper,\Terminate}} (\Gamma\Rightarrow (\alpha, \sigma)\trans(\alpha', \sigma'),\Delta)$, 
	for each path $tr = \rho(\alpha', \sigma')...(\ter, \sigma_n)\in \CT(\rho, \tau', \nu')$ ($n\ge 1$), 
	$tr' = \rho(\alpha,\sigma)tr$ is a path in $\CT(\rho, \tau, \nu)$ that has $tr$ as its proper suffix. 
	From the finiteness of $\CT(\rho, \tau, \nu)$, we have the finiteness of $\CT(\rho, \tau', \nu')$.
	Now consider two cases: 
	(1) If $\CT(\rho, \tau, \nu)$ is empty, then by the above $\CT(\rho, \tau', \nu')$ is also empty. So $\CT(\rho, \tau', \nu') = \CT(\rho, \tau, \nu)$;
	(2) If $\CT(\rho, \tau', \nu')$ is not empty, especially, when $(\tau, \tau')$ is a progressive step with the support of the derivation $\vdash_{P_{\Oper, \Terminate}} (\Gamma\Rightarrow (\alpha', \sigma')\termi, \Delta))$, 
	then from the above, 
	$\CT(\rho, \tau, \nu)$ must also be non-empty. 
	Moreover, in this case, we have $\CT(\rho, \tau', \nu')\pmult \CT(\rho, \tau, \nu)$. 
	
	\textbf{Case 3:}
	If from node $\nu$ a substitution rule $(\sigma \Sub)$ is applied, 
	let $\tau = \Sub(\sigma : [\alpha]\phi)$ be the target formula of $\nu$, 
	then $\tau' = \sigma : [\alpha]\phi$ and $\nu' = (\Gamma\Rightarrow \tau', \Delta)$. 
	The situation for $\tau\equiv \sigma: \la\alpha\ra\phi$ is similar and we omit it here. 
	For any evaluation $\rho$ that violates $\mfr{P}(\nu)$, 
	by Definition~\ref{def:Abstract Substitution} there exists a $\rho'\in \Eval$ such that for any formula $\psi\in \DLF$, 
	$\rho(\Sub(\psi))\equiv \rho'(\psi)$. 
	Especially, by the structureness of $\Sub$, there exist $\alpha_0, \sigma_0$ and $\phi_0$ such that 
	$\rho(\Sub(\sigma : [\alpha]\phi))\equiv \sigma_0 : [\alpha_0]\phi_0\equiv \rho'(\sigma : [\alpha]\phi)$. 
	Therefore, it is not hard to see that $\rho'$ violates $\mfr{P}(\nu')$ and 
	we have $\CT(\rho', \tau', \nu') = \CT(\rho, \tau, \nu)$.

	
	
	\ifx
	\textbf{Case 3:}
	If from node $\nu$ rule $(\sigma \textit{Sub})$ is applied, 
	let $\tau = (\sigma : [\alpha]\phi)[t/x]$ be the target formula of $\nu$, then $\tau' = \sigma : [\alpha]\phi$ and $\nu' = (\Gamma\Rightarrow \tau', \Delta)$. 
	The situation for $\tau\equiv \sigma: \la\alpha\ra\phi$ is similar and we omit it here. 
	For any evaluation $\rho$ that violates $\mfr{P}(\nu)$, 
	let $\rho'\dddef \rho[x\mapsto \rho(t)]$. 
	By the definition of substitutions in Section~\ref{section:Terms and Substitutions}, we can get that 
	for any term $u\in \TA$, 
	$\rho(u[t/x]) \equiv \rho|_{\FV(u)\setminus\FV_x(u)}(u)[\rho(t)/x] \equiv \rho'^*|_{\FV(u)\setminus\FV_x(u)}(u)[\rho(t)/x] \equiv \rho'^*(u)$. 
	Therefore, it is not hard to see that 
	$\rho'$ violates $\mfr{P}(\nu')$ and 
	we have $\CT(\rho', \tau', \nu') = \CT(\rho, \tau, \nu)$. 
	\fi
\end{proof}

From the ``case 2'' of the above proof, 
we see that derivation $\vdash_{P_{\Oper, \Terminate}} (\Gamma\Rightarrow (\alpha', \sigma')\termi, \Delta))$ plays the crucial role for proving the non-emptiness of the counter-examples in order to obtain a strict relation $\pmult$. 

\ifx
{\parindent 0pt
	\textbf{Content of Proposition~\ref{prop:lifting process 1}:}  
	Given any formulas $\phi, \psi$, if formula $\phi\to \psi$ is valid, 
	then rule 
	$$
	\begin{aligned}
		\infer[]
		{\Gamma\Rightarrow \sigma : \psi, \Delta}
		{\Gamma\Rightarrow \sigma : \phi, \Delta}
	\end{aligned}
	$$
	is sound for any configuration $\sigma\in \Conf_\std(\Gamma\cup \Delta\cup\{\phi, \psi\})$. 
}

\begin{proof}[Proof of Proposition~\ref{prop:lifting process 1}]
	Let $A = \Gamma\cup \Delta\cup\{\phi, \psi\}$. 
	Assume proposition $\mfr{P}(\Gamma\Rightarrow \sigma : \phi, \Delta)$ holds. 
	For any evaluation $\rho\in \Eval$, if $\rho\models \Gamma$ and $\rho\not\models \phi$ for all $\phi\in \Delta$, we prove $\rho\models \sigma : \psi$. 
	From the assumption we obtain $\rho\models \sigma : \phi$, which, by Definition~\ref{def:Semantics of Labelled LDL Formulas}, is $\rho, \sigma\models \phi$. 
	Because $\sigma$ is standard w.r.t. $A$, by Definition~\ref{def:Simple Configurations},  there exists a $\rho'\in \Eval$ such that $(\rho, \sigma)\se_A \rho'$, so  
	$\rho'\models \phi$, hence $\rho'\models \psi$ by the validity of $\phi\to \psi$. 
	So $\rho, \sigma\models \psi$, which is $\rho\models \sigma : \psi$. 
\end{proof}
\fi

{\parindent 0pt
	\textbf{Content of Proposition~\ref{prop:lifting process 2}:} 
	Given a sound rule of the form
	$$
	\begin{aligned}
		\infer[]
		{\Gamma\Rightarrow \Delta}
		{\Gamma_1\Rightarrow \Delta_1
			&...&
			\Gamma_n\Rightarrow \Delta_n}
	\end{aligned}, \mbox{ $n\ge 1$}, 
	$$
	in which all formulas are unlabelled, rule 
	$$
	\begin{aligned}
		\infer[]
		{\sigma : \Gamma\Rightarrow \sigma : \Delta}
		{\sigma : \Gamma_1\Rightarrow \sigma : \Delta_1
			&...&
			\sigma : \Gamma_n\Rightarrow \sigma : \Delta_n}
	\end{aligned}
	$$
	is sound for any configuration $\sigma\in \Conf_\free(\Gamma\cup \Delta\cup \Gamma_1\cup \Delta_1\cup...\cup\Gamma_n\cup\Delta_n)$. 
}

\begin{proof}[Proof of Proposition~\ref{prop:lifting process 2}]
	Let $A = \Gamma\cup \Delta\cup \Gamma_1\cup \Delta_1\cup...\cup\Gamma_n\cup\Delta_n$. 
	Assume formulas $\mfr{P}(\sigma : \Gamma_1\Rightarrow \sigma : \Delta_1)$,...,$\mfr{P}(\sigma : \Gamma_n\Rightarrow \sigma : \Delta_n)$ are valid, 
	we need to prove the validity of formula $\mfr{P}(\sigma : \Gamma\Rightarrow \sigma : \Delta)$. 
	First, notice that formula $\mfr{P}(\Gamma_i\Rightarrow \Delta_i)$ ($i\ge 1$) is valid, 
	because for any evaluation $\rho\in \Eval$, by Definition~\ref{def:free configurations}-\ref{item:free configurations cond 2} there exists a $\rho'\in \Eval$ satisfying that $(\rho', \sigma)\se_A \rho$. 
	So for any formula $\phi\in A$, $\rho\models \phi$ iff $\rho'\models \sigma : \phi$. 
	Therefore, formula $\mfr{P}(\Gamma\Rightarrow \Delta)$ is valid. 
	On the other hand, from that $\mfr{P}(\Gamma\Rightarrow \Delta)$ is valid, 
	we can get that $\mfr{P}(\sigma : \Gamma\Rightarrow \sigma : \Delta)$ is valid. 
	This is because for any evaluation $\rho\in \Eval$, by Definition~\ref{def:free configurations}-\ref{item:free configurations cond 1} there is a $\rho'\in \Eval$ such that $(\rho,\sigma)\se_A\rho'$, which means for any formula $\phi\in A$, $\rho\models\sigma : \phi$ iff $\rho'\models \phi$. 
\end{proof}

\section{Formal Definitions of \emph{While} Programs in $\LDL$}
\label{section:Formal Definitions of While Programs}

\subsection{Syntax  and Semantics of \emph{While} Programs}
We formally define program domain $\Domain{\WP}$ for \emph{While} programs. 
Below, given a set $A$, $\mcl{P}(A)$ represents its power set. 
Given a function $f : A \to B$, 
$f[x\mapsto v]$ represents the function that maps $x$ to value $v$, and maps the other variables to the same value as $f$.

\textbf{Expressions}.
The variable set $\Var_\WP$ ranges over the set of integer numbers $\mbb{Z}$. 

An expression $e$ is defined inductively as:
\begin{enumerate}[(1)]
	\item a variable $x\in \Var_\WP$, a constant $n\in \mbb{Z}$ are expressions;
	\item $e_1 \sim e_2$ where $\sim\in \{+, -, \times, /\}$ are expressions if $e_1$ and $e_2$ are expressions. 
\end{enumerate}
$+,-,\times,/$ are the usual arithmetic operators in domain $\mbb{Z}$. 
Denote the set of expressions as $\Exp_\WP$. 

\textbf{$\Prog_\WP$ \& $\Conf_\WP$ \& $\Fmla_\WP$}.
A formula $\phi\in \Fmla_\WP$ is an arithmetic first-order logical formula defined as follows:
$$
\phi\dddef e_1 \le e_2\ |\ \neg \phi\ |\ \phi\wedge \phi\ |\ \forall x.\phi. 
$$
Relation $\le$ is the usual ``less-than-equal relation'' in domain $\mbb{Z}$. 
Other logical operators such as $\to$, $\vee$, $\exists$ can be expressed by $\neg, \wedge$ and $\forall$ accordingly in the usual way. 
Other relations such as $=$ and $<$ can be expressed by $\le$ and logical operators accordingly. 

A program $\alpha\in \Prog_\WP$ is defined as follows in BNF forms: 
$$
\alpha\dddef x:=e\ |\ \alpha\seq \alpha\ |\ \Wif\ \phi\ \Wthen\ \alpha\ \Welse\ \alpha\ \Wend\ |\ \Wwhile\ \phi\ \Wdo\ \alpha\ \Wend, 
$$
where $\phi\in \Fmla_\WP$.  

A configuration $\sigma\in \Conf_\WP$ is defined as follows in BNF forms:
$$
\sigma \dddef x \mapsto e\ |\ x\mapsto e , \sigma, 
$$
where there is at most one $x\mapsto e$ for each variable $x\in \Var_\WP$. 
We use $\{\cdot\}$ to wrap a configuration $\sigma$ as: $\{\sigma \}$.  

Let $\Term_\WP\dddef \Exp_\WP\cup \Prog_\WP\cup \Conf_\WP\cup \Fmla_\WP$. 
With $\Prog_\WP, \Conf_\WP, \Fmla_\WP$ defined, we have $\LDL$ formulas $(\DLF)_\WP$ in program domain $\Domain{\WP}$ according to Definition~\ref{def:LDL Formulas}. 

\textbf{Free Variables}. 
For each term, its set of \emph{free variables} is captured as function 
$\FV : \Term_\WP\to \mcl{P}(\Var_\WP)$, which is inductively defined as follows:
\begin{enumerate}[(1)]
	\item $\FV(x)\dddef \{x\}$, $\FV(n)\dddef \emptyset$ where $n\in \mbb{Z}$;
	\item $\FV(e_1\sim e_2)\dddef \FV(e_1)\cup \FV(e_2)$, where $\sim\in \{+, -, \times, /\}$;
	\item $\FV(e_1\le e_2)\dddef \FV(e_1)\cup \FV(e_2)$;
	\item $\FV(\neg\phi)\dddef \FV(\phi)$;
	\item $\FV(\phi_1\wedge \phi_2)\dddef \FV(\phi_1)\cup \FV(\phi_2)$;
	\item $\FV(\forall x.\phi)\dddef \FV(\phi)\setminus\{x\}$;
	\item $\FV(x := e)\dddef \FV(e)\setminus \{x\}$;
	\item $\FV(\alpha\seq\beta)\dddef \FV(\alpha)\cup (\FV(\beta)\setminus\BD(\alpha))$;
	\item $\FV(\Wif\ \phi\ \Wthen\ \alpha\ \Welse\ \beta\ \Wend)\dddef \FV(\phi)\cup \FV(\alpha)\cup \FV(\beta)$;
	\item $\FV(\Wwhile\ \phi\ \Wdo\ \alpha\ \Wend)\dddef \FV(\phi)\cup \FV(\alpha)$;
	\item $\FV(x \mapsto e)\dddef \FV(e)$;
	\item $\FV(x\mapsto e, \sigma)\dddef \FV(x\mapsto e)\cup \FV(\sigma)$;
\end{enumerate}
where function $\BD: (\Prog_\WP\cup \Conf_\WP)\to \mcl{P}(\Var_\WP)$ returns the set of binding variables in programs and configurations, which is formally defined as:
\begin{enumerate}[(1)]
	\item $\BD(x := e)\dddef \{x\}$;
	\item $\BD(\alpha\seq \beta)\dddef \BD(\alpha)\cup \BD(\beta)$;
	\item $\BD(\Wif\ \phi\ \Wthen\ \alpha\ \Welse\ \beta\ \Wend)\dddef \BD(\alpha)\cup \BD(\beta)$;
	\item $\BD(\Wwhile\ \phi\ \Wdo\ \alpha\ \Wend)\dddef \BD(\alpha)$;
	\item $\BD(x\mapsto e)\dddef \{x\}$;
	\item $\BD(x\mapsto e, \sigma)\dddef \BD(x\mapsto e)\cup \BD(\sigma)$. 
\end{enumerate}
If $\FV(t)=\emptyset$ for a term $t\in \Term_\WP$, we call it a \emph{closed term}. 
For a set $A$ of terms, we use $\Clo(A)$ to denote its subset of all closed terms.

We extend function $\FV$ to program states and $\LDL$ formulas $(\DLF)_\WP$. 
For any program state $(\alpha, \sigma)\in \Prog_\WP\times \Conf_\WP$
or $\LDL$ formula $\phi\in (\DLF)_\WP$, $\FV((\alpha, \sigma))/\FV(\phi)$ is defined inductively as follows:

\ifx
, $\FV((\alpha, \sigma))$ is defined such that $$\FV((\alpha, \sigma))\dddef \FV(\sigma)\cup (\FV(\alpha)\setminus \BD(\sigma)).$$ 
For any $\LDL$ formula $\phi\in (\DLF)_\WP$, $\FV(\phi)$ is defined inductively as follows: 
\fi
\begin{enumerate}[(i)]
	\item $\FV((\alpha, \sigma))\dddef \FV(\sigma)\cup (\FV(\alpha)\setminus \BD(\sigma))$, if $(\alpha, \sigma)\in \Prog_\WP\times \Conf_\WP$;
	\item $\FV(F)$ is already defined, if $F\in \Fmla_\WP$;
	\item $\FV(\sigma : [\alpha]\phi)\dddef \FV(\sigma)\cup (\FV(\alpha)\setminus\BD(\sigma))\cup (\FV(\phi)\setminus\BD(\sigma))$;
	\item $\FV(\sigma : \phi)\dddef \FV(\sigma)\cup (\FV(\phi)\setminus \BD(\sigma))$, if $\phi$ is not in a form: $[\alpha]\psi$;
	\item $\FV(\neg\phi)\dddef \FV(\phi)$;
	\item $\FV(\phi_1\wedge \phi_2)\dddef \FV(\phi_1)\cup \FV(\phi_2)$. 
\end{enumerate}
A \emph{closed program state} $(\alpha, \sigma)\in \Prog_\WP\times \Conf_\WP$ is a program state satisfying that $\FV(\alpha, \sigma)=\emptyset$. 
A \emph{closed $\LDL$ formula} $\phi\in (\DLF)_\WP$ is a formula such that $\FV(\phi)=\emptyset$. 


\textbf{Substitutions}. 
Given a function $\eta: \Var_\WP\to \Exp_\WP$ that maps each variables to an expression,  
its extension $\eta^*_A: \Term_\WP\to \Term_\WP$ restricted on a set $A$ of variables, called a \emph{substitution}, maps each term $t\in\Term_\WP$ to a term by replacing each free variable $x\in A$ of $t$ with expression $\eta(x)$. 
$\eta^*_A$ is inductively defined as follows:
\begin{enumerate}[(1)]
	\item $\eta^*_A(x)\dddef \eta(x)$ if $x\in A$, $\eta^*_A(x)\dddef x$ otherwise;
	\item $\eta^*_A(n)\dddef n$ if $n\in \mbb{Z}$; 
	\item $\eta^*_A(e_1\sim e_2)\dddef \eta^*_{A}(e_1)\sim \eta^*_{A}(e_2)$, where $\sim\in \{+, -, \times, /\}$;
	\item $\eta^*_A(e_1\le e_2)\dddef \eta^*_{A}(e_1)\le \eta^*_{A}(e_2)$;
	\item $\eta^*_A(\neg\phi)\dddef \eta^*_{A}(\phi)$;
	\item $\eta^*_A(\phi_1\wedge \phi_2)\dddef \eta^*_{A}(\phi_1)\wedge \eta^*_{A}(\phi_2)$;
	\item $\eta^*_A(\forall x.\phi)\dddef \forall x.\eta^*_{A\setminus \{x\}}(\phi)$;
	\item $\eta^*_A(x := e)\dddef x := \eta^*_{A\setminus \{x\}}(e)$;
	\item $\eta^*_A(\alpha\seq\beta)\dddef \eta^*_{A}(\alpha)\seq \eta^*_{A\setminus \BD(\alpha)}(\beta)$;
	\item $\eta^*_A(\Wif\ \phi\ \Wthen\ \alpha\ \Welse\ \beta\ \Wend)\dddef \Wif\ \eta^*_{A}(\phi)\ \Wthen\ \eta^*_{A}(\alpha)\ \Welse\ \eta^*_{A}(\beta)\ \Wend$;
	\item $\eta^*_A(\Wwhile\ \phi\ \Wdo\ \alpha\ \Wend)\dddef \Wwhile\ \eta^*_{A}(\phi)\ \Wdo\ \eta^*_{A}(\alpha)\ \Wend$;
	\item $\eta^*_A(x\mapsto e)\dddef x\mapsto \eta^*_{A}(e)$;
	\item $\eta^*_A(x\mapsto e, \sigma)\dddef \eta^*_{A}(x\mapsto e), \eta^*_{A}(\sigma)$. 
	
	\ifx
	\item $\eta^*_A(x)\dddef \eta(x)$ if $x\in A$, $\eta^*_A(x)\dddef x$ otherwise;
	\item $\eta^*_A(n)\dddef n$ if $n\in \mbb{Z}$; 
	\item $\eta^*_A(e_1\sim e_2)\dddef \eta^*_{A\cap\FV(e_1)}(e_1)\sim \eta^*_{A\cap\FV(e_2)}(e_2)$, where $\sim\in \{+, -, \times, /\}$;
	\item $\eta^*_A(e_1\le e_2)\dddef \eta^*_{A\cap\FV(e_1)}(e_1)\le \eta^*_{A\cap\FV(e_2)}(e_2)$;
	\item $\eta^*_A(\neg\phi)\dddef \eta^*_{A\cap\FV(\phi)}(\phi)$;
	\item $\eta^*_A(\phi_1\wedge \phi_2)\dddef \eta^*_{A\cap\FV(\phi_1)}(\phi_1)\wedge \eta^*_{A\cap\FV(\phi_2)}(\phi_2)$;
	\item $\eta^*_A(\forall x.\phi)\dddef \forall x.\eta^*_{A\cap(\FV(\phi)\setminus \{x\})}(\phi)$;
	\item $\eta^*_A(x := e)\dddef x := \eta^*_{A\cap(\FV(e)\setminus \{x\})}(e)$;
	\item $\eta^*_A(\alpha\seq\beta)\dddef \eta^*_{A\cap\FV(\alpha)}(\alpha)\seq \eta^*_{A\cap(\FV(\beta)\setminus \BD(\alpha))}(\beta)$;
	\item $\eta^*_A(\Wif\ \phi\ \Wthen\ \alpha\ \Welse\ \beta\ \Wend)\dddef \Wif\ \eta^*_{A\cap\FV(\phi)}(\phi)\ \Wthen\ \eta^*_{A\cap\FV(\alpha)}(\alpha)\ \Welse\ \eta^*_{A\cap\FV(\beta)}(\beta)\ \Wend$;
	\item $\eta^*_A(\Wwhile\ \phi\ \Wdo\ \alpha\ \Wend)\dddef \Wwhile\ \eta^*_{A\cap\FV(\phi)}(\phi)\ \Wdo\ \eta^*_{A\cap\FV(\alpha)}(\alpha)\ \Wend$;
	\item $\eta^*_A(x\mapsto e)\dddef x\mapsto \eta^*_{A\cap(\FV(e)\setminus \{x\})}(e)$;
	\item $\eta^*_A(x\mapsto e, \sigma)\dddef \eta^*_{A\cap(\FV(x\mapsto e)\setminus \BD(x\mapsto e,\sigma))}(x\mapsto e), \eta^*_{A\cap(\FV(\sigma)\setminus \BD(x\mapsto e,\sigma))}(\sigma)$. 
	\fi
\end{enumerate}
For any term $t\in \Term_\WP$, 
we define $\eta^*(t)\dddef \eta^*_{\Term_\WP}(t)$. 
We often use $t[e/x]$ to denote a substitution $\eta^*_{\{x\}}(t)$, with $\eta(x) = e$. 

For a substitution $\eta^*(t)$ of a term $t\in \Term_\WP$, 
we always assume that it is \emph{admissible} in the usual sense that we guarantee that each free variable $y$ of a replacing term, say $\eta(x)$ for some variable $x$, is still free in $\eta^*(t)$ after the substitution by variable renaming if necessary. 


\ifx
Given $\Prog_\WP$, $\Conf_\WP$ and $\Fmla_\WP$, 
let $\Beha_\WP\dddef \{(\alpha, \sigma)\trans (\alpha', \sigma')\ |\ \alpha, \alpha'\in \Prog_\WP, \sigma, \sigma'\in \Prog_\WP\}\subseteq \Fmla_\WP$ and 
$\Termi_\WP\dddef \{(\alpha, \sigma)\termi\ |\ \alpha\in \Prog_\WP, \sigma\in \Conf_\WP\}\subseteq \Fmla_\WP$ be the set of program transitions and terminations. 
$\Clo(\Beha_\WP)$ and $\Clo(\Termi_\WP)$ can be defined according to Section~\ref{section:Program Behaviours}. 
\fi

We extend substitution $\eta^*_A$ to program states $\Prog_\WP\times \Conf_\WP$ and $\LDL$ formulas $(\DLF)_\WP$ as follows:
\begin{enumerate}[(i)]
	\item $\eta^*_A((\alpha, \sigma))\dddef (\eta^*_{A\setminus \BD(\sigma)}(\alpha), \eta^*_A(\sigma))$, if $(\alpha, \sigma)\in \Prog_\WP\times \Conf_\WP$;
	\item $\eta^*_A(F)$ is already defined, if $F\in \Fmla_\WP$;
	\item $\eta^*_A(\sigma : [\alpha]\phi)\dddef \eta^*_A(\sigma) : [\eta^*_{A\setminus \BD(\sigma)}(\alpha)]\eta^*_{A\setminus \BD(\sigma)}(\phi)$;
	\item $\eta^*_A(\sigma : \phi)\dddef \eta^*_A(\sigma) : \eta^*_{A\setminus \BD(\sigma)}(\phi)$, if $\phi$ is not in a form: $[\alpha]\psi$;
	\item $\eta^*_A(\neg\phi)\dddef \neg \eta^*_A(\phi)$;
	\item $\eta^*_A(\phi_1\wedge \phi_2)\dddef \eta^*_A(\phi_1)\wedge \eta^*_A(\phi_2)$.
\end{enumerate}

\textbf{$\Eval_\WP$ \& $\app_\WP$}. 
Let $s$ be a function $s : \Var_\WP\to \Clo(\Exp_\WP)$, then $s^*$ is an \emph{evaluation} of $\Eval_\WP$. 
Given a configuration $\sigma\in \Conf_\WP$, $\sigma(\cdot) : \Var_\WP\to 
\Exp_\WP$ is a function defined inductively as follows: 
\begin{enumerate}[(1)]
	\item $(x\mapsto e)(x)\dddef e$;
	\item $(x\mapsto e, \sigma)(y)\dddef (x\mapsto e)(y)$ if $y\equiv x$, 
	$(x\mapsto e, \sigma)(y)\dddef \sigma(y)$ otherwise. 
\end{enumerate}
For any $\sigma\in \Conf_\WP$ and $\phi\in \Fmla_\WP$, we define interpretation $\app_\WP(\sigma, \phi)\dddef \sigma^*(\phi)$. 

\ifx
\textbf{Equivalence $\cfeq$}. 
Given two configurations $\sigma_1, \sigma_2\in \Conf_\WP$, \emph{configuration equivalence} $\sigma_1\cfeq \sigma_2$ is defined such that $\sigma^*_1(t)\equiv \sigma^*_2(t)$ for any $t\in \Term_\WP$. 
\fi


\textbf{$\Prop_\WP$ \& $\Oper_\WP$ \& $\Terminate_\WP$}. 
Propositions $\Prop_\WP$ is defined as the set of all closed formulas of $\Fmla_\WP$. 
For a formula $\phi\in \Clo(\Fmla_\WP)$, $\boolsem_\WP(\phi) \dddef 1$ if formula $\phi$ is true in the theory of integer numbers.  
Table~\ref{table:Operational Semantics of While Programs} depicts the operational semantics of \emph{While} programs. 
For a program transition $(\alpha, \sigma)\trans(\alpha', \sigma')\in \Clo(\Fmla_\WP)$, $\boolsem_\WP((\alpha, \sigma)\trans(\alpha', \sigma'))\dddef 1$ if $(\alpha, \sigma)\trans(\alpha', \sigma')$ can be inferred according to Table~\ref{table:Operational Semantics of While Programs}. 
With $\Oper_\WP$ defined, we can obtain $\Terminate_\WP$ according to Section~\ref{section:Program Behaviours}. 

\begin{table}[tbhp]
	\begin{center}
		\noindent\makebox[\textwidth]{%
			\scalebox{0.9}{
				\begin{tabular}{c}
					\toprule
					$
					\infer[^{(o:x:=e)}]
					{(x:=e, \sigma)\trans(\ter, \sigma^x_{\sigma^*(e)})}
					{
					}
					$
					\ \
					$
					\infer[^{(o:;)}]
					{(\alpha_1 ; \alpha_2, \sigma)\trans(\alpha'_1 ; \alpha_2, \sigma')}
					{
						(\alpha_1, \sigma)\trans (\alpha'_1, \sigma')
						&
						\alpha'_1\not\equiv \ter
					}
					$
					\ \ 
					$
					\infer[^{(o:;\ter)}]
					{(\alpha_1; \alpha_2, \sigma)\trans(\alpha_2, \sigma')}
					{
						(\alpha_1, \sigma)\trans (\ter, \sigma')
					}
					$
					\\
					$
					\infer[^{(o:\textit{ite-1})}]
					{(\Wif\ \phi\ \Wthen\ \alpha_1\ \Welse\ \alpha_2\ \Wend, \sigma)\trans(\alpha'_1, \sigma')}
					{
						(\alpha_1, \sigma)\trans (\alpha'_1, \sigma')
						&
						\phi\mbox{ is true}
					}
					$
					\ \ 
					$
					\infer[^{(o:\textit{ite-2})}]
					{(\Wif\ \phi\ \Wthen\ \alpha_1\ \Welse\ \alpha_2\ \Wend, \sigma)\trans(\alpha'_2, \sigma')}
					{
						(\alpha_2, \sigma)\trans (\alpha'_2, \sigma')
						&
						\phi\mbox{ is false}
					}
					$
					\\
					$
					\infer[^{(o:\textit{wh-1})}]
					{(\textit{while}\ \phi\ \textit{do}\ \alpha\ \textit{end}, \sigma)\trans(\alpha'\ ;\ \textit{while}\ \phi\ \textit{do}\ \alpha\ \textit{end}, \sigma')}
					{
						(\alpha, \sigma)\trans(\alpha', \sigma')
						&
						\sigma^*(\phi)\mbox{ is true}
					}
					$
					\ \ 
					$
					\infer[^{(o:\textit{wh-1}\ter)}]
					{(\textit{while}\ \phi\ \textit{do}\ \alpha\ \textit{end}, \sigma)\trans(\textit{while}\ \phi\ \textit{do}\ \alpha\ \textit{end}, \sigma')}
					{
						(\alpha, \sigma)\trans(\ter, \sigma')
						&
						\sigma^*(\phi)\mbox{ is true}
					}
					$
					\\
					$
					\infer[^{(o:\textit{wh-2})}]
					{(\textit{while}\ \phi\ \textit{do}\ \alpha\ \textit{end}, \sigma)\trans(\ter, \sigma)}
					{
						\sigma^*(\phi)\mbox{ is false}
					}
					$
					\\
					\bottomrule
				\end{tabular}
			}
		}
	\end{center}
	\caption{Operational Semantics of \emph{While} Programs}
	\label{table:Operational Semantics of While Programs}
\end{table}

\subsection{Proof System $P_{\Oper_\WP, \Terminate_\WP}$}

\textbf{Proof System $P_{\Oper_\WP, \Terminate_\WP}$}. 
Table~\ref{table:Inference Rules for Program Behaviours of While programs} and~\ref{table:Inference Rules for Program Terminations of While programs} list the inference rules for program behaviours $\Oper_\WP$ and program terminations $\Terminate_\WP$ respectively. 
Here we omit the other rules in first-order logic (cf.~\cite{Harel00}, some of them are already shown in Table~\ref{table:General Rules for LDL}) that are sometimes necessary when making derivations in $P_{\Oper_\WP, \Terminate_\WP}$.

\begin{table}[tb]
	\begin{center}
		\noindent\makebox[\textwidth]{%
			\scalebox{0.9}{
				\begin{tabular}{c}
					\toprule
					$
					\infer[^{(x:=e)}]
					{\Gamma\Rightarrow (x:=e, \sigma)\trans(\ter, \sigma^x_{\sigma^*(e)}), \Delta}
					{
					}
					$
					\ \
					$
					\infer[^{(;)}]
					{\Gamma\Rightarrow (\alpha_1 ; \alpha_2, \sigma)\trans(\alpha'_1 ; \alpha_2, \sigma'), \Delta}
					{
						\Gamma\Rightarrow (\alpha_1, \sigma)\trans (\alpha'_1, \sigma'), \Delta
					}
					$
					\ \ 
					$
					\infer[^{(;\ter)}]
					{\Gamma\Rightarrow(\alpha_1; \alpha_2, \sigma)\trans(\alpha_2, \sigma'), \Delta}
					{
						\Gamma\Rightarrow (\alpha_1, \sigma)\trans (\ter, \sigma'), \Delta
					}
					$
					\\
					$
					\infer[^{(\textit{ite-1})}]
					{\Gamma\Rightarrow (\Wif\ \phi\ \Wthen\ \alpha_1\ \Welse\ \alpha_2\ \Wend, \sigma)\trans(\alpha'_1, \sigma'), \Delta}
					{
						\Gamma\Rightarrow (\alpha_1, \sigma)\trans (\alpha'_1, \sigma'), \Delta
						&
						\Gamma\Rightarrow \app(\sigma, \phi), \Delta
					}
					$
					\ \ 
					$
					\infer[^{(\textit{ite-2})}]
					{\Gamma\Rightarrow (\Wif\ \phi\ \Wthen\ \alpha_1\ \Welse\ \alpha_2\ \Wend, \sigma)\trans(\alpha'_2, \sigma'), \Delta}
					{
						\Gamma\Rightarrow (\alpha_2, \sigma)\trans (\alpha'_2, \sigma'), \Delta
						&
						\Gamma\Rightarrow \neg\app(\sigma, \phi), \Delta
					}
					$
					\\
					$
					\infer[^{(\textit{wh1})}]
					{\Gamma\Rightarrow (\textit{while}\ \phi\ \textit{do}\ \alpha\ \textit{end}, \sigma)\trans(\alpha';\ \textit{while}\ \phi\ \textit{do}\ \alpha\ \textit{end}, \sigma'), \Delta}
					{
						\Gamma, \app(\sigma, \phi)\Rightarrow (\alpha, \sigma)\trans(\alpha', \sigma'), \Delta
						&
						\Gamma\Rightarrow \app(\sigma, \phi), \Delta
					}
					$
					\\
					$
					\infer[^{(\textit{wh1}\ter)}]
					{\Gamma\Rightarrow (\textit{while}\ \phi\ \textit{do}\ \alpha\ \textit{end}, \sigma)\trans(\textit{while}\ \phi\ \textit{do}\ \alpha\ \textit{end}, \sigma'), \Delta}
					{
						\Gamma, \app(\sigma, \phi)\Rightarrow (\alpha, \sigma)\trans(\ter, \sigma'), \Delta
						&
						\Gamma\Rightarrow \app(\sigma, \phi), \Delta
					}
					$
					\ \ 
					$
					\infer[^{(\textit{wh2})}]
					{\Gamma\Rightarrow (\textit{while}\ \phi\ \textit{do}\ \alpha\ \textit{end}, \sigma)\trans(\ter, \sigma), \Delta}
					{
						\Gamma\Rightarrow \neg \app(\sigma, \phi), \Delta
					}
					$
					\\
					\bottomrule
				\end{tabular}
			}
		}
	\end{center}
	\caption{Inference Rules for Program Behaviours of \emph{While} programs}
	\label{table:Inference Rules for Program Behaviours of While programs}
\end{table}

The rules for program behaviours in Table~\ref{table:Inference Rules for Program Behaviours of While programs} are directly from the operational semantics of \emph{While} programs (Table~\ref{table:Operational Semantics of While Programs}). 
Their soundness and completeness w.r.t. $\Oper_\WP$ are trivial. 
Table~\ref{table:Inference Rules for Program Terminations of While programs}, \ref{table:Another Type of Inference Rules for Program Terminations of While programs} list two types of inference rules for program terminations. 
In Table~\ref{table:Inference Rules for Program Terminations of While programs}, 
termination $(\alpha, \sigma)\termi \sigma'$ means terminating with configuration $\sigma'$ as the result. In other words, there exists a path $(\alpha, \sigma)...(\ter, \sigma')$ over $\Oper_\WP$. 
In rule $(\termi \textit{wh} 1)$, where $I$ is an invariant of the loop body $\alpha$; 
$x$ is a termination factor of the while statement indicating its terminations. 
The soundness and completeness of the rules in Table~\ref{table:Inference Rules for Program Terminations of While programs} w.r.t. $\Terminate_\WP$ can be proved based on the boolean semantics $\boolsem_\WP$ in an inductively way according to the syntax of \emph{While} programs. 

\begin{table}[tb]
	\begin{center}
		\noindent\makebox[\textwidth]{%
			\scalebox{0.9}{
				\begin{tabular}{c}
					\toprule
					$
					\infer[^{(\termi x:=e)}]
					{\Gamma\Rightarrow (x:=e, \sigma)\termi \sigma^x_{\sigma^*(e)}, \Delta}
					{
					}
					$
					\ \
					$
					\infer[^{(\termi ;)}]
					{\Gamma\Rightarrow (\alpha_1 ; \alpha_2, \sigma)\termi \sigma', \Delta}
					{
						\Gamma\Rightarrow (\alpha_1, \sigma)\termi \sigma'', \Delta
						&
						\Gamma\Rightarrow (\alpha_2, \sigma'')\termi \sigma', \Delta
					}
					$
					\\
					$
					\infer[^{(\termi \textit{ite-1})}]
					{\Gamma\Rightarrow (\Wif\ \phi\ \Wthen\ \alpha_1\ \Welse\ \alpha_2\ \Wend, \sigma)\termi\sigma', \Delta}
					{
						\Gamma\Rightarrow (\alpha_1, \sigma)\termi \sigma', \Delta
						&
						\Gamma\Rightarrow \sigma^*(\phi), \Delta
					}
					$
					\ \ 
					$
					\infer[^{(\termi \textit{ite-2})}]
					{\Gamma\Rightarrow (\Wif\ \phi\ \Wthen\ \alpha_1\ \Welse\ \alpha_2\ \Wend, \sigma)\termi \sigma', \Delta}
					{
						\Gamma\Rightarrow (\alpha_2, \sigma)\termi\sigma', \Delta
						&
						\Gamma\Rightarrow \sigma^*(\neg\phi), \Delta
					}
					$
					\\
					$
					\infer[^{(\termi \textit{wh1})}]
					{\Gamma\Rightarrow (\textit{while}\ \phi\ \textit{do}\ \alpha\ \textit{end}, \sigma)\termi\sigma', \Delta}
					{
						\Gamma\Rightarrow \sigma^*(x > 0 \wedge I\wedge \phi), \Delta
						&
						\sigma^*(x = n\wedge I)\Rightarrow (\alpha, \sigma)\termi \sigma''\wedge \sigma''^*(x < n\wedge I)
						&
						\sigma'^*(x\le 0\wedge I)\Rightarrow \sigma'^*(\neg\phi)
					}
					$
					\\
					$
					\infer[^{(\termi \textit{wh2})}]
					{\Gamma\Rightarrow (\textit{while}\ \phi\ \textit{do}\ \alpha\ \textit{end}, \sigma)\termi\sigma, \Delta}
					{
						\Gamma\Rightarrow \sigma^*(\neg\phi), \Delta
					}
					$
					\ \ 
					$
					\infer[^{(\termi)}]
					{\Gamma\Rightarrow (\alpha, \sigma)\termi, \Delta}
					{
						\Gamma\Rightarrow (\alpha, \sigma)\termi\sigma', \Delta
					}
					$
					\\
					\bottomrule
				\end{tabular}
			}
		}
	\end{center}
	\caption{Inference Rules for Program Terminations of \emph{While} programs}
	\label{table:Inference Rules for Program Terminations of While programs}
\end{table}

The soundness of rule $(\Sub)$ in \emph{While}-programs domain is instantiated by choosing the substitution $(\cdot)[e/x]$ ($e\in \Exp_\WP$): 
$$
\begin{aligned}
	\infer[^{(\Sub)}]
	{\Gamma[e/x]\Rightarrow \Delta[e/x]}
	{\Gamma\Rightarrow\Delta}
\end{aligned}.
$$
To see its soundness, we only need to prove that function $(\cdot)[e/x]$ is just a `substitution' defined in Definition~\ref{def:Abstract Substitution}. 
We observe that for any evaluation $\sigma^*\in \Eval_\WP$, let $\rho\dddef \sigma^*[x\mapsto \sigma^*(e)]$. 
Then for any formula $\phi\in (\DLF)_\WP$, by the definition of the substitutions above, we can have 
$\sigma^*(\phi[e/x])\equiv \sigma^*_{\Var_\WP\setminus \{x\}}(\phi)[\sigma^*(e)/x]\equiv \rho_{\Var_\WP\setminus \{x\}}(\phi)[\sigma^*(e)/x]\equiv \rho(\phi)$. 

\ifx
\textbf{Extra Rules $(\Extra: \sigma\cfeq)$, $(\SExtra:\sigma\textit{Sub})$}. 
In the derivations of Section~\ref{section:Example One: A While Program}, we need the following extra rules:
$$
\begin{aligned}
	\infer[^{(\Extra:\sigma\cfeq)}]
	{\Gamma\Rightarrow \sigma : [\alpha]\phi, \Delta}
	{\Gamma\Rightarrow \sigma' : [\alpha]\phi, \Delta}
\end{aligned}, \mbox{if $\sigma\cfeq \sigma'$}, 
$$
and 
$$
\begin{aligned}
	\infer[^{(\SExtra : \sigma\textit{Sub})}]
	{\Gamma[e/x]\Rightarrow \Delta[e/x]}
	{\Gamma\Rightarrow\Delta}
\end{aligned}. 
$$
The soundness of rule $(\Extra:\sigma\cfeq)$ is according to the definition of $\cfeq$, and the fact that when $\sigma\cfeq \sigma'$, their behaviours are actually the same given any program based on the operational semantics of \emph{While} programs (Table~\ref{table:Operational Semantics of While Programs}). 
The soundness of rule $(\SExtra : \sigma\textit{Sub})$ can be proved directly according to the definition of substitution above. 

We prove the safeness of rule $(\SExtra : \sigma\textit{Sub})$ as follows:
\begin{proposition}[Safeness of Rule $(\SExtra : \sigma\textit{Sub})$]
	Let $(\tau, \tau')$, where $\tau\dddef (\sigma : [\alpha]\phi)[e/x]$ (resp. $\tau\dddef (\sigma : \la\alpha\ra\phi)[e/x]$), $\tau'\dddef \sigma : [\alpha]\phi$ (resp. $\tau'\dddef \sigma : \la\alpha\ra)$), be a CP pair appearing on the right side of rule $(\SExtra : \sigma\textit{Sub})$, then $(\tau, \tau')$ is critical and thus rule $(\SExtra : \sigma\textit{Sub})$ is safe. 
\end{proposition}

\begin{proof}
	We only consider the case for modality $[\cdot]$, the case for modality $\la\cdot\ra$ is similar. 
	It is sufficient to prove that for any term $t\in \Term_\WP$ and evaluation $\rho_\WP\in \Eval_\WP$, there is an evaluation $\rho'_\WP\in \Eval_\WP$ such that 
	$\rho_\WP(t[e/x])\equiv \rho'_\WP(t)$. 
	Let $\rho'_\WP\dddef \rho_\WP[x\mapsto \rho_\WP(e)]$. 
	By the definition of substitution above, 
	we can actually have $\rho_\WP(t[e/x]) \equiv (\rho_\WP|_{\dom(\rho)\setminus \{x\}}(t))[\rho_\WP(e)/x]\equiv (\rho'_\WP|_{\dom(\rho')\setminus \{x\}}(t))[\rho_\WP(e)/x]\equiv \rho'_\WP(t)$. 
\end{proof}
\fi

\begin{table}[tb]
	\begin{center}
		\noindent\makebox[\textwidth]{%
			\scalebox{0.9}{
				\begin{tabular}{c}
					\toprule
					$
					\begin{aligned}
						\infer[^{(\termi \trans)}]
						{\Gamma\Rightarrow (\alpha, \sigma\{e\})\termi, \Delta}
						{\Gamma\Rightarrow (\alpha', \sigma'\{e'\})\termi, \Delta
							&
							\Gamma\Rightarrow (\alpha, \sigma)\trans(\alpha', \sigma'), \Delta
							&
							\Gamma\Rightarrow e' \ge 0, \Delta
							&
							\Gamma\Rightarrow e'\le e, \Delta
						}
					\end{aligned}
					$
					\\
					$
					\begin{aligned}
						\infer[^{(\Sub)}]
						{\Gamma[e/x]\Rightarrow \Delta[e/x]}
						{\Gamma\Rightarrow\Delta}
					\end{aligned}
					$
					\\
					\bottomrule
				\end{tabular}
			}
		}
	\end{center}
	\caption{Another Type of Inference Rules for Program Terminations of \emph{While} programs}
	\label{table:Another Type of Inference Rules for Program Terminations of While programs}
\end{table}

\textbf{An Alternative Cyclic Proof System for Program Terminations}.
The rules in Table~\ref{table:Inference Rules for Program Terminations of While programs} actually rely on the syntactic structures of \emph{While} programs. 
Table~\ref{table:Another Type of Inference Rules for Program Terminations of While programs} gives another possible set of rules for reasoning about program terminations in a cyclic approach like $\pfDLp$, which also relies on rule $(\Sub)$ and other rules in first-order logic. 
In rule $(\termi\trans)$, 
given a configuration $\sigma\in \Conf_\WP$, $\sigma\{e\}$ expresses that an expression $e$ called the \emph{termination factor} of $\sigma$ appears in $\sigma$, which indicates how far a program can terminate. 
Rule $(\Sub)$ helps constructing a configuration with a suitable form in order to derive a bud for each potential infinite proof branch. 
The soundness conditions of the cyclic proof system based on Table~\ref{table:Another Type of Inference Rules for Program Terminations of While programs} can be defined similarly as in Section~\ref{section:Construction of A Cyclic Preproof Structure}, where a progressive step of a derivation trace is a pair $((\alpha, \sigma\{e\})\termi, (\alpha', \sigma'\{e'\})\termi)$ of an instance of rule $(\termi \trans)$, but with the premise `$\Gamma\Rightarrow e' < e, \Delta$'.  

The soundness of rule $(\termi \trans)$ is obvious. 
Here we skip the detailed proof of the soundness of the cyclic proof system based on Table~\ref{table:Another Type of Inference Rules for Program Terminations of While programs} but only give a rough idea, which makes use of the well-foundedness of the less-than relation $<$ between integer numbers. 
Intuitively, starting from a state $(\alpha_1, \sigma_1\{e_1\})$, if $(\alpha_1, \sigma_1)\not\termi$, then by the definition of `cyclic preproof' (Defnition~\ref{def:Cyclic Preproof}) we have an infinite sequence of relations: $e_1\ge e_2\ge...\ge e_n\ge ...$ with expressions $e_1,...,e_n\ge 0$. And by the definition of progressive step, among the sequence there is an infinite number of relation `$>$', thus violating the well-foundedness of relation $<$.  
We conjecture that such a cyclic proof system is also complete. More analysis will be given in future work.

\ifx
The soundness of rule $(\termi \trans)$ is obvious. 
The soundness of the cyclic proof system can be proved according to the well-foundedness of relation $<$ in the set of natural numbers:  
Intuitively, starting from a state $(\alpha_1, \sigma_1\{e_1\})$, if $(\alpha_1, \sigma_1)\not\termi$, then by the definition of `cyclic preproof' (Defnition~\ref{def:Cyclic Preproof}) we have an infinite sequence of relations: $e_1\ge e_2\ge...\ge e_n\ge ...$ with expressions $e_1,...,e_n\ge 0$. And by the definition of progressive step, among the sequence there is an infinite number of relation `$>$', thus violating the well-foundedness of relation $<$.  
We conjecture that such a cyclic proof system is also complete. More analysis will be given in future work. 
\fi

\section{Cyclic Deduction of A Synchronous Program}
\label{section:Example Two: A Synchronous Loop Program}

This example shows that we can use $\LDL$ to verify a program whose ``loop structures'' are implicit. That is where $\LDL$ is really useful as the loop information can be tracked in the cyclic proof structures of $\LDL$.  
Consider the synchronous program $E\in \Prog_\E$ (Table~\ref{table:An Example of Program Structures}) written in Esterel~\cite{Berry92}:
$$E \dddef \Etrap\ A \parallel B\ \Eend, $$
where
$
A \dddef
\Eloop\ (\Eemit\ S(0)\ ;\ x := x - S\ ;\ \Eif\ x = 0\ \Ethen\ \Eexit\ \Eend\ ;\ \Epause)\ \Eend
$, 
$
B \dddef
\Eloop\ (\Eemit\ S(1)\ ;\ \Epause)\ \Eend
$
are two programs running in parallel.

Different from \textit{while} programs, the behaviour of a synchronous program is characterized by a sequence of \emph{instances}.
At each instance, several (atomic) executions of a program may occur. 
The value of each variable is unique in an instance. 
When several programs run in parallel, 
their executions at one instance are thought to occur simultaneously. 
In this manner, the behaviour of a parallel synchronous program is deterministic. 

\ifx
A synchronous program model guarantees that when several programs run in parallel, 
the status of each variable in the programs at the end of an instance is unique,
regardless of the orders of all executions during that instance.
Therefore, at each instance, the behaviour of a parallel synchronous program is deterministic.
\fi

In this example, the key word $\Epause$ marks the end of an instance. 
At the beginning of each instance, the values of all signals are set to $\bot$, representing `absent' state. 
The $\Eloop$ statements in programs $A$ and $B$ are executed repeatedly for an infinite number of times until some $\Eexit$ statement is encountered. 
At each instance, program $A$ firstly emits signal $S$ with value $0$ and substracts the local variable $x$ with the current value of $S$; 
then checks if $x = 0$. 
While program $B$ emits signal $S$ with value $1$. 
The value of signal $S$ in one instance should be the sum of all of its emitted values by different current programs. 
So the value of $S$ should be $1 + 0 = 1$. 
The whole program $E$ continues executing until $x = 0$ is satisfied, when $\Eexit$ terminates the whole program by jumping out of the $\Etrap$ statement.  

As an instance, 
when initializing $x$ as value $3$, 
the values of all variables at the end of each instance are listed as follows:
\begin{itemize}
	\item instance 1: $(x = 3, S = 1)$;
	\item instance 2: $(x = 2, S = 1)$;
	\item instance 3: $(x = 1, S = 1)$;
	\item instance 4: $(x = 0, S = 1)$.
\end{itemize}

\ifx
A configuration in Esterel is a stack~\cite{Butucaru10}, allowing expressing local variables with same names.
For example, $\{x\mapsto 5 \Split S\mapsto 1 \Split S \mapsto 2\}$ represents a configuration in which there are 3 variables: $x$, and two $S$s with different values.
We use ``$\Split$'' instead of ``$,$'' to remind that it is a stack, with the right-most side as the top of the stack. 
$\app(\sigma)$ has a similar interpretation as in Example 1, 
but because of the stack structure of $\sigma$, for a variable, say $x$, its value is only 
determined by the value of the top-most variable $x$ of $\sigma$. 
For instance, $\app(\{x\mapsto 5\Split S\mapsto 1\Split S\mapsto 2\})$ maps 
$S$ to $2$ rather than $1$. 
\fi

In a parallel Esterel program, the executions of concurrent programs are usually dependent on each other to maintain the consistency of simultaneous executions, imposing special orders of atomic executions in one instance. 
In the program $E$ above, for instance, 
the assignment $x := x - S$ can only be executed after all values of $S$ in both programs $A$ and $B$ are collected.
In other words, the assignment $x := x - S$ can only be executed after $\Eemit\ S(0)$ and $\Eemit\ S(1)$. 
This characteristic of Esterel programs makes direct compositional reasoning impossible because 
the executions between parallel programs in an instance are not true interleaving. 
One has to symbolically execute the parallel program as a whole in order to obtain the right execution orders in each instance. 

In the following we prove the property 
$$
\nu_2 \dddef (v > 0)\Rightarrow \{x\mapsto v\Split S\mapsto \bot\} : \la E\ra\true, 
$$
which says that under configuration $\{x\mapsto v, S\mapsto\bot\}$ (with $v$ a free variable), if $v > 0$, then $E$ can finally terminate (satisfying $\true$). 

\ifx
Parallel Esterel programs are usually not compositional, in the sense that the executions of one concurrent program may depend on the executions of another. 
So it is not possible to reason about their executions separately. 
In the program $E$ above, for instance, 
the value of $S$ in program $A$ relies on the execution of $\Eemit\ S(1)$ of program $B$ in the same instance. 
Therefore, the assignment $x := x - S$ can only be executed after $\Eemit\ S(1)$. 
This actually makes separated reasoning of $A$ and $B$ impossible. 
More analysis can be found in the example in~\cite{Gesell12}.

Appendix~\ref{section:Other Details of Example Two} shows the deduction procedure based on Esterel program transitions for a property of program $E$ captured by modality $\la\cdot\ra$ in $\LDL$.   
\fi

\ifx
Following the program $E$ in Section~\ref{section:Example Two: A Synchronous Loop Program},
we prove the property below based on the symbolic executions of program $E$ according to its operational semantics: 
$$
P_2 \dddef \{x\mapsto v\Split S\mapsto \bot\} : (x > 0)\Rightarrow{} \{x\mapsto v\Split S\mapsto \bot\} : \la E\ra\true,
$$
which says that under any configuration (with $v$ a free variable), if $x > 0$, then $E$ can finally terminate (with satisfying $\true$). 
\fi

\begin{table}[tb]
	\noindent\makebox[\textwidth]{%
		\scalebox{0.9}{
			\begin{tabular}{l|l}
				\toprule
				\begin{tabular}{l}
					\begin{tikzpicture}[->,>=stealth', node distance=3cm]
						\node[draw=none] (txt2) {
							$
								\infer[^{(\la\alpha\ra)}]
								{\mbox{$\nu_2 : 1$}}
								{\infer[^{(\la\alpha\ra)}]
									{2}
									{\infer[^{(\la\alpha\ra)}]
										{3}
										{\infer[^{(\sigma\textit{Cut})}]
											{4}
											{
												\infer[^{(\textit{Wk} R)}]
												{5}
												{
													\infer[^{(\textit{Ter})}]
													{17}
													{}
												}
												&
												\infer[^{(\sigma\vee L)}]
												{6}
												{
													\infer[^{(\la\alpha\ra)}]
													{7}
													{
														\infer[^{(\la\alpha\ra )}]
														{12}
														{
															\infer[^{(\textit{LE})}]
															{13}
															{
																\infer[^{(\Sub)}]
																{14}
																{  
																	\infer[^{(\textit{LE})}]
																	{15}
																	{16}
																}
															}
														}
													}
													&
													\infer[^{(\la\alpha\ra)}]
													{8}
													{
														\infer[^{(\la\ter\ra)}]
														{9}
														{
															\infer[^{(\textit{Int})}]
															{10}
															{\infer[^{(\textit{Ter})}]
																{11}
																{}
															}
														}
													}
												}
											}
										}   
									}
								}
							$
						};
						

						\path
						;
						
						
						\draw[dotted,thick,red] ([xshift=-0.7cm, yshift=1.75cm]txt2.center) -- 
						([xshift=-1.5cm, yshift=1.75cm]txt2.center) --
						([xshift=-1.5cm, yshift=-1.75cm]txt2.center) --
						([xshift=-1.25cm, yshift=-1.75cm]txt2.center);
					\end{tikzpicture}
				\end{tabular}
				&
				\begin{tabular}{l}
					Definitions of other symbols:\\
					$A \dddef
					\Eloop\ (\Eemit\ S(0)\ ;\ x := x - S\ ;\ \Eif\ x = 0\ \Ethen\ \Eexit\ \Eend\ ;\ \Epause)\ \Eend$\\
					$B \dddef
					\Eloop\ (\Eemit\ S(1)\ ;\ \Epause)\ \Eend$\\
					$A' \dddef (\Eif\ x = 0\ \Ethen\ \Eexit\ \Eend\ ;\ \Epause)$\\
					$\sigma_1(v) \dddef \{x\mapsto v\Split S\mapsto \bot\}$\\
					$\sigma_2(v) \dddef \{x\mapsto v\Split S\mapsto 0\}$\\
					$\sigma_3(v) \dddef \{x\mapsto v\Split S\mapsto 1\}$\\
					$\sigma_4(v) \dddef \{x\mapsto v-1\Split S\mapsto 1\}$\\
					$\sigma_5(v) \dddef \{x\mapsto v-1\Split S\mapsto \bot\}$\\
				\end{tabular}
				
				\\
				\midrule
				\multicolumn{2}{l}{
					\begin{tabular}{l l l l}
						1: & $v > 0$ & $\Rightarrow$ & \ul{$\sigma_1(v) : \la \Etrap\ A\parallel B\ \Eend\ra\true$}
						\\
						2: & $v > 0$ & $\Rightarrow$ & \ul{$\sigma_2(v) : \la \Etrap \ ((x := x - S\ ; A')\ ;\ A)\parallel B)\ \Eend\ra\true$}
						\\
						3: & $v > 0$ & $\Rightarrow$ & \ul{$\sigma_3(v) : \la \Etrap \ ((x := x - S\ ; A')\ ;\ A)\parallel (\Epause\ ;\ B)\ \Eend\ra\true$}
						\\
						4: & $v > 0$ & $\Rightarrow$ & \ul{$\sigma_4(v) : \la \Etrap \ (A'\ ;\ A)\parallel (\Epause\ ;\ B)\ \Eend\ra\true$}
						\\
						5: &  $v > 0$ & $\Rightarrow$ & $\sigma_4(v) : \la \Etrap \ (A'\ ;\ A)\parallel (\Epause\ ;\ B)\ \Eend\ra\true, v -1 \neq 0\vee v - 1 = 0$
						\\
						17: &  $v > 0$ & $\Rightarrow$ & $v -1 \neq 0\vee v - 1 = 0$
						\\
						\midrule
						6: & $v > 0, x - 1 \neq 0\vee x - 1 = 0$ & $\Rightarrow$ & \ul{$\sigma_4(v) : \la \Etrap \ (A'\ ;\ A)\parallel (\Epause\ ;\ B)\ \Eend\ra\true$}
						\\
						7: & $v > 0, v - 1\neq 0$ & $\Rightarrow$ & \ul{$\sigma_4(v) : \la \Etrap \ (A'\ ;\ A)\parallel (\Epause\ ;\ B)\ \Eend\ra \true$}
						\\
						12:& $v > 0, v - 1\neq 0$ & $\Rightarrow$ & \ul{$\sigma_4(v) : \la \Etrap \ (\Epause\ ;\ A)\parallel (\Epause\ ;\ B)\ \Eend\ra \true$}
						\\
						13:& $v > 0, v - 1\neq 0$ & $\Rightarrow$ & \ul{$\sigma_5(v) : \la \Etrap \ A\parallel B\ \Eend\ra \true$}
						\\
						14: & $(v - 1) + 1 > 0, v - 1 \neq 0$ & $\Rightarrow$ & 
						\ul{$\sigma_5(v) : \la \Etrap \ A\parallel B\ \Eend\ra \true$}
						\\
						15: & $v + 1 > 0, v \neq 0$ & $\Rightarrow$ & \ul{$\sigma_1(v) : \la \Etrap \ A\parallel B\ \Eend\ra \true$}
						\\
						16:& $v > 0$ & $\Rightarrow$ & \ul{$\sigma_1(v) : \la \Etrap \ A\parallel B\ \Eend\ra \true$}
						\\
						\midrule
						8: & $v > 0, v - 1 = 0$ & $\Rightarrow$ & $\sigma_4(v) : \la \Etrap \ (A'\ ;\ A)\parallel (\Epause\ ;\ B)\ \Eend\ra \true$
						\\
						9: & $v > 0, v - 1 = 0$ & $\Rightarrow$ & $\sigma_5(v) : \la \ter\ra \true$
						\\
						10: & $v > 0, v - 1 = 0$ & $\Rightarrow$ & $\sigma_5(v) : \true$
						\\
						11: & $v > 0, v - 1 = 0$ & $\Rightarrow$ & $\true$
						\\
					\end{tabular}
				}
				\\
				\bottomrule
			\end{tabular}
		}
	}
	\caption{Derivations of Property $\nu_2$}
	\label{figure:The derivation of Example 2}
\end{table}

The derivations of $\nu_2$ is depicted in Table~\ref{figure:The derivation of Example 2}. 
We omit the sub-proof-procedures of all program transitions when applying rule $(\la\alpha\ra)$, which depends on the operational semantics of Esterel programs (cf.~\cite{Berry92}).

From node 2 to 3 is a progressive step, where we omit the derivations of  
the termination $(\Etrap \ ((x := x - S\ ; A')\ ;\ A)\parallel (\Epause\ ;\ B)\ \Eend, \sigma_3(v))\termi$, which depends on the operational semantics of Esterel (cf.~\cite{Berry92}). 
It is not hard to see that this program does terminate as the value of variable $x$ decreases by $1$ (by executing $x := x - S$) in each loop so that statement $\Eexit$ will finally be reached. 
From node 13 to 14 and node 15 to 16, rule
$$
\begin{aligned}
	\infer[^{(\textit{LE})}]
	{\Gamma, \phi\Rightarrow \Delta}
	{\Gamma, \phi'\Rightarrow\Delta},
\end{aligned}
\ \ \mbox{if $\phi\to \phi'\in \Fmla$ is valid} 
$$
is applied, which can be derived by the following derivations:
$$
\begin{aligned}
	\infer[^{(\textit{Cut})}]
	{\Gamma, \phi\Rightarrow \Delta}
	{
		\infer[^{(\textit{Wk} L)}]
		{\Gamma, \phi, \phi'\Rightarrow\Delta}
		{\Gamma, \phi'\Rightarrow \Delta}
		&
		\infer[^{(\textit{Ter})}]
		{\Gamma, \phi\Rightarrow \phi', \Delta}
		{}
	}
\end{aligned}
$$
From node 14 to 15, rule 
$$
\begin{aligned}
	\infer[^{(\Sub)}]
	{\Gamma[e/x]\Rightarrow\Delta[e/x]}
	{\Gamma\Rightarrow \Delta}  
\end{aligned}
$$
is applied, whose substitution function is $(\cdot)[e/x]$ as introduced in Section~\ref{section:Examples of Term Structures}. 
Observe that sequent 14 can be written as: 
$$
(v + 1 > 0)[v-1/v], (v \neq 0)[v-1/v] \Rightarrow 
(\sigma_1(v) : \la \Etrap \ A\parallel B\ \Eend\ra \true)[v-1/v]. 
$$

\ifx
From node 16, we can obtain node 17 by applying $(\sigma\textit{Cut})$ (to generate $v > 0$) and $(\sigma\textit{Wk}L)$. 
Intuitively, one can easily see that $v + 1 > 0\wedge v \neq 0$ is logical equivalent to $v > 0$ in the theory of integer numbers. 
From node 11, node 12 is obtained by applying rules $(\sigma\textit{Int} L)$ and $(\textit{Int})$. 
From sequent 15 to 16 rule $(\sigma \textit{Sub})$ is applied, 
since 15 can also be written as: 
$(v + 1 > 0)[v- 1/v], (v \neq 0)[v-1/v]\Rightarrow (\sigma_1(v) : \la \Etrap \ A\parallel B\ \Eend\ra \true)[v-1/v]$.

The whole proof relies on a safe extra rule 
$$
\begin{aligned}
	\infer[^{(\SExtra:\sigma\textit{Sub})}]
	{\Gamma[e/x]\Rightarrow\Delta[e/x]}
	{\Gamma\Rightarrow \Delta}  
\end{aligned}
, 
$$
whose soundness and safeness can be analyzed similarly as in \emph{While} programs (see Appendix~\ref{section:Formal Definitions of While Programs}), which we omit in this paper. 
From node 3 to 4 is a progressive step, where we omit the derivations of  
the termination $(\Etrap \ ((x := x - S\ ; A')\ ;\ A)\parallel (\Epause\ ;\ B)\ \Eend, \sigma_3(v))\termi$, which depends on the operational semantics of Esterel (cf.~\cite{Berry92}). 
It is not hard to see that this program does terminate as the value of variable $x$ decreases by $1$ in each loop so that statement $\Eexit$ will finally be reached. 
The derivation from node 6 to 8 can be obtained by applying 
rules $(\sigma\textit{Int}L)$ and $(\sigma\textit{Wk}L)$. 
From node 16, we can obtain node 17 by applying $(\sigma\textit{Cut})$ (to generate $v > 0$) and $(\sigma\textit{Wk}L)$. 
Intuitively, one can easily see that $v + 1 > 0\wedge v \neq 0$ is logical equivalent to $v > 0$ in the theory of integer numbers. 
From node 11, node 12 is obtained by applying rules $(\sigma\textit{Int} L)$ and $(\textit{Int})$. 
From sequent 15 to 16 rule $(\sigma \textit{Sub})$ is applied, 
since 15 can also be written as: 
$(v + 1 > 0)[v- 1/v], (v \neq 0)[v-1/v]\Rightarrow (\sigma_1(v) : \la \Etrap \ A\parallel B\ \Eend\ra \true)[v-1/v]$. 
\fi

Sequent 16 is a bud with sequent 1 as its companion. 
The whole preproof is cyclic as 
the only derivation path: $1,2,3,4,6,7,12,13,14,15,16,1,...$ has 
a progressive trace whose elements are underlined in Table~\ref{figure:The derivation of Example 2}. 

\ifx
is progressive as the derivation step from 3 to 4 is progressive (with $v - 1 < v$).  
So all infinite derivation traces are progressive. 

The derivation from sequent 2 to 3 comes from extra theory $\mcl{E}$, namely rule $(\mcl{E})$, which is based on the fact that $\sigma_3$ actually stores the same values for $n$ and $s$ as $\sigma_2$, because by replacing replacing $n$ with its stored value $v$ in the expression $((n+v+1)(n-v))/2$, we have  $((v+v+1)(v-v))/2 = 0$. 
So $\sigma_3$ actually maps $s$ to $0$. 
This observation, however, relies on the definition of interpretation $\app$ in the example. 
Configuration $\sigma_3(v)$ constructed in node 3 is crucial, as starting from it, we can find a bud node --- 11 ---
that is identical to node 4. 
The derivation from sequent 4 to \{5, 6\} provides a lemma: 
$\sigma_3(v) : (n > 0\vee n\le 0)$, which is obvious valid. 
From node 5, we can obtain 17 by applying rule $(\textit{Wk})$ and rule $(\textit{Int})$ in sequence. 
From node 15, node 16 can be obtained through $(\textit{Cut})$ and $(\sigma \textit{Wk} L)$, a derived rule corresponding to $(\textit{Wk})$ for the left-side derivations of sequents. 
Similarly, node 11 is obtained from 10 by applying $(\sigma \textit{Int} L)$ and $(\textit{Int})$. 
From sequent 14 to 15 rule $(\sigma \textit{Sub})$ is applied, 
since 14 can also be written as: 
$(v + 1\ge 0)[v - 1/v], (v + 1 > 0)[v - 1/v]\Rightarrow (\sigma_3 : [\textit{while}\
(n > 0)\
\textit{do}\
\alpha_1\
\textit{end}\ ]\phi_1)[v-1/v]$.

The derivations of $P_2$ is depicted in Table~\ref{figure:The derivation of Example 2}. 
Instances of rule $(\la\alpha\ra)$ rely on the program transitions listed as follows, which are based on the operational semantics of Esterel introduced in for example~\cite{Butucaru10}:
\begin{itemize}
	\item from node 1 to 2: $(\Etrap\ L\ \Eend, \sigma_1(v)\{v\})\trans (\Etrap\ (A_1\parallel B);\ L\ \Eend, \sigma_2(v)\{v\})$
	\item from node 2 to 3: $(\Etrap\ (A_1\parallel B);\ L\ \Eend, \sigma_2(v)\{v\})\trans (\Etrap\ (A_1\parallel \Epause);\ L\ \Eend, \sigma_3(v)\{v\})$
	\item from node 3 to 4: $(\Etrap\ (A_1\parallel \Epause);\ L\ \Eend, \sigma_3(v)\{v\})\trans (\Etrap\ (A'\parallel \Epause);\ L\ \Eend, \sigma_4(v)\{v-1\})$
	\item from node 5 to 10: $(\Etrap\ (A'\parallel \Epause);\ L\ \Eend, \sigma_4(v)\{v-1\})\trans (\Etrap\ (\Epause\parallel \Epause);\ L\ \Eend, \sigma_4(v)\{v-1\})$
	\item from node 6 to 7: $(\Etrap\ (A'\parallel \Epause);\ L\ \Eend, \sigma_5(v)\{v-1\})\trans (\ter, \sigma_4(v)\{v-1\})$
	\item from node 10 to 11: $(\Etrap\ (\Epause\parallel \Epause);\ L\ \Eend, \sigma_4(v)\{v-1\})\trans(\Etrap\ L\ \Eend, \sigma_5(v)\{v-1\})$
\end{itemize}

The termination factors $v$ in $E$ are natural numbers, with 
the ``less-than relation'' $<$ between natural numbers as the well-founded relation. 
From node 4, we get two branches from 5 and 6 by considering two conditions respectively: $x = 0$ and $x \neq 0$. 
Node 12 is obtained from node 11 by substituting variable $v$ with $v + 1$ based on the following observations in sequent 11: 
(1) $\sigma_5(v)$ is just $\sigma_1(v-1)$; 
(2) $\sigma_1(v-1) : (x > 0)$ is logically equivalent to $(\sigma_4(v) : (x \neq 0)) \wedge (\sigma_1(v) : (x > 0))$. 

Sequent 12 is a bud with 1 as its companion. 
The whole preproof is progressive as the derivation step from 3 to 4 is progressive (with $v - 1 < v$).  
So all infinite derivation traces are progressive. 
\fi

\end{document}